\documentclass{fundam}

\usepackage{amsmath,amssymb}
\usepackage{lineno}
\usepackage{stmaryrd}

\usepackage[ruled,vlined]{algorithm2e}
\usepackage{mathtools}
\usepackage{mathrsfs}
\usepackage{url}

\usepackage{lineno}
\usepackage{cleveref}

\usepackage{tikz}
\usetikzlibrary{shapes,snakes}
\usepackage{xcolor}

\usepackage{macros}

 \usepackage{hyperref}

\begin{document}

%%%%%%%  parameters to be filled in by copy-editor  %%%%%%%%%%

\setcounter{page}{63}
\publyear{24}
\papernumber{2167}
\volume{190}
\issue{2-4}

\finalVersionForARXIV
%%%\finalVersionForIOS

%%%%%%%%%%%%%%%%%%%%%%%%%%%%%%%%%%%%%%

\title{Waiting Nets: State Classes and Taxonomy}

\author{Lo\"ic H\'elou\"et\thanks{Address for correspondence:  University Rennes, Inria, CNRS, IRISA, France.}
 \\
University Rennes, Inria, CNRS, IRISA, France\\
 loic.helouet@inria.fr
 \and Pranay Agrawal\\
ENS-Paris-Saclay, France\\
pranay.agrawal@ens-paris-saclay.fr
}

\maketitle

\runninghead{L. H\'elou\"et and P. Agrawal}{Waiting Nets: State Classes and Taxonomy} %TODO optional, please use if title is longer than one line

\begin{abstract}
In time Petri nets (TPNs), time and control are tightly connected: time measurement for a transition starts only when all resources needed to fire it are available. Further, upper bounds on duration of enabledness can force transitions to fire (this is called {\em urgency}). For many systems, one wants to decouple control and time, i.e. start measuring time as soon as a part of the preset of a transition is filled, and fire it after some delay \underline{and} when all needed resources are available. This paper considers an extension of TPN called {\em waiting nets} that dissociates time measurement and control.
%Waiting Nets allows firing rules of the form "fire $t$ within $[\alpha, \beta]$ as soon as places $p_1,p_2$ are filled".
Their semantics allows time measurement to start with incomplete presets, and can ignore urgency when upper bounds of intervals are reached but all resources needed to fire are not yet available. Firing of a transition is then allowed as soon as missing resources are available. It is known that extending bounded TPNs with stopwatches leads to undecidability. Our extension is weaker, and we show how to compute a finite state class graph for bounded waiting nets, yielding decidability of reachability and coverability. We then compare expressiveness of waiting nets with that of other models w.r.t. timed language equivalence,
and show that they are strictly more expressive than TPNs.
\end{abstract}

\begin{keywords}
Time, Petri nets, State Classes.
\end{keywords}

\section{Introduction }
\label{sec:intoduction}

Time Petri nets (TPNs)~\cite{Merlin74} are an interesting model to specify cyber-physical systems. They allow for the specification of concurrent or sequential events (modeled as transitions occurrences), resources, time measurement, and urgency.
In TPNs, time constraints are modeled by attaching an interval $[\alpha_t,\beta_t]$ to every transition $t$. If $t$ has been enabled for at least $\alpha_t$ time units it {\em can} fire. If $t$ has been enabled for $\beta_t$ time units, it is {\em urgent}: time cannot elapse, and $t$ {\em must} either fire or be disabled. Urgency is an important feature of TPNs, as it allows for the modeling of strict deadlines, but gives them a huge expressive power. In their full generality, TPNs are Turing powerful. A consequence is that most properties that are decidable for Petri Nets~\cite{Esparza96} (coverability~\cite{Rackoff78}, reachability~\cite{Mayr81}, boundedness~\cite{Rackoff78}...) are undecidable for TPNs. Yet, for the class of bounded TPNs, reachability~\cite{Popova-Zeugmann91} and coverability are decidable. The decision procedure relies on a symbolic representation of states with {\em state classes} and then on the definition of abstract runs as paths in a so-called state class graph~\cite{BerthomieuD91,LimeR06}.

There are many variants of Petri nets with time. An example is {\em timed Petri nets} (TaPN)~\cite{Ramchandani74}, where tokens have an age, and time constraints are attached to arcs of the net. In TaPNs, a token whose age reaches the upper bound of constraints becomes useless. The semantics of TaPNs enjoys some monotonicity, and well-quasi-ordering techniques allow to solve coverability or boundedness problems~\cite{AN01,ReynierS09}. However, reachability of a given marking remains undecidable for TaPNs~\cite{Ruiz99}. Without any notion of urgency, TaPN cannot model delay expiration. In~\cite{AkshayGH16}, a model mixing TaPN and urgency is proposed, with decidable coverability, even for unbounded nets. We refer readers to~\cite{sofsem11} for a survey on TaPNs and their verification, and to~\cite{Bowden96} for a survey on different ways to introduce time in Petri nets.

Working with bounded models is enough for many cyber-physical systems. However, bounded TPNs suffer another drawback: time measurement and control are too tightly connected. In TPNs, time is measured by starting a new clock for every transition that becomes enabled. By doing so, measuring a duration for a transition $t$ starts only when all resources needed to fire $t$ are available. Hence, one cannot stop and restart a clock, nor start measuring  time while waiting for missing resources. To solve this problem,~\cite{BerthomieuLRV07} equips bounded TPNs with stopwatches. Nets are extended with read arcs, and the understanding of a read arc from a place $p$ to a transition $t$ is that when $p$ is filled, the clock attached to $t$ is frozen.
Extending bounded TPNs with stopwatches leads to undecidability of coverability, boundedness and reachability. This is not a surprise, as timed automata with stopwatches are already a highly undecidable model~\cite{CassezL00}. For similar reasons, time Petri nets with preemptable resources~\cite{BucciFSV04}, where time progress depends on the availability of resources cannot be formally verified.

This paper considers {\em waiting nets}, a new extension of TPN that decouples time measurement and control. Waiting nets distinguish between enabling of a transition and enabling of its firing, which allows rules of the form "start measuring time for $t$ as soon as $p$ is filled, and fire $t$ within $[\alpha, \beta]$ time units when $p$ and $q$ are filled". Syntactically, this model is a simple extension of TPNs: the model is built by adding a particular type of {\em control place} to TPNs. Waiting nets allow clocks of enabled transitions to reach their upper bounds, and wait for missing control to fire.
A first contribution of this paper is a formal definition of the semantics of this new model. An immediate result is that waiting nets are more expressive than TPNs, as TPNs are a simple syntactic restriction of waiting nets.
A former attempt called Timed Petri nets with Resets (TPNR) distinguishes some delayable transitions that can fire later than their upper bounds~\cite{ParrotBR21}. For bounded TPNR, reachability and TCTL model checking are decidable. However, delayable transitions are never urgent, and once delayed can only fire during a maximal step with another transition fired on time. Further, delayable transitions start measuring time as soon as their preset is filled, and hence do not allow decoupling of time and control as in waiting nets.

As a second contribution of this paper, we define state class graphs for waiting nets. We build on the work of~\cite{BerthomieuD91} to define domains, i.e. symbolic representations for the values of clocks that measure time elapsed since enabling of transitions. We then define a successor relation via transition firing among domains. This allows for the construction of a state class graph depicting reachable markings and domains of a waiting net. We show that the state class graphs of bounded waiting nets are finite, yielding decidability of reachability and coverability (which are PSPACE-complete). This is a particularly interesting result, as these properties are undecidable for {\em stopwatch Petri nets},  {\em even in the bounded case}. The table~\ref{complexity_table} summarizes known decidability results for reachability, coverability and boundedness problems for time variants of Petri nets, including the new results for waiting nets proved in this paper. Another interesting result shown in this paper is that a state class may have more that one successor via a transition. This shows that state class graphs of waiting nets are not deterministic, and that TPNs and waiting nets are hence different models.

\begin{table}[h]
\begin{center}
\caption{Decidability and complexity results for time(d) variants of Petri nets.}
\label{complexity_table}%\vspace*{-1mm}
\scalebox{0.8}{
\begin{tabular}{|l|c|c|c|}
\hline
                         & Reachability & coverability & Boundedness \\
\hline
%(untimed) Petri Nets     &      Decidable~\cite{Mayr81}        &  Decidable~\cite{Rackoff78}           &  Decidable~\cite{Rackoff78}           \\
\hline
Time Petri Nets          &  Undecidable~\cite{Jone77}            &   Undecidable~\cite{Jone77}           &              Undecidable~\cite{Jone77} \\
(bounded)                & Decidable             &  Decidable            &  --            \\
\hline
Timed Petri nets          & Undecidable~\cite{Ruiz99}             &  Decidable~\cite{Frutos-EscrigRA00,AN01} &     Decidable~\cite{Frutos-EscrigRA00}         \\
(bounded)         &     Decidable         &     Decidable         &         --     \\
\hline
Restricted Urgency         &    Undecidable~\cite{AkshayGH16}          &  Decidable~\cite{AkshayGH16}            & Decidable~\cite{AkshayGH16}             \\
(bounded)         &          Decidable    &  Decidable            &       --       \\
\hline
Stopwatch Petri nets          &  Undecidable~\cite{BerthomieuLRV07}  &   Undecidable~\cite{BerthomieuLRV07} &   Undecidable~\cite{BerthomieuLRV07}           \\
(bounded)         &    Undecidable~\cite{BerthomieuLRV07}         &     Undecidable~\cite{BerthomieuLRV07}         &      --        \\
\hline
TPNR          &  Undecidable~\cite{ParrotBR21}  &   Undecidable~\cite{ParrotBR21} &   Undecidable~\cite{ParrotBR21}           \\
(bounded)         &    Decidable~\cite{ParrotBR21}         &    Decidable~\cite{ParrotBR21}         &      --        \\
\hline
Waiting Nets & Undecidable (Rmk.~\ref{remark_inheritance})& Undecidable (Rmk.~\ref{remark_inheritance})& Undecidable (Rmk.~\ref{remark_inheritance})\\
(bounded)   & PSPACE-Complete (Thm.~\ref{Cor_reachability_coverability}) & PSPACE-Complete (Thm.~\ref{Cor_reachability_coverability}) & --\\
\hline
\end{tabular}
}
\end{center}
\end{table}

Our last contribution is a study of the expressiveness of waiting nets w.r.t. timed language equivalence. We compare the timed languages generated by timed automata, TPNs, waiting nets. We consider variants of these models with injective or non-injective labelings, and with $\epsilon-$transitions. Interestingly, the expressiveness of bounded waiting nets lays between that of bounded TPNs and timed automata.

This paper is an extended version of an article originally published in~\cite{HelouetA22}.
The rest of the paper is organized as follows: we recall the formal background needed later in the paper in Section~\ref{sec:preliminaries}. We recall in particular definitions of clock constraints, timed automata and timed languages.
We define the syntax and semantics of waiting nets in Section~\ref{sec:waiting}.
Section~\ref{sec:reachability} shows a state class graph construction for waiting nets and proves finiteness of the set of domains of waiting nets. Finiteness of the obtained state class graph is then used to show PSPACE completeness of reachability and coverability for bounded waiting nets. Section~\ref{sec:expressiveness} compares the expressive power of waiting nets with that of other models and with other timed variants of Petri nets before conclusion. For readability, several technical proofs of Section~\ref{sec:reachability} are only sketched, but are provided in Appendix.

\section{Preliminaries}
\label{sec:preliminaries}

%Before defining our models, let us introduce the notations needed to define firing domains of transitions in a net, timed languages, that are used later to compare expressiveness of models, and timed automata, the standard model to recognize timed languages~\cite{AlurD94}.
We denote by $\mathbb R^{\geq 0}$ the set of non-negative real values, and by $\mathbb Q$ the set of rational numbers.
A { \em rational interval} $[ \alpha,\beta ]$ is the set of values between a lower bound $\alpha\in \mathbb Q$ and an upper bound $\beta\in \mathbb Q$. We also consider intervals without upper bounds of the form $[\alpha,\infty)$, to define
%sets of
values that are greater than or equal to $\alpha$.

\medskip
A {\em clock} is a variable $x$ taking values in $\mathbb R^{\geq 0}$.
A variable $x_t$ will be used to measure the time elapsed since transition $t$ of a net was last newly enabled.
Let $X$ be a set of clocks. A {\em valuation}
for $X$ is a map $v:X\to \mathbb R^{\geq 0}$ that associates a positive or zero real value $v(x)$ to every variable $x\in X$.
Intervals alone are not sufficient to define the domains of clock valuations met with TPNs and timed automata.
An {\em atomic constraint} on $X$ is an inequality of the form $a \leq x$, $x\leq b$, $a\leq x-y$ or $x-y \leq b$ where $a,b\in \mathbb Q$ and $x,y\in X$. A {\em constraint} is a conjunction of atomic constraints.
We denote by $Cons(X)$ the set of constraints over clocks in $X$.
We will say that a valuation $v$ satisfies a constraint $\phi$, and write $v\models \phi$ iff replacing $x$ by $v(x)$ in $\phi$ yields a tautology.
A constraint $\phi$ is {\em satisfiable} iff there exists a valuation $v$ for $X$ such that $v\models \phi$. Constraints over real-valued variables can be efficiently encoded with Difference Bound Matrices (DBMs) and their satisfiability checked
%in PTIME~\ref{Kachian} and even
in $O(n^3)$~\cite{Dill89}. For completeness, we give a formal definition of DBMs in Appendix~\ref{DBM}. We also show in Appendix~\ref{appendix_Floyd_Warshall} that checking satisfiability of a constraint amounts to detecting negative cycles in a weighted graph, which can be checked using a variant of the Floyd-Warshall algorithm.
The {\em domain} specified by a constraint $\phi$ is the (possibly infinite) set of valuations that satisfy $\phi$.

Given an alphabet $\Sigma$, a {\em timed word} is an element of $(\Sigma\times \mathbb R^+)^*$ of the form $w=(\sigma_1,d_1)(\sigma_2,d_2)\dots $ such that $d_i\leq d_{i+1}$. A timed language is a set of timed words. Timed automata~\cite{AlurD94} are frequently used to recognize timed languages. In the last section of this paper, we compare expressiveness of several variants of time(d) Petri nets and timed automata. For completeness, we recall the definition and semantics of timed automata.

\begin{definition}[Timed Automaton]
A {\em Timed Automaton} $\ta$ is a tuple $\ta=(L, \ell_0, X,\Sigma, Inv, E, F)$, where $L$ is a set of locations, $\ell_0\in L$ is the initial location, $X$ is a set of clocks, $\Sigma$ is an alphabet, $Inv : L \to Cons(X)$ is a map associating an invariant to every location. The set of states $F\subseteq L$ is a set of final locations, and $E$ is a set of edges. Every edge is of the form $(\ell,g,\sigma,R, \ell') \in L\times Cons(X) \times \Sigma \times 2^X\times L$.
\end{definition}

Intuitively, the semantics of a timed automaton allows elapsing time in a location $\ell$ (in which case clocks valuations grow uniformly) if this does not result in a valuation that violates $Inv(\ell)$, or firing a discrete transition $(\ell,g,\sigma,R, \ell')$ from location $\ell$ with clock valuation $v$ if $v$ satisfies guard $g$, and the valuation $v'$ obtained by resetting all clocks in $R$ to $0$ satisfies $Inv(\ell')$. One can notice that invariants can prevent firing a transition. Every run of a timed automaton starts from $(\ell_0,v_0)$, where $v_0$ is the valuation that assigns value $0$ to every clock in $X$.

\medskip
Let $\ta=(L, \ell_0, X,\Sigma, Inv, E, F)$ be a timed automaton.
A configuration of $\ta$ is a pair $(\ell,v)$ where $\ell\in L$ is a location, and $v$ a valuation of clocks in $X$.
Let $v_{R:=0}$ denote a valuation $v'$ such that $v'(x)=0$ if $x\in R$ and $v'(x)=v(x)$ otherwise.
A {\em discrete move} via transition $e=(\ell,g_e,\sigma_e,R, \ell')$ is allowed from $(\ell,v)$ iff $v\models g_e$ (the guard of the transition is satisfied) and $v_{R:=0}\models Inv(\ell')$. We will denote such moves by $(\ell,v)\overset{e}\longrightarrow (\ell',v')$.
Let $d\in \mathbb R^+$, and let $v+d$ denote a valuation such that $v+d(x) = v(x)+d$ for every $x\in X$.
A {\em timed move} of $d$ time units is allowed if $v+d'\models Inv(\ell)$ for every real value $d'\leq d$. We will denote such moves by $(\ell,v)\overset{d}\longrightarrow (\ell',v+d)$.

A {\em run} of $\ta$ is a sequence of discrete and timed moves that starts from configuration  $(\ell_0,v_0)$, where $v_0$ is the valuation that associates value $0$ to every clock in $X$. A run is accepting if it ends in a final location.
Without loss of generality, we will assume that runs of timed automata alternate timed and discrete moves (possibly of duration $0$). A run is hence a sequence of the form $\rho = (\ell_0,v_0) \overset{d_0}\longrightarrow (\ell,v_0+d_0) \overset{e_0}\longrightarrow (\ell_1,v_1)\dots$, where each $d_i\in \mathbb R$ is a duration, and each $e_i$ is a transition. The {\em timed word} associated with run $\rho$ is the word $w_\rho=(\sigma_{e_0},dt_0) \dots (\sigma_{e_i},dt_i)\dots$ where $dt_i =\sum_{k=0}^{k=i-1} d_k$. The {\em (timed) language} recognized by $\ta$ is denoted $\lang(\ta)$, and is the set of timed words associated with accepting runs of $\ta$.
The definition of timed automata can be easily extended to include silent transitions. Timed automata then contain transitions of the form $e=(\ell, g_e,\epsilon,R_e,\ell')$ labeled by an unobservable letter denoted $\epsilon$, and are hence automata over an alphabet $\Sigma \cup \epsilon$. In this case, the word recognized along a run $\rho$ is the projection of $w_\rho$ on pairs $(\sigma_{e_i},dt_i)$ such that $\sigma_{e_i}\neq \epsilon$.

In the rest of the paper, we will denote by $TA$ the class of timed automata. We will be in particular interested by the subclass $TA(\leq,\geq)$ in which guards are conjunctions of atomic constraints of the form $x \geq c$ and invariants are conjunctions of atomic constraints of the form $x \leq c$. Several translations from TPNs to TAs have been proposed (see for instance~\cite{CassezR06,LimeR06}). In particular, the solution of~\cite{LimeR06} uses the state class graph of a TPN to build a time-bisimilar timed automaton in class $TA(\leq,\geq)$. This shows that one needs not the whole expressive power of timed automata to encode timed languages recognized by TPNs.

\section{Waiting nets}
\label{sec:waiting}

TPNs are a powerful model: they can be used to encode a two-counter machine (see for instance~\cite{ReynierS09} for an encoding of counter machines), and can hence simulate the semantics of many other formal models. A counterpart to this expressiveness is that most problems (reachability, coverability, boundedness, verification of temporal logics...) are undecidable.
Decidability is easily recovered when considering the class of bounded TPNs. Indeed, for bounded TPNs, one can compute a finite symbolic model called a state class graph~\cite{BerthomieuD91}, in which timing information is symbolically represented by firing domains. For many applications, working with bounded resources is sufficient. However, TPNs do not distinguish between places that represent control (the "state" of a system), and those that represent resources: transitions are enabled when \underline{all} places in their preset are filled. A consequence is that one cannot measure time spent in a control state, when some resources are missing.

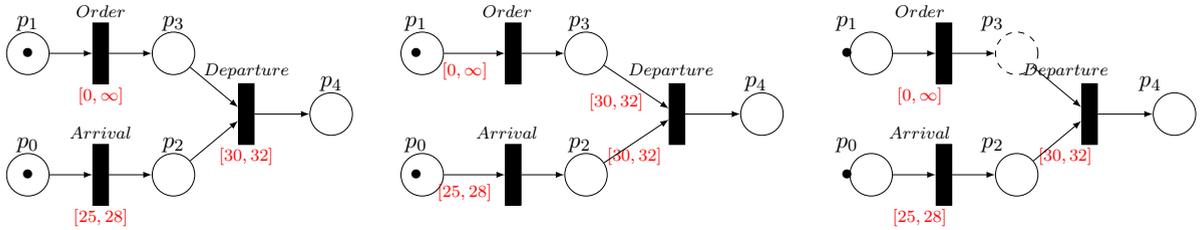
\begin{figure}[h]
\scalebox{0.8}{
\begin{tikzpicture}
%%%%%%%%%%%%%%%%%%%%%%%%%%%%%%%%%%%%%%%%%%%%%%%%%%%%%%%%%%%%%%%%%%%%%%%%
%
\begin{scope}
\node [draw,circle,minimum size=0.7cm] (p0) at (0.3,0) {$\bullet$};
\node [] (lp0) at (0.3,0.5) {$p_0$};

\node [draw,circle,minimum size=0.7cm] (p1) at (0.3,2) {$\bullet$};
\node [] (lp1) at (0.3,2.5) {$p_1$};

\node [draw,circle,minimum size=0.7cm] (p2) at (2.7,0) {};
\node [] (lp2) at (2.7,0.5) {$p_2$};

\node [draw,circle,minimum size=0.7cm] (p3) at (2.7,2) {};
\node [] (lp3) at (2.7,2.5) {$p_3$};

\node [draw,circle,minimum size=0.7cm] (p4) at (5.3,1) {};
\node [] (lp4) at (5.3,1.5) {$p_4$};

\node [draw,rectangle,fill=black,minimum height=1cm] (t0) at (1.5,0){};
\node [] (lt0) at (1.5,0.7) {\scriptsize$Arrival$};
\node [color=red] (it0) at (1.5,-0.7) {\scriptsize$[25,28]$};

\node [draw,rectangle,fill=black,minimum height=1cm] (t1) at (1.5,2){};
\node [] (lt1) at (1.5,2.7) {\scriptsize$Order$};
\node [color=red] (it1) at (1.5,1.3) {\scriptsize$[0,\infty]$};

\node [draw,rectangle,fill=black,minimum height=1cm] (t2) at (3.9,1){};
\node [] (lt2) at (3.9,1.7) {\scriptsize$Departure$};
\node [color=red] (it1) at (3.9,0.3) {\scriptsize$[30,32]$};

\draw [-latex] (p0) edge  (t0){};
\draw [-latex] (p1) edge  (t1){};
\draw [-latex] (t0) edge  (p2){};
\draw [-latex] (t1) edge  (p3){};
\draw [-latex] (p2) edge  (t2){};
\draw [-latex] (p3) edge  (t2){};

\draw [-latex] (t2) edge  (p4){};
%\node [] (laba) at (1.5,-1.5) {$a)$};
\end{scope}\hspace*{-1mm}
\begin{scope}[xshift=6.8cm]
\node [draw,circle,minimum size=0.7cm] (p0) at (0,0) {$\bullet$};
\node [] (lp0) at (0,0.5) {$p_0$};

\node [draw,circle,minimum size=0.7cm] (p1) at (0,2) {$\bullet$};
\node [] (lp1) at (0,2.5) {$p_1$};

\node [draw,circle,minimum size=0.7cm] (p2) at (2.7,0) {};
\node [] (lp2) at (2.7,0.5) {$p_2$};

\node [draw,circle,minimum size=0.7cm] (p3) at (2.7,2) {};
\node [] (lp3) at (2.7,2.5) {$p_3$};

\node [draw,circle,minimum size=0.7cm] (p4) at (5.6,1) {};
\node [] (lp4) at (5.6,1.5) {$p_4$};

\node [draw,rectangle,fill=black,minimum height=1cm] (t0) at (1.5,0){};
\node [] (lt0) at (1.5,0.7) {\scriptsize$Arrival$};
\node [color=red] (it0) at (0.8,-0.3) {\scriptsize$[25,28]$};

\node [draw,rectangle,fill=black,minimum height=1cm] (t1) at (1.5,2){};
\node [] (lt1) at (1.5,2.7) {\scriptsize$Order$};
\node [color=red] (it1) at (0.8,1.7) {\scriptsize$[0,\infty]$};

\node [draw,rectangle,fill=black,minimum height=1cm] (t2) at (4.2,1){};
\node [] (lt2) at (4.2,1.7) {\scriptsize$Departure$};
\node [color=red] (it2) at (3.3,1.2) {\scriptsize$[30,32]$};
\node [color=red] (it3) at (3.6,0.3) {\scriptsize$[30,32]$};

\draw [-latex] (p0) edge  (t0){};
\draw [-latex] (p1) edge  (t1){};
\draw [-latex] (t0) edge  (p2){};
\draw [-latex] (t1) edge  (p3){};
\draw [-latex] (p2) edge  (t2){};
\draw [-latex] (p3) edge  (t2){};

\draw [-latex] (t2) edge  (p4){};
%\node [] (laba) at (1.5,-1.5) {$a)$};
\end{scope}\hspace*{-3mm}
\begin{scope}[xshift=13.9cm]

\node [draw,circle,minimum size=0.7cm] (p0) at (0.3,0) {$\bullet$};
\node [] (lp0) at (0.3,0.5) {$p_0$};

\node [draw,circle,minimum size=0.7cm] (p1) at (0.3,2) {$\bullet$};
\node [] (lp1) at (0.3,2.5) {$p_1$};

\node [draw,circle,minimum size=0.7cm] (p2) at (2.7,0) {};
\node [] (lp2) at (2.7,0.5) {$p_2$};

\node [draw,dashed, circle,minimum size=0.7cm] (p3) at (2.7,2) {};
\node [] (lp3) at (2.7,2.5) {$p_3$};

\node [draw,circle,minimum size=0.7cm] (p4) at (5.3,1) {};
\node [] (lp4) at (5.3,1.5) {$p_4$};

\node [draw,rectangle,fill=black,minimum height=1cm] (t0) at (1.5,0){};
\node [] (lt0) at (1.5,0.7) {\scriptsize$Arrival$};
\node [color=red] (it0) at (1.5,-0.7) {\scriptsize$[25,28]$};

\node [draw,rectangle,fill=black,minimum height=1cm] (t1) at (1.5,2){};
\node [] (lt1) at (1.5,2.7) {\scriptsize$Order$};
\node [color=red] (it1) at (1.5,1.3) {\scriptsize$[0,\infty]$};

\node [draw,rectangle,fill=black,minimum height=1cm] (t2) at (3.9,1){};
\node [] (lt2) at (3.9,1.7) {\scriptsize$Departure$};
\node [color=red] (it1) at (3.9,0.3) {\scriptsize$[30,32]$};

\draw [-latex] (p0) edge  (t0){};
\draw [-latex] (p1) edge  (t1){};
\draw [-latex] (t0) edge  (p2){};
\draw [-latex] (t1) edge  (p3){};
\draw [-latex] (p2) edge  (t2){};
\draw [-latex] (p3) edge  (t2){};

\draw [-latex] (t2) edge  (p4){};
%\node [] (laba) at (1.5,-1.5) {$a)$};

\end{scope}
\end{tikzpicture}
}\vspace*{-3mm}
\caption{A TPN $a)$ a Timed Petri net $b)$, and a Waiting net $c)$}
\label{Example_Train}\vspace*{-1mm}
\end{figure}

Consider the examples of Figure~\ref{Example_Train}, that are designed to represent an arrival of a train followed by a departure. The arrival in a station is modeled by transition $Arrival$, that should occur between $25$ and $28$ minutes after the beginning of a run. The station is modeled by place $p_2$, and the departure of the train by transition $Departure$. A train can leave a station only if a departure order has been sent, which is modeled by transition $Order$. Assume that one wants to implement a scenario of the form "the train leaves the station between 30 and 32 minutes after its arrival if it has received a departure order".

\medskip
Consider the TPN of Figure~\ref{Example_Train}-a). The time constraint attached to $Departure$ is an interval of the form $[30,32]$, which is the only way to impose a precise timing in TPNs. This TPN does not implement our scenario, but rather behaviors in which the train leaves the station between 30 and 32 minutes after the instant when it is in station {\bf and} a departure order is received. This means that a train may spend $32 + d$ minutes in station, if the order is released $d$ minutes after its arrival.
Now, consider the second net in Figure~\ref{Example_Train}-b). It is a timed Petri net, i.e. a variant of nets where tokens are given an age (the durations that have elapsed since they were created in a place) and flow relations constrain the age of consumed tokens. In a nutshell, the semantics of timed Petri nets allow consumption of tokens which age satisfy the timing constraints of flow relations. Tokens that become too old to satisfy a constraint are never consumed and can be ignored. In the net of Figure~\ref{Example_Train}-b), a token can move from place $p_1$ to place $p_3$ at any time (i.e. a departure order can be issued at any instant), and a second token can move from place $p_0$ to place $p_2$ when its age is between $25$ and $28$ minutes, representing arrival of a train within a delay ranging from 25 to 28 minutes.  Following the semantics of timed Petri nets, one can simply "forget" the token in place $p_2$ if the delay to produce a token in place $p_3$ exceeds $28+32$. In other words, in this model, one can forget a train arrived in station if the departure order arrives more than 32 minutes after its $Arrival$. This is not the scenario we want to model. Further, there is no timed Petri net that can specify this example, because timed Petri nets do not have a notion of urgency, and hence cannot enforce occurrence of an action after a predetermined delay (in our case a departure of a train at latest after 32 minutes).

\medskip
We propose an extension of TPNs called {\em waiting nets} (WTPNs for short), that decouples control and resources during time measurement. We consider two types of places: {\em standard} places, and {\em control} places, with the following functions:
Time measurement for a transition $t$ starts as soon as $t$ has enough tokens in the standard places of its preset. Then, $t$ can fire if its clock value lays in its timing interval, {\bf and} if it has enough tokens in the control places of its preset.
To illustrate this extension, consider the example of Figure~\ref{Example_Train}-c). This example is similar to the net of Figure~\ref{Example_Train}-a), but with $p_3$ replaced by a control place (symbolized by a dashed circle).
Let us detail the changes brought by control places.
Let transition $Arrival$ fire at a date $d$ between $25$ and $28$ minutes, filling place $p_2$. This immediately enables transition $Departure$, that is forced to fire at latest $32$ time units later if transition $Order$ is fired at a date smaller than $d+32$, or exactly at the date of firing of transition $Order$ if this transition is fired more than $32$ minutes after date $d$. This net models a situation where a train can arrive at date $d \in [25, 28]$ minutes, and leave as soon as at least 30 minutes have elapsed since arrival, and a departure order was given. If the departure order is received at a date smaller than $d+32$, then the train can leave at latest at date $d+32$. If the departure order is given at a date $d'>d+32$, then the train leaves at date $d'$, i.e. when the departure order is given. This is exactly the scenario described above.
We can now define waiting nets formally as follows :

\begin{definition}
\label{def:WTPN}
A {\em waiting net} is a tuple $\mathcal{W} \!=\!\! \big( P,C, T, {^\bullet}\!(), ()^\bullet\! , \alpha, \beta,\lambda, (M_0.N_0) \big)$, where
  \begin{itemize}
      \item $P$ is a finite set of standard places, $C$ is finite set of control places, such that $P \cup C \neq \emptyset$ and $P \cap C = \emptyset$. A {\em marking} $M.N$ is a pair of maps $M: P\to \mathbb N$, $N:C\to \mathbb N$ that associate an integral number of tokens respectively to standard and control places.

      \item $T$ is a finite set of transitions. Every $t\in T$ has a label $\lab(t)\in \Sigma \cup\epsilon$,
      \item $\preset{} \in  (\mathbb{N}^{P \cup C})^T$ is the backward incidence function,
			$\postset{} \in  (\mathbb{N}^{P \cup C})^T$ is the forward incidence function,
			%which also associates a marking to every transition.
      \item $(M_0.N_0) \in \mathbb{N}^{P \cup C}$ is the initial marking of the net,
      \item $\alpha: T\to \mathbb{Q}^+$ and $\beta:  T \to \mathbb{Q}^+ \cup {\infty}$ are functions giving for each transition respectively its earliest and latest firing times. For every transition $t$, we have  $\alpha(t) \leq \beta(t)$.\smallskip
  \end{itemize}
\end{definition}

The labeling map $\lab$ is {\em injective} if $\lab(t)=\lab(t') \Leftrightarrow t=t'$, and may allow labeling transitions with a silent action label denoted $\epsilon$. To differentiate standard and control places in the preset of a transition, we will denote by $\spreset{t}$ the restriction of $\preset{t}$ to standard places, and by $\cpreset{}$ the restriction of $\preset{}$ to control places. We will write $M(p) = k$ (resp. $N(c)=k$) to denote the fact that standard place $p\in P$ (resp. control place $c\in C$) contains $k$ tokens.
Given two markings $M.N$ and $M'.N'$ we will say that $M.N$ is greater than $M'.N'$ and write $M.N \geq M'.N'$ iff
$\forall p\in P, M(p)\geq M'(p)$ and $\forall c\in C, N(c) \geq N'(C)$.

\begin{figure}[!ht]
\vspace*{-1mm}
\begin{center}
\scalebox{0.82}{
\begin{tikzpicture}
%%%%%%%%%%%%%%%%%%%%%%%%%%%%%%%%%%%%%%%%%%%%%%%%%%%%%%%%%%%%%%%%%%%%%%%%
\begin{scope}
% Stating place
\node [draw,circle,minimum size=0.7cm] (p0) at (0,-0.5) {$\bullet$};
\node [] (lp0) at (-0.7,-0.5) {$p_0$};

\node [draw,rectangle,fill=black,minimum width=1cm] (t0) at (0,-1.5){};
\node [] (lt0) at (-0.8,-1.5) {\scriptsize$Ad$};
\node [color=red] (it0) at (0.9,-1.5) {\scriptsize$[0,\infty]$};

\node [draw,circle,minimum size=0.7cm] (p1) at (-1,-3) {};
\node [] (lp1) at (-1.6,-3) {$p_1$};

\node [draw,circle,minimum size=0.7cm] (p2) at (3,-2) {};
\node [] (lp2) at (3.6,-1.8) {$p_2$};

\draw [-latex] (p0) edge (t0){};
\draw [-latex] (t0) edge (p1){};
\draw [-latex, bend right] (t0) edge (p2){};

\node [draw,rectangle,fill=black,minimum width=1cm] (t1) at (3,-3.2){};
\node [] (lt0) at (2.2,-3.2) {\scriptsize$Cp$};
\node [color=red] (it0) at (3.8,-3.2) {\scriptsize$[1,4]$};

\node [draw,dashed, circle,minimum size=0.7cm] (p3) at (1,-3) {};
\node [] (lp3) at (1.6,-3) {$p_3$};

\draw [-latex] (p2) edge (t1){};
\draw [-latex, dashed, bend left] (t1) edge (p3){};

\node [draw,rectangle,fill=black,minimum width=1cm] (t2) at (0,-4){};
\node [] (lt0) at (-0.7,-4) {\scriptsize$So$};
\node [color=red] (it0) at (0.9,-4) {\scriptsize$[0,3]$};

\node [draw,rectangle,fill=black,minimum width=1cm] (t3) at (-2.5,-4){};
\node [] (lt0) at (-3.3,-4) {\scriptsize$No$};
\node [color=red] (it0) at (-1.6,-4) {\scriptsize$[0,8]$};

\draw [-latex] (p1) edge (t3){};
\draw [-latex] (p1) edge (t2){};

\draw [-latex, dashed] (p3) edge (t2){};

\node [draw,circle,minimum size=0.7cm] (p4) at (0,-5) {};
\node [] (lp2) at (-0.6,-5) {$p_4$};

\node [draw,circle,minimum size=0.7cm] (p5) at (-2.5,-5) {};
\node [] (lp2) at (-3.1,-5) {$p_5$};

\draw [-latex] (t2) edge (p4){};
\draw [-latex] (t3) edge (p5){};

%\node [] (lbgen) at (0,-5.7) {$a)$};

\end{scope}

%%%%%%%%%%%%%%%%%%%%%%%%%%%%%%%%%%%%%%%%%%%%%%%%%%%%%%%%%%
\begin{scope}[xshift = 8cm]
% Stating place
\node [draw,circle,minimum size=0.7cm] (p0) at (0,-0.5) {$\bullet$};
\node [] (lp0) at (-0.8,-0.5) {$p_0$};

\node [draw,rectangle,fill=black,minimum width=1cm] (t0) at (0,-1.5){};
\node [] (lt0) at (-0.8,-1.5) {\scriptsize$Ad$};
\node [color=red] (it0) at (0.9,-1.5) {\scriptsize$[0,\infty]$};

\node [draw,circle,minimum size=0.7cm] (p1) at (-1,-3) {};
\node [] (lp1) at (-1.6,-3) {$p_1$};

\node [draw,circle,minimum size=0.7cm] (p2) at (3,-2) {};
\node [] (lp2) at (3.6,-1.8) {$p_2$};

\draw [-latex] (p0) edge (t0){};
\draw [-latex] (t0) edge (p1){};
\draw [-latex, bend right] (t0) edge (p2){};

\node [draw,rectangle,fill=black,minimum width=1cm] (t1) at (3,-3.2){};
\node [] (lt0) at (2.2,-3.2) {\scriptsize$Cp$};
\node [color=red] (it0) at (3.8,-3.2) {\scriptsize$[1,4]$};

\node [draw,dashed, circle,minimum size=0.7cm] (p3) at (1,-3) {};
\node [] (lp3) at (1.6,-3) {$p_3$};

\node [draw, circle,minimum size=0.7cm] (p6) at (0,-3) {};
\node [] (lp6) at (0,-2.4) {$p_6$};

\draw [-latex] (p2) edge (t1){};
\draw [-latex, dashed, bend left] (t1) edge (p3){};

\node [draw,rectangle,fill=black,minimum width=1cm] (t2) at (-0.3,-4){};
\node [] (lt0) at (-1.1,-4) {\scriptsize$So$};
\node [color=red] (it0) at (0.6,-4) {\scriptsize$[0,3]$};

\node [draw,rectangle,fill=black,minimum width=1cm] (t3) at (-2.5,-4){};
\node [] (lt0) at (-3.3,-4) {\scriptsize$No$};
\node [color=red] (it0) at (-1.6,-4) {\scriptsize$[0,8]$};

\node [draw,rectangle,fill=black,minimum width=1cm] (t4) at (2,-4){};
\node [] (lt0) at (1.3,-4) {\scriptsize$To$};
\node [color=red] (it0) at (2.9,-4) {\scriptsize$[3,3]$};

\draw [-latex] (p1) edge (t3){};
\draw [-latex] (p1) edge (t2){};
\draw [-latex] (p6) edge (t2){};
\draw [-latex] (p6) edge (t4){};
\draw [-latex,bend left] (t0) edge (p6){};

\draw [-latex, dashed] (p3) edge (t2){};

\node [draw,circle,minimum size=0.7cm] (p4) at (0,-5) {};
\node [] (lp2) at (-0.6,-5) {$p_4$};

\node [draw,circle,minimum size=0.7cm] (p5) at (-2.5,-5) {};
\node [] (lp2) at (-3.1,-5) {$p_5$};

\node [draw,circle,minimum size=0.7cm] (p7) at (2,-5) {};
\node [] (lp2) at (1.4,-5) {$p_7$};

\draw [-latex] (t2) edge (p4){};
\draw [-latex] (t3) edge (p5){};
\draw [-latex] (t4) edge (p7){};

%\node [] (lbgen) at (0,-5.7) {$b)$};

\end{scope}
\end{tikzpicture}
}\vspace*{-1mm}
\caption{$a)$ Decoupled time and control in a waiting net. $b)$ ... with a timeout transition.}
\label{Example_wn}
\end{center}\vspace*{-7mm}
\end{figure}
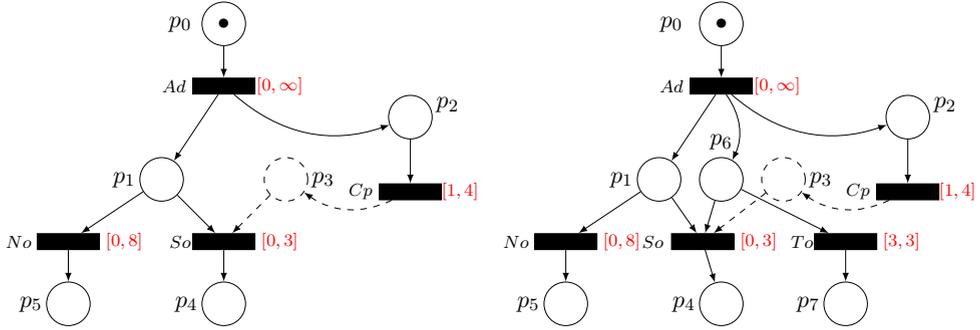

Figure~\ref{Example_wn}-a) is a waiting net modelling an online sale offer, with limited duration. Control places are represented with dashed lines. A client receives an ad, and can then buy a product up to 8 days after reception of the offer, or wait to receive a coupon offered to frequent buyers to benefit from a special offer at reduced price. However, this special offer is valid only for 3 days.
In this model, a token in control place $p_3$ represents a coupon allowing the special offer. However, time measure for the deal at special price starts as soon as the ad is sent. Hence, if the coupon is sent 2 days after the ad, the customer still has 1 day to benefit from this offer. If the coupon arrives more than 3 days after the ad, he has to use it immediately. Figure~\ref{Example_wn}-b) enhances this example to model expiration of the coupon after 3 days with a transition. Transition $T_O$ consumes urgently a token from place $p_6$ exactly $3$ time units after firing of transition $Ad$ if it is still enabled, which means that the special offer expires within 3 days, and coupon arriving later than $3$ days after the add cannot be used.

The semantics of waiting nets associates clocks to transitions, and lets time elapse if their standard preset is filled. It allows firing of a transition $t$ if the standard \underline{and} the control preset of $t$ is  filled.

\begin{definition}[Enabled, Fully Enabled, Waiting transitions]
\label{def:enabled-newenabled}
\begin{itemize}
\itemsep=0.98pt
    \item[$\bullet$] A transition $t$ is {\em enabled} in marking $M.N$ iff $M \geq \spreset{t}$ (for every standard place $p$ in the  preset of $t$, $M(p) \geq \spreset{t}(p)$). We denote by $\mathsf{Enabled}(M)$ the set of transitions which are enabled from marking $M$, i.e. $\mathsf{Enabled}(M) := \{ t \, | \, M \geq \spreset{t} \}$.

    \item[$\bullet$] A transition $t$ is {\em fully enabled} in $M.N$ iff, for every place in the  preset of $t$, $M.N(p) \geq \preset{t}(p)$. $\mathsf{FullyEnabled}(M.N)$ is the set of transitions which are fully enabled in marking $M.N$, i.e. $\mathsf{FullyEnabled}(M.N) := \{ t \, | \, M.N \geq \preset{t} \}$.
		
		\item[$\bullet$] A transition $t$ is {\em waiting} in $M.N$ iff $t \in \mathsf{Enabled}(M) \setminus \mathsf{FullyEnabled}(M.N)$ ($t$ is enabled, but is still waiting for the control part of its preset). We denote by $\mathsf{Waiting}(M.N)$ the set of waiting transitions.
\end{itemize}
\end{definition}

Obviously, $\mathsf{FullyEnabled}(M.N) \subseteq \mathsf{Enabled}(M)$. For every enabled transition $t$, there is a clock $x_t$ that measures for how long $t$ has been enabled. For every fully enabled transition $t$, $t$ can fire when $x_t\in [\alpha(t),\beta(t)]$. We adopt an {\em urgent} semantics, i.e. when a transition is fully enabled and $x_t=\beta(t)$, then this transition, or another one enabled at this precise instant {\em has to} fire without letting time elapse.
Firing of a transition $t$ from marking $M.N$ consumes tokens from all places in $\preset{t}$ and produces tokens in all places of $\postset{t}$. A consequence of this token movement is that some transitions are disabled, and some
other transitions become enabled after firing of $t$.
%% Harmonize the use of \preset, \postset

\begin{definition}[Transition Firing]
\label{def_trans_firing}
Firing of a transition $t$ from marking $M.N$ is done in two steps. The first step computes an intermediate marking
$M''.N'' = M.N -\preset{t}$ obtained by removing tokens consumed by the transition from its preset. Then, a new marking $M'.N' = M''.N'' + \postset{t}$ is computed. We will write $M.N\overset{t}\longrightarrow M'.N'$ whenever
firing of $t$ from $M.N$ produces marking $M'.N'$.  A transition $t_i$ is {\em newly enabled} after firing of $t$ from
$M.N$ iff it is enabled in $M'.N'$, and either it is not enabled in $M''.N''$, or $t_i$ is a new
%%\eject
\noindent occurrence of $\,t$. We denote by  $\,\newenab{M.N, t}\,$ the set of transitions newly enabled after firing~$\,t\,$
   from marking $M.N$.
    $$ \newenab{M.N, t} := \{ t_i \in T \mid \preset{t_i} \leq M.N - \preset{t} +\postset{t}
		\land
		\big( (t_i = t) \lor ( \preset{t_i} \geq  M.N - \preset{t}) \big ) \} $$			
\end{definition}

%HERE
As explained informally with the examples of Figure~\ref{Example_wn}, the semantics of waiting nets allows transitions firing when some time constraints on the duration of enabling are met. Hence, a proper notion of state for a waiting net has to consider both place contents and time elapsed. This is captured by the notion of {\em configuration}. In configurations, time is measured by attaching a clock to every enabled transition. To simplify  notations, we define valuations of clocks on a set $X_T=\{x_t \mid t\in T\}$ and write $x_t=\bot$ if $t\not\in enabled(M)$. To be consistent, for every value $r\in \mathbb R$, we set $\bot + r := \bot$.

\begin{definition}[Configuration]
A {\em Configuration} of a waiting net is a pair $(M.N,v)$ where $M.N$ is a marking and $v$ is a valuation of clocks in $X_T$.
The {\em initial configuration} of a waiting net is a pair $(M_0.N_0,v_0)$, where
    $v_0(x_t)= 0 $ if $t \in \mathsf{enabled(M_0)}$ and $v_0(x_t)=\bot$ otherwise.
		A transition $t$ is {\em firable} from configuration $(M.N,v)$ iff it is fully enabled, and $v(x_t)\in [ \alpha(t),\beta(t)]$.
\end{definition}

 Slightly abusing our notation, we will write $v(t)$ instead of $v(x_t)$. The semantics of waiting nets is defined in terms of {\em timed} or {\em discrete} moves from one configuration to the next one. Timed moves increase the value of clocks attached to enabled transitions (when time elapsing is allowed) while discrete moves are transitions firings that reset clocks of newly enabled transitions. We give the formal operation definition of timed and discrete moves in Figure~\ref{fig_operational_semantics} and then detail these operational~rules.

%%%%%%%%%%%%%%%%%%%%%%%%%%%%%%%%%%%%%%%%%%%
\begin{figure}[htbp]
\vspace*{-2mm}
\centering
\hspace*{-5mm}
\begin{minipage}{0.4\linewidth}
\scalebox{0.93}{
$\begin{array}{l}
\\
%d\in \mathbb R\\
\forall t\in \mathsf{Waiting}(M.N),\\
\qquad v'(t) = \min(\beta(t),v(t)+d)\\
\forall t\in \mathsf{FullyEnabled}(M.N),\\
\qquad v(t)+d \leq \beta(t) \\
\hspace{1cm}\mbox{~and~} v'(t)= v(t)+d\\
\forall t\in T\setminus \mathsf{Enabled}(M), v'(t)= \bot\\
\hline
\qquad (M.N,v) \overset{d}\longrightarrow (M.N,v')
\end{array}$
}
\end{minipage}
\hspace{0.15cm}
\begin{minipage}{0.4\linewidth}
\scalebox{0.93}{
$\begin{array}{c}
 M. N \geq \preset{t} \\
 M'. N' = M. N - \preset{t} + \postset{t}\\
\alpha(t) \leq v(t) \leq \beta(t)\\
\forall t_i \in T, v'(t_i) = \left\{
	\begin{array}{l}
	  0 \mbox{~if~} t_i \in \newenab{M.N, t}\\
		\bot \mbox{~if~} t_i \not\in \mathsf{enabled}(M)\\
		v(t_i) \mbox{~otherwise}
	\end{array}
	\right.\\
\hline
(M. N,v) \overset{t}\longrightarrow (M'. N',v')
\end{array}$
}
\end{minipage}%\vspace*{-1mm}
\caption{Operational semantics of Waiting Nets}
\label{fig_operational_semantics}\vspace*{-3mm}
\end{figure}

\vspace{\medskipamount}
A {\bf timed move} from a configuration $(M.N,v)$ lets $d\in \mathbb R^{\geq 0}$ time units elapse, but leaves the marking unchanged. We adopt an {\em urgent semantics} that considers differently {\em fully enabled} transitions and {\em waiting} transitions. First of all, elapsing of $d$ time units from $(M.N,v)$ is allowed only if, for every fully enabled transition $t$, $v(t)+d \leq \beta(t)$. If $v(t)+d > \beta(t)$ then $t$ becomes urgent before $d$ time units, and letting a duration $d$ elapse from $(M.N,v)$ is forbidden.
If we already have $v(t)=\beta(t)$, then no timed move is allowed, and firing of $t$ is {\em urgent}: $t$ {\em has to} be fired or disabled by the firing of another transition before elapsing time.
On the other hand, if elapsing $d$ time units does not violate urgency of any fully enabled transition, the new valuation for the clock $x_t$ attached to a fully enabled transition after elapsing $d$ time units is $v'(t)=v(t)+d$.
Urgency does not apply to waiting transitions, which can let an arbitrary amount of time elapse when at least one control place in their preset is not filled. Now, %%%
\noindent as we model the fact that an event has been enabled for a sufficient duration, we let the value of clocks attached to waiting transitions increase up to the upper bound allowed by their time interval, and then freeze these clocks.
So, for a waiting transition, we have $v'(t) = \min (\beta(t),v(t)+d)$. We will write $v \oplus d$ to denote the valuation of clocks reached after elapsing $d$ time units from valuation $v$.
A timed move of duration $d$ from configuration $(M.N,v)$ to $(M.N,v')$ is denoted $(M. N,v) \overset{d}\longrightarrow (M. N,v')$. As one can expect, waiting nets enjoy time additivity (i.e. $(M.N,v) \overset{d_1}\longrightarrow  (M.N,v_1) \overset{d_2}\longrightarrow (M.N,v_2)$ implies that $(M.N,v) \overset{d_1+d_2}\longrightarrow  (M.N,v_2)$, and continuity, i.e. if $(M.N,v) \overset{d}\longrightarrow  (M.N,v')$, then for every $d'<d$ $(M.N,v) \overset{d'}\longrightarrow  (M.N,v'')$.	

\vspace{\medskipamount}
A {\bf discrete move} fires a fully enabled transition $t$ whose clock $x_t$ meets the time constraint of $t$ (i.e., $\alpha(t) \leq x_t \leq \beta[t]$), and resets clocks attached to transitions newly enabled
by token moves.  A discrete move relation from configuration $(M.N,v)$ to $(M'.N',v')$ via transition $t_i \in \mathsf{FullyEnabled}(M.N)$ is denoted $(M. N,v) \overset{t_i}\longrightarrow (M'. N',v')$. Overall, the semantics of a waiting net $\mathcal{W}$ is a timed transition system (TTS) with initial state $q_0 = (M_0. N_0, v_0)$ and which transition relation follows the time and discrete move semantics rules.

\begin{definition}
A {\em run} of a waiting net $\mathcal{W}$ from a configuration $(M.N,v)$ is a sequence
$$\rho=(M.N,v) \overset{e_1}\longrightarrow (M_1.N_1,v_1)\overset{e_2}\longrightarrow (M_2.N_2,v_2) \cdots \overset{e_k}\longrightarrow (M_k.N_k,v_k)$$ where every $e_i$ is either a duration $d_i \in \mathbb R^{\geq 0}$, or a  transition $t_i\in T$, and
every $(M_{i-1}.N_{i-1},v_{i-1})\overset{e_i}\longrightarrow (M_{i}.N_{i},v_{i})$ is a legal move of $\mathcal{W}$.
\end{definition}

We denote by $\mathsf{Runs(\mathcal W)}$ the set of runs of $\mathcal W$ starting from $(M_0.N_0,v_0)$. A marking $M.N$ is {\em reachable} iff there exists a run from $(M_0.N_0,v_0)$ to a configuration $(M.N,v)$ for some $v$. $M.N$ is {\em coverable} iff there exists a reachable marking $M'.N' \geq M.N$.
We will say that a waiting net is {\em bounded} iff there exists an integer $K$ such that, for every reachable marking $M.N$ and every place $p\in P$ and $p'\in C$, we have $M(p)\leq K$ and $N(p')\leq K$.
Given two markings $M_0.N_0$ and $M.N$ the {\em reachability} problem asks whether $M.N$ is reachable from $(M_0.N_0,v_0)$, and the
{\em coverability} problem whether there exists a marking $M'.N'\geq M.N$ reachable from $(M_0.N_0,v_0)$.

\begin{remark}
\label{remark_inheritance}
A waiting net with an empty set of control places is a TPN. Hence, waiting nets inherit all undecidability results of TPNs: reachability, coverability and boundeness~\cite{Jone77} are undecidable in general for unbounded waiting nets.
\end{remark}

Given a run $\rho = (M_0. N_0,v_0) \overset{e_1}\longrightarrow (M_1.N_1,v_1)\overset{e_2}\longrightarrow (M_2.N_2,v_2) \cdots$, the timed word associated with $\rho$ is the word $w_\rho = (t_1,d_1)\cdot(t_2,d_2) \cdots$
where the sequence $t_1\cdot t_2 \dots $ is the projection of $e_1\cdot e_2 \cdots $ on $T$, and for every $(t_i,d_i)$ such that $t_i$ appears on move $(M_{k-1}.N_{k-1},v_{k-1})\overset{e_k}\longrightarrow (M_k.N_k,v_k)$, $d_i$ is the sum of all durations in $e_1 \dots e_{k-1}$. The sequence $t_1.t_2\dots$ is called the {\em untiming} of $w_\rho$. The {\em timed language} of a waiting net is the set of timed words $\lang(\mathcal{W}) =\{ w_\rho  \mid \rho \in \mathsf{Runs(\mathcal W)}\}$. Notice that unlike in timed automata and unlike in the models proposed in~\cite{BCHLR13}, we do not define accepting conditions for runs or timed words, and hence consider that the timed language of a net is prefix closed.
The {\em untimed language} of a waiting net $\mathcal W$ is the language
$\lang^U(\mathcal W) = \{ w\! \in T^* \mid  \exists w_\rho \in \lang(\mathcal W), \linebreak w~\mbox{~is the untiming  of~}  w_\rho\}$.
These definitions can be adapted the obvious way to consider labels of transitions instead of transition themselves (i.e. timed words are sequences of pairs of the form $(\lambda(t_i),d_i)$, and labeling of transitions by $\epsilon$.
To simplify notations, we will consider runs alternating timed and discrete moves. This results in no loss of generality, since durations of consecutive timed moves can be summed up, and a sequence of two discrete move can be seen as a sequence of transitions with 0 delays between discrete moves. In the rest of the paper, we will write $(M.N,v) \overset{(d,t)}\longrightarrow (M'.N',v')$ to denote the sequence of moves $(M.N,v) \overset{d}\longrightarrow (M.N,v\oplus d) \overset{t}\longrightarrow (M'.N',v')$.

\medskip
Let us illustrate the semantics of waiting nets with the example in Figure~\ref{Example_wn}-a). In this net, we have $P =  \{ p_0,p_1, p_2,p_4,p_5\}$, $C = \{p_3\}$,
    $T = \{Ad,No,So,Cp\}$, $\alpha(Ad)=\alpha(No)=\alpha(So)=0$, $\alpha(Cp)=1$, $\beta(Ad)=\infty$, $\beta(No)=8$, $\beta(So)=3$, $\beta(Cp)=4$.
    We also have $\spreset{So}=p_1$ and $\cpreset{So}=p_3$, $\postset{So}=p_4$ (we let the reader infer $\preset{}$ and $\postset{}$ for other transitions). For conciseness, we will denote by $\{ p_i, p_j,..\}.\{p_k,p_r,...\}$ a marking where standard places $p_i, p_j,...$ and control places $p_k,p_r,...$ are marked.
		The net starts in an initial configuration $(M_0.N_0,v_0)$ where $M_0.N_0=\{p_0\}.\{\}$,  $v_0(Ad)=0$ and $v_0(t)=\bot$ for all other transitions in $T$.
	
\medskip			
A possible run of the net of Figure~\ref{Example_wn}-a) is :
%$$(M_0.N_0,v_0) \overset{(d_0,Ad)}\longrightarrow (M_1.N_0,v_1) \overset{(2,Cp)} (M_2.N_2,v_1) \overset{(1,No)} (M_2.N_2,v_1)$$		
$$(\{p_0\}.\{\},v_0) \overset{(d_0,Ad)}\longrightarrow (\{p_1,p_2\}.\{\},v_1) \overset{(2,Cp)}\longrightarrow (\{p_1\}.\{p_3\},v_2) \overset{(1,No)}\longrightarrow (\{p_5\}.\{p_3\},v_3)$$		
where $d_0$ is any arbitrary positive real value $v_1(ad)=\bot$ and $v_1(No)=v_1(So)=v_1(Cp)=0$, $v_2(ad)=v_2(Cp)=\bot$ and $v_2(No)=v_2(So)=2$.
This run represents a situation where a advertisement is sent at some date $d$, a coupon 2 days later, and the customer decides to take the normal offer. Another possible run is :
$$(\{p_0\}.\{\},v_0) \overset{(d_0,Ad)}\longrightarrow (\{p_1,p_2\}.\{\},v_1) \overset{(3,Cp)}\longrightarrow (\{p_1\}.\{p_3\},v'_2) \overset{(1,So)}\longrightarrow (\{p_4\}.\{\},v'_3)$$		
with $v'_2(ad)=v'_2(Cp)=\bot$ and $v'_2(No)=v'_2(So)=3$. This run represents a scenario where a customer decides to use the coupon sent. Notice that from $(\{p_1\}.\{p_3\},v'_2)$, time cannot elapse, because transition $So$ is fully enabled, $v'_2(So)=3$, and $\beta(So)=3$. Notice also that firing of $So$ can only occur after firing of $Cp$, but yet time measurement starts for $So$ as soon as $\spreset{So}=\{p_1\}$ is filled, i.e. immediately after firing of $Ad$. This example is rather simple: the net is acyclic, and each transition is enabled/disabled only once.  One can rapidly see that the only markings reachable are $M_0.N_0=\{p_0\}.\{\}$, $M_1.N_0=\{p_1,p_2\}.\{\}$, $M_2.N_2=\{p_1\}.\{p_3\}$, $M_3.N_3=\{p_5\}.\{p_3\}$ and $M'_3.N'_3=\{p_4\}.\{\}$ visited in the runs above, and two additional markings $M_4.N_4=\{p_2,p_5\}.\{\}$ (reached when transition $No$ is fired before transition $Cp$ and $M_5.N_5 =\{p_5\}.\{p_3\}$ reachable immediately from $M_4.N_4$ by firing $Cp$. A normal order can be sent at most 8 time units after advertising, a special order must be sent at most 3 time units after advertising if a coupon was received, etc.

\section{Reachability}
\label{sec:reachability}

%In section~\ref{sec:waiting}, we have provided a semantics for Waiting Nets.
In a configuration $(M.N,v)$ of a waiting net $\wnet$, $v$ assigns real values to clocks. The timed transition system giving the semantics of a waiting net is hence in general infinite, even when $\wnet$ is bounded. In timed automata, clock valuations can be partitioned into equivalence classes called regions~\cite{AlurD94} to get
%\eject
\noindent  a finite abstraction. For TPNs, these equivalence classes are called {\em domains}, and can be used to define an untimed transition systems called a {\em state class graph}~\cite{BerthomieuD91}, that recognizes firable sequences of transitions. In this section, we build similar graphs for waiting nets. We also prove that the set of domains in these graphs is always finite, and use this result to show that reachability and coverability are decidable for bounded waiting nets.
The type of domain built for waiting nets is identical to that of TPNs, and the construction of successors is based on the same operations, namely verification of the emptiness of domains, variable substitution, variable   elimination, and calculus of a normal form. The main change is that one has to consider time elapsed when computing a successor domain to take into account waiting transitions that have reached their upper bound. A consequence is that a domain can have several successors via the same transition firing. Yet, verification of coverability and reachability remain similar to that of TPNs, i.e. yield an exponential time complexity w.r.t. the number of transitions.

\medskip
Let $t_i$ be a fully enabled transition with $\alpha(t_i)=3$ and $\beta(t_i)= 12$, and assume that $t_i$ has been enabled for $1.6$ time units. According to the semantics of WPNs, $v(t_i)=1.6$, and $t_i$ cannot fire yet, as $x_{t_i} < \alpha(t_i)$. Transition $t_i$ can fire only after a certain duration $\theta_i$ such that $1.4 \leq \theta_i \leq 10.4$. Similar constraints hold for all enabled transitions. We will show later that they are not only upper and lower bounds on $\theta_i$'s of the form $a_i \leq \theta_i \leq b_i$, but can  include dependencies of the form $\theta_i - \theta_j \leq c_{ij}$ on pairs of variables.
These constraints form a {\em domain} of legal real values defining the time that can elapse before firing an enabled transition or  reaching the upper bound of the interval attached to it.

\begin{definition}[State Class, Domain]
\label{def-stateclass}
A {\em state class} of a waiting net $\wnet$ is a pair $(M.N, D)$, where $M.N$ is a marking of $\wnet$ and $D$ is a set of inequalities called {\em firing domain}. The inequalities in $D$ are of two types:

\vspace{-\medskipamount}
$$\begin{cases}
a_i \leq \theta_i \leq b_i, &\text{ where } a_i,b_i \in \mathbb{Q}^+ \text{ and } t_i \in \mathsf{Enabled}(M)\\
\theta_j - \theta_k \leq c_{jk}. &\text{ where } c_{jk} \in \mathbb{Q} \text{ and } t_j,t_k \in \mathsf{Enabled}(M).
\end{cases}$$

\end{definition}

A variable $\theta_i$ in a firing domain $D$ over variables $\theta_1,\dots,\theta_m$  represents the time that can elapse before firing  transition $t_i$ if $t_i$ is fully enabled, and the time that can elapse before the clock attached to $t_i$ reaches the upper bound $\beta(t_i)$ if $t_i$ is waiting. Hence, if a transition is fully enabled, and $a_i \leq \theta_i \leq b_i$, then $t_i$ cannot fire before $a_i$ time units, and cannot let more than $b_i$ time units elapse, because it becomes urgent and has to fire or be disabled before $b_i$ time units. Inequalities in $D$ define an infinite set of possible values for $\theta_i's$. We denote by $\llbracket D \rrbracket$ the set of solutions for a firing domain $D$. Domains give upper and lower bounds on $\theta_i's$ with inequalities of the form $a_i \leq \theta_i \leq b_i$, but also contain diagonal constraint of the form $\theta'_j - \theta'_k \leq c_{jk}$. These constraints appear when two transitions $t_i,t_j$ are enabled in a state class, and one has to add inequalities of the form $\theta_i \leq \theta_j$ to check firability of $t_i$. We will say that $t_i$ is {\em firable} from $(M.N,D)$ if there exists $v\in\llbracket D \rrbracket$ such that $(M.N,v)\overset{t_i}\longrightarrow (M'.N',v')$.
We will say that $(M'.N',D')$ is a {\em successor} of $(M.N,D)$ via transition $t_i$ iff $t_i$ is firable from $(M.N,D)$ and $D'$ is a domain describing the set of values for $\theta_k's$ that are reachable when elapsing time and firing $t_i$ from a valuation in $D$.
The initial state class for a waiting net is the state class $(M_0.N_0,D_0)$ where $M_0.N_0$ is the usual initial marking of a waiting net, and $D_0$ is the domain
$$D_0 = \{ \alpha(t_i) \leq \theta_i \leq \beta(t_i) \mid t_i \in \mathsf{enabled}(M_0) \}$$
\eject

The main reason to work with domains is to abstract time and bring back decidability issues for reachability or coverability on an infinite timed transition system to a similar problem on a finite untimed structure. This calls for three properties of the abstraction: first, one has to be able to compute effectively a successor relation $\mathsf{Succ()}$ among state classes.
Second, the closure $\mathsf{Succ^*(M_0.N_0,D_0)}$ of the initial state class through this successor relation has to be finite. Last, the construction has to be sound and complete. These three properties hold for time Petri nets~\cite{MenascheB83,BerthomieuD91}. In the rest of this paper, we will show that a sound, complete and effective construction of state class graphs for waiting nets exists, but that it slightly differs from existing constructions defined for TPNs.

\medskip
The way to define a set of solutions $\llbracket D \rrbracket$ is not unique. We will say that $D_1$, $D_2$ are {\em equivalent}, denoted $D_1 \equiv D_2$ iff  $\llbracket D_1 \rrbracket = \llbracket D_2 \rrbracket$. A set of solutions $\llbracket D \rrbracket$ is hence not uniquely defined, but fortunately, a unique representation called a {\em canonical form} exists~\cite{Menasche82}.

\begin{definition}[Canonical Form]
\label{def-canonical-form}
The canonical form of a firing domain $D$ is the unique domain
$$D^*=\left \{ \begin{array}{l}
a_i^* \leq \theta_i \leq b_i^*\\
\theta_j - \theta_k \leq c_{jk}^*
\end{array}\right.$$
where $a^*_i = Inf(\theta_i)$, $b^*_i = Sup(\theta_i)$, and $c^*_{jk} = Sup(\theta_j - \theta_k)$
\end{definition}

The  canonical form $D^*$ is the minimal set of constraints defining the set of solutions
$\llbracket D \rrbracket$. If two sets of constraints are equivalent then they have the same canonical form. The constraints we consider are of the form
$a_i \leq \theta_i \leq b_i$ and  $ \theta_i -\theta_j \leq c_{ij}$, where $a_i, b_i, c_{ij}$ are rational values. This type of constraint can be easily encoded by \textit{Difference Bound Matrices}~\cite{Dill89}, or by constraint graphs (i.e., graphs where vertices represent variables or the value 0, and an edge from $\theta_i$ to $\theta_j$ of weight $w_{x,y}$ represents the fact that $\theta_i-\theta_j \leq w_{x,y}$).
The canonical form for a domain is obtained as a closure operation, that amounts to computing the shortest paths for each pair $\theta_i,\theta_j$ in the constraint graph. This can be achieved using the Floyd-Warshall algorithm, in $O(n^3)$, where $n$ is the number of variables in the domain. We refer interested readers to~\cite{BengtssonY03} for a survey on operations on constraints, specified using DBMs.
Note that the set of constraints in a domain $D$ may not be satisfiable, i.e. $\llbracket D \rrbracket$ can be empty. In particular, for TPNs and waiting nets, unsatisfiable domains may appear when checking whether a particular transition can be the first one to fire in a domain $D$. Satisfiability of a domain is decidable, and amounts to checking existence of negative cycles in the constraint graph of $D$, wich can again be achieved in polynomial time with Floyd-Warshall.

\begin{figure}[h]
\vspace*{-3mm}
\begin{center}
\scalebox{0.78}{
\begin{tikzpicture}
%%%%%%%%%%%%%%%%%%%%%%%%%%%%%%%%%%%%%%%%%%%%%%%%%%%%%%%%%%%%%%%%%%%%%%%%
%
\begin{scope}
\node [draw,circle,minimum size=0.7cm] (p0) at (0.3,0) {$\bullet$};
\node [] (lp0) at (0.3,0.5) {$p_0$};
\node [draw,circle,minimum size=0.7cm, dashed] (c0) at (0.3,-1) {$\bullet$};
\node [] (lp0) at (0.3,-1.5) {$c_0$};

\node [draw,circle,minimum size=0.7cm] (p1) at (0.3,2) {$\bullet$};
\node [] (lp1) at (0.3,2.5) {$p_1$};

\node [draw,circle,minimum size=0.7cm] (p2) at (2.7,0) {};
\node [] (lp2) at (2.7,0.5) {$p_2$};

\node [draw,circle,minimum size=0.7cm] (p3) at (2.7,2) {};
\node [] (lp3) at (2.7,2.5) {$p_3$};

\node [draw,rectangle,fill=black,minimum height=1cm] (t0) at (1.5,0){};
\node [] (lt0) at (1.5,0.7) {\scriptsize$t_0$};
\node [color=red] (it0) at (1.5,-0.7) {\scriptsize$[2,3]$};

\node [draw,rectangle,fill=black,minimum height=1cm] (t1) at (1.5,2){};
\node [] (lt1) at (1.5,2.7) {\scriptsize$t_1$};
\node [color=red] (it1) at (1.5,1.3) {\scriptsize$[4,5]$};

\draw [-latex] (p0) edge  (t0){};
\draw [-latex] (c0) edge  (t0){};
\draw [-latex] (p1) edge  (t1){};
\draw [-latex] (t0) edge  (p2){};
\draw [-latex] (t1) edge  (p3){};

\node [] (dom) at (7,2.5) {$D_0 = \left \{
\begin{array}{l}
2 \leq \theta_0 \leq 3\\
4 \leq \theta_1 \leq 5\\
\end{array}
\right.$};

\node [] (mat) at (7,0) {
$\begin{array}{l|ccc}
&0 & \theta_0 & \theta_1 \\
\hline
0 & (\leq,0) & (\leq,-2) & (\leq,-4)\\
\theta_0 & (\leq,3) & (\leq,0) & (\leq,\infty)\\
\theta_1 & (\leq,5) & (\leq,\infty) & (\leq,0) \\
\end{array}
$};

\end{scope}

\begin{scope}[xshift = 12cm]

\node [circle,draw] (v0) at (2,0){$0$};
\node [circle,draw] (t0) at (0,0){$\theta_0$};
\node [circle,draw] (t1) at (4,0){$\theta_1$};

\draw [-latex, bend right] (v0) edge node[above] {$-2$} (t0){};
\draw [-latex, bend right] (t0) edge node[above]{$3$} (v0){};
\draw [-latex, bend right] (v0) edge node [above] {$-4$} (t1){};
\draw [-latex, bend right] (t1) edge node[above]{$5$} (v0){};

\draw [-latex, bend left=40, dashed] (t1) edge node[below]{$0$} (t0){};
\end{scope}
\end{tikzpicture}
}\vspace*{-2mm}
\caption{A Waiting net, its initial domain, an equivalent DBM representation, and its constraint graph.}
\label{fig_example_dom_mat_graph}
\end{center}\vspace*{-7mm}
\end{figure}
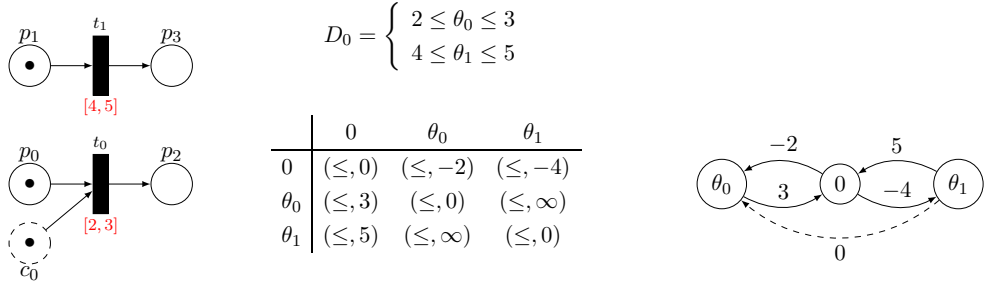

\medskip
Consider for instance the example of Figure~\ref{fig_example_dom_mat_graph}. This simple waiting net has two enabled transitions in its initial marking $M_0.N_0=\{p_0,p_1\}.\{c_0\}$, and its initial domain is the pair of inequalities $D_0$ represented in the Figure. These constraint can alternatively be represented by a DBM (also shown in the Figure), or a constraint graph ( restricted to plain edges) represented on the left of the Figure. Clearly, $D_0$ has a solution. For instance $\theta_0 = 2.2, \theta_1=4.5$ is a solution for the system of inequalities in $D_0$. Let us test if transition $t_1$ can fire before $t_0$. This amounts to adding the constraint $\theta_1 \leq \theta_0$ to domain $D_0$, or equivalently the dashed edge with weight 0 in the constraint graph. As expected, $D_0 \cup \{\theta_1 \leq \theta_0 \}$ is not satisfiable, and one can verify that adding the dashed edge creates a negative cycle in the constraint graph.

Syntactically, state classes, domains and their canonical forms are identical for TPNs and waiting nets. Indeed, the fact that a transition is waiting or fully enabled does not affect the representation of constraints, as each $\theta_i$ represents a time remaining until the upper bound $\beta(t_i)$ is reached by clock $x_{t_i}$. However, there is a major difference in the construction of successors for state classes of TPNs and those of waiting nets. In TPNs, the maximal duration that can elapse in a state class is constrained by {\em all} enabled transitions.
For waiting nets, the maximal duration that can elapse in a state class is constrained by fully enabled transitions only: elapsing $\delta$ time units before firing some transition $t_i$ can be allowed even if a waiting transition $t_j$ reaches its upper bound after $\delta' < \delta$ time units, because variable $\theta_j$  represents a {\em time to upper bound of interval $[\alpha(t_j),\beta(t_j)]$} rather than a {\em time to fire}.
On the other hand, when elapsing time, one has to consider which clocks attached to waiting transitions have reached their upper bound. A consequence is that a state class has several successors via the same transition, and hence state class graphs of waiting nets are not deterministic.

In state classes of TPNs, verifying that $t_i$ can fire from $(M.N,D)$ is checked by verifying that $D^{+\theta_i} :: = D \cup \{ \theta_i \leq \theta_j \mid t_j \in \mathsf{Enabled}(M.N)\}$ is satisfiable. Intuitively, satisfiability of $D^{+\theta_i}$ means that it is possible to fire $t_i$ witout violating urgency of other transitions.
Following the semantics of section~\ref{sec:waiting},
a transition $t_i$ can fire from a configuration $(M.N,v)$ iff it is fully enabled and $v(t_i)\in [\alpha(t_i), \beta(t_i) ]$. Hence, from configuration $(M.N,v)$, firing of $t_i$ is one of the next discrete moves iff there exists a duration $\theta_i$ such that $t_i$ can fire from $(M.N,v+\theta_i)$, i.e., after letting duration $\theta_i$ elapse, and no other transition becomes urgent before $\theta_i$ time units.
One can check firability of a transition $t_i$ as in TPNs through satisfiability of $D^{+\theta_i}$, but considering only fully enabled transitions. Note however that if there is a  waiting transition $t_j$, $D$ contains a pair of constraints of the form $\theta_j \leq b_j$, and $\theta_i-\theta_j \leq 0$, then the constraint on $\theta_j$ may over-constrain the possible values for $\theta_i$. To avoid any conflict between fully enabled transitions and waiting transitions, we define a notion of {\em projection} of a domain on its fully enabled transitions, that erases all upper bounds of waiting transitions, and diagonal constraints involving variables associated with waiting transitions.

\begin{definition}[Projection]
Let $D$ be a firing domain with variables $a_i,b_i,c_{jk}$ set as in def.~\ref{def-stateclass}.
The {\em projection} of $D$ on its fully enabled transitions is a domain
\vspace{-\medskipamount}
$$\begin{array}{lcll}
	 &  &\{ a_i \leq \theta_i \leq b_i \mid t_i \in \mathsf{FullyEnabled}(M.N) \}\\
	D_{|full} = & \cup & \{ a_i \leq \theta_i \leq \infty, \mid t_i \in \mathsf{Waiting}(M.N) \}\\
	& \cup & \{ \theta_j - \theta_k \leq c_{jk} \in D \mid t_j,t_k \in \mathsf{FullyEnabled}(M.N)\}.
\end{array}$$
\end{definition}

%Let us give an intuition behind projection: a waiting transition $t_i$ cannot fire before it clock reaches the lower bound of its interval, which means that the time to fire of such a clock has to be greater than $a_i$. However, as $t_i$ can wait an arbitrary amount of time before being fired when a control places of its preset is empty, there is no upper bound on $\theta_i$. The time to upper bound of simply enabled transitions should not constrain other fully enabled transitions, so the only variable differences that matter for urgent firing are differences among TTFs of fully enabled transitions.
%
%
In waiting nets, $t_i$ is {\em firable} from a state class $(M.N,D)$ iff $M.N \geq \preset{t_i}$ and $D_{|full} \cup \{ \theta_i \leq \theta_j \mid t_j \in \mathsf{FullyEnabled}(M.N)\}$ is satisfiable, i.e. there exists a value for $\theta_i$ and for every $\theta_j$ attached to a fully enabled transition $t_j$ that does not violate the constraint that $\theta_j$ is greater than or equal to $\theta_i$ and does not exceed the upper bound $b_j$ appearing in $D$.

The construction of the set of reachable state classes of a waiting net is an inductive procedure.
Originally, a waiting net starts in a configuration $(M_0.N_0,v_0)$, so the initial state class of our system is $(M_0,D_0)$, where
$D_0=\{ \alpha(t_i) \leq \theta_i \leq \beta(t_i) \mid t_i\in \mathsf{Enabled}(M_0.N_0)\}$. Then, for every state class $(M.N,D)$, and every transition $t$ firable from $(M.N,D)$, we compute all possible successors $(M'.N',D')$ reachable after firing of $t$.
Note that we only need to consider $t \in \mathsf{FullyEnabled}(M.N)$, as $t$ can fire only when $N > \cpreset{t}$.
Computing $M'.N'$ follows the usual firing rule of a Petri net: $M'.N' =M.N -\preset{t} +\postset{t}$ and we can hence also compute $\newenab{M.N, t}$, $\mathsf{enabled}(M'.N')$ and $\mathsf{FullyEnabled}(M'.N')$.
It remains to show the effect of transitions firing on domains to compute all possible successors of a class. Firing a transition $t$ from $(M.N,D)$ propagates constraints of the firing domain $D$ on variables attached to transitions that remain enabled. Variables associated to newly enabled transitions only have to meet lower and upper bounds on their firing times. We can now show that for waiting nets, the set of successors of a state class is finite and can be effectively computed despite waiting transitions and non-determinism.

%%%%%%%%%%%%%%%%%%%%%%%%%%%%%%%%%%%%%%%%%%%%%%%%%%%%%%%%%%%%%%%%%%%%
% FORMER EXAMPLE : a simple waiting net with parallel transitions
% and common enabling
% Illustrates the influence of bounds when computing classes
%%%%%%%%%%%%%%%%%%%%%%%%%%%%%%%%%%%%%%%%%%%%%%%%%%%%%%%%%%%%%%%%%%%%
\begin{figure}[htbp]
\begin{center}
\begin{tikzpicture}
\begin{scope}[scale=0.88]
\node [draw,circle,minimum size=0.7cm] (p0) at (0,0) {$\bullet$};
\node [] (lp0) at (-0.6,0.3) {$p_0$};

\node [draw,circle,minimum size=0.7cm] (p1) at (2.5,1.2) {$\bullet$};
\node [] (lp1) at (1.9,1.5) {$p_1$};

\node [draw,dashed,circle,minimum size=0.7cm] (p2) at (2,0) {};
\node [] (lp2) at (1.4,0.3) {$p_2$};

\node [draw,circle,minimum size=0.7cm] (p3) at (2.5,-1.2) {$\bullet$};
\node [] (lp3) at (1.9,-0.9) {$p_3$};

\node [draw,circle,minimum size=0.7cm] (p4) at (4.5,1.2) {};
\node [] (lp4) at (5,1.5) {$p_4$};

\node [draw,circle,minimum size=0.7cm] (p5) at (4.5,-1.2) {};
\node [] (lp4) at (5,-0.9) {$p_5$};

\node [draw,rectangle,fill=black,minimum height=1cm] (t0) at (1,0){};
\node [] (lt0) at (1,0.8) {\scriptsize$t_0$};
\node [color=red] (it0) at (0.9,-0.7) {\scriptsize$[0,\infty)$};

\node [draw,rectangle,fill=black,minimum height=1cm] (t1) at (3.5,1.2){};
\node [] (lt1) at (3.5,1.9) {\scriptsize$t_1$};
\node [color=red] (it1) at (3.5,0.5) {\scriptsize$[0,3]$};

\node [draw,rectangle,fill=black,minimum height=1cm] (t2) at (3.5,-1.2){};
\node [] (lt2) at (3.5,-0.5) {\scriptsize$t_2$};
\node [color=red] (it1) at (3.5,-1.9) {\scriptsize$[5,6]$};

\draw [-latex] (p0) edge  (t0){};
\draw [-latex] (t0) edge  (p2){};
\draw [-latex] (p1) edge  (t1){};
\draw [-latex] (p3) edge  (t2){};
\draw [-latex] (p2) edge  (t1){};
\draw [-latex] (p2) edge  (t2){};
\draw [-latex] (t1) edge  (p4){};
\draw [-latex] (t2) edge  (p5){};
\end{scope}
\end{tikzpicture}
\vspace{-\medskipamount}
\caption{Upper bounds of transitions influence state classes of waiting Nets}
\label{FigureInfluenceBounds}
\vspace{-2\medskipamount}
\end{center}
\end{figure}
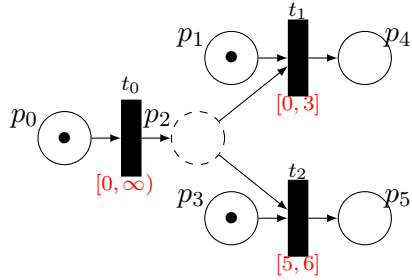

Let us illustrate the influence of control places on the construction of state classes with the waiting net of Figure~\ref{FigureInfluenceBounds}.
This net starts in a configuration $C_0=(M_0.N_0,v_0)$ with $M_0.N_0= \{p_0,p_1,p_3\}.\{ \}$. Transitions $t_0,t_1,t_2$ are enabled in this marking, but $t_0$ is the only fully enabled transition. From configuration $C_0$, one can let an arbitrary amount of time $\delta \in \mathbb R^{\geq 0}$ elapse, and then fire transition $t_0$. If $t_0$ fires from $C_0'=(M_0.N_0,v_0+\delta)$, the net reaches a new configuration $(M'.N', v')$ where $M'.N'=\{p_1,p_3\}.\{p_2\}$, and $v'(t_0)=\bot, v'(t_1)=v'(t_2)=\delta$.

\medskip
Let $C_1=(M'.N',v_1)$ be a configuration reached when $0 \leq \delta < 3$. The value of clock $x_{t_1}$ is still smaller than the upper bound $\beta(t_1)=3$. Then, from $C_1$, one can still wait before firing $t_1$, i.e., $t_1$ is not urgent and can fire immediately or within a duration $3-\delta$, and transition $t_2$ is not firable, and cannot become firable before the urgent firing of $t_1$.
Now, let $C_2=(M'.N',v_2)$ be a configuration reached after elapsing $\delta$ with $3 \leq \delta < 5$. In $C_2$, we have $v_2(t_2)<5$ and $v_2(t_1)=3$, as the value of clock $x_{t_1}$ was frozen when reaching the upper bound $\beta(t_3)=3$. Hence, in $C_2$, transition $t_1$ is urgent and must fire, and transition $t_2$ still has to wait before firing. Hence, choosing $3 \leq \delta < 5$ forces to fire $t_1$ immediately after $t_0$. As this firing consumes the token in control place $p_2$, it will prevent $t_2$ from firing.
In a similar way, let us call $C_3=(M'.N',v_3)$ a configuration reached after elapsing $\delta \geq 5$ time units and then firing $t_0$. In this configuration, we have $v_3(t_1)=3$ and $v_3(t_2) \in [5,6]$, forcing $t_1$ or $t_2$ to fire immediately without elapsing time.
This example shows that the time elapsed in a configuration has to be considered when computing successors of a state class. We have to consider whether the upper bound of a waiting transition has been reached or not, and hence to differentiate several cases when firing a single transition $t$. Fortunately, these cases are finite, and depend only on upper bounds attached to waiting transitions in domain $D$.

%%%%%%%%%%%%%%%%%%%%%%%%%%%%%%%%%%%%%%%%%%%%%%%%%%%%%%%%%%%%%%%%%%%%%%%%%%%%%%%%%%%%

\begin{definition}[Upper Bounds Ordering]
Let $M.N$ be a marking, $D$ be a firing domain with constraints of the form $a_i \leq \theta_i \leq b_i$. Let $B_{M.N,D}= \{ b_i \mid t_i \in \mathsf{enabled}(M)\}$. We can order bounds in $B_{M.N,D}$, and define
$\bound{i}$ as the $i^{th}$ bound in $B_{M.N,D}$. We also define $\bound{0} = 0 $ and $\bound{|B_{M.N,D}|+1} = \infty$.
\end{definition}

In the example of Figure~\ref{FigureInfluenceBounds}, the initial domain $D_0$ is
$$ D_0=\left \{ \begin{array}{l} 0 \leq \theta_0 \leq \infty\\
 0 \leq \theta_1 \leq 3\\
5 \leq \theta_2 \leq 6\\
\end{array}
\right.$$

The bounds appearing in $D_0$ are hence $0,3,6,\infty$.
Consider a transition $t_f$ firable from $C=(M.N,D)$. This means that there is a way to choose a delay $\theta_f$ that does not violate urgency of all other transitions.
We use $B_{M.N,D}$ to partition the set of possible values for delay $\theta_f$ in a finite set of intervals, and find which transitions reach their upper bound when $\theta_f$ belongs to an interval.
Recall that $t_f$ can fire only if adding constraint $\theta_f \leq \theta_j$ for every fully enabled transition $t_j$ to domain $D$ still yields a satisfiable set of constraints.
This means that when considering that $t_f$ fires after a delay $\theta_f$ such that $\bound{i}\leq  \theta_f \leq \bound{i+1}$,
as $D$ also contains a constraint of  the form $a_f \leq \theta_f \leq b_f$, considering an interval such that  $\bound{i}$ is greater than $\min \{ b_j \in B_{M.N,D} \mid t_j \in \mathsf{FullyEnabled}(M.N)\}$ or smaller than $a_f$ immediately leads to inconsistency of constraint
$D_{|full} \cup \underset{t_j\in \mathsf{FullEnabled}(M.N)} \bigwedge \theta_f \leq \theta_j \wedge \bound{i} \leq \theta_f \leq \bound{i+1}$. We denote by $B^{t_f}_{M.N,D}$ the set of bounds $B_{M.N,D}$ pruned out from these inconsistent bound values. Now, choosing a particular interval $[\bound{i},\bound{i+1}]$ for the possible values in $\theta_f$ indicates for which waiting transitions $t_1, \dots t_k$ the clocks $x_{t_1}, \dots x_{t_k}$ measuring time elapsed since enabling has reached upper bounds $\beta(t_1), \dots \beta(t_k)$. The values of these clocks become irrelevant, and hence the corresponding $\theta_i$'s have to be eliminated from the domains.

%HERE

\begin{definition}[Time progress (to the next bound)]
Let $M.N$ be a marking, $D$ be a firing domain, and $b =\min B_{M.N,D}$ be the smallest upper bound for enabled transitions.
The domain reached after progressing time to bound $b$ is the domain $D'$ obtained by:
\vspace{-\medskipamount}
\begin{itemize}
\item replacing every variable $\theta_i$ by expression $\theta_i'-b$
\item eliminating every $\theta'_k$ whose upper bound is $b$,
\item computing the normal form for the result and renaming all $\theta_i'$ to $\theta_i$
\end{itemize}
\end{definition}

Progressing time to the next upper bound allows to remove variables related to waiting transitions whose clocks have reached their upper bounds from a firing domain. We call these transitions {\em timed-out transitions}. For a  transition $t_k \in waiting(M.N)$ if $v(x_{t_k}) = \beta(t_k)$, variable $\theta_k$, that represents the time needed to reach the upper bound of the interval is not meaningful any more: either $t_k$ gets disabled in the future, or is fired urgently as soon as it becomes fully enabled, because $\theta_k=0$. So the only information to remember is that $t_k$ will be urgent as soon as it becomes fully enabled.

\begin{definition}[Successors]
\label{def-successors}
A {\em successor} of a class $C=(M.N,D)$ after firing of a transition $t_f$ is a class $C'=(M'.N',D')$ such that $M'.N'$ is the marking obtained after firing $t_f$ from $M.N$, and $D'$
is a firing domain reached after firing $t_f$ in some interval $[b_r,b_{r+1}]$ where $b_r,b_{r+1}$ are consecutive values in $B^{t_f}_{M.N,D}$.
\end{definition}

Given $C$ and a firable transition $t_f$, we can compute the set $\mathsf{Post}(C, t_f)$ of successors of $C$, i.e. $\mathsf{Post}(C, t_f) := \{(M'.N',\mathsf{next}_r(D, t_f)) \,|\, b_r \in B^{t_f}_{M.N,D} \cup \{0\} \}$. The next marking is the same for every successor and is $M'.N' = M.N - {^\bullet}\!t_f + t_f^\bullet$. We then compute $\mathsf{next}_r(D, t_f)$ as follows:

\paragraph{1) Time progress:} We successively progress time from $D$ to bounds $b_1 < b_2 <\dots <b_r$ to eliminate variables of all enabled transitions reaching their upper bounds, up to bound $b_r$. We call $D^r$ the domain obtained this way. Every transition $t_k$ in $Enabled(M.N)$ that has no variable $\theta_k$ in $D^r$ is hence a waiting transition whose upper bound has been reached.
		
\paragraph{2) Firing condition:}
					We add to $D^r$ the following constraints: we add the inequality $(b_r \leq \theta_f \leq b_{r+1})$, and for every transition $t_j \in \mathsf{FullyEnabled}(M) \setminus \{ t_f\}$, we add to $D^r$ the inequality $\theta_f \leq \theta_j$. This means that no other transition was urgent when $t_f$ has been fired. Let $D^u$ be the new firing domain obtained this way. If any fully enabled transition $t_j$ has to fire before $t_f$, then we have a constraint of the form $a_j \leq \theta_j \leq b_j$ with $b_j < a_f$, and $D^u$ is not satisfiable. As we know that $t_f$ is firable, this cannot be the case, and $D^u$ has a solution, but yet, we have to include in the computation of the next firing domains reached after firing of $t_f$ the constraints on
%possible value of
$\theta_f$ due to urgency of other transitions.
					
\paragraph{3) Substitution of variables:}
					As $t_f$ fires after elapsing $\theta_f$ time units, the time to fire of other transitions whose clocks did not yet exceed their upper bounds decreases by the same amount of time. Variables of timed-out transitions have already been eliminated in $D^u$. So for every $t_j\neq t_f$ that has an associated constraint $a_j \leq \theta_j \leq b_j$ we do a variable substitution reflecting the fact that the new time to fire $\theta'_j$ decreases w.r.t. the former time to fire $\theta_j$. We set $\theta_j := \theta_f + \theta'_j$.
          When this is done, we obtain a domain $D'^{u,b_r}$ over a set of variables $\theta'_{i_1},\dots \theta'_{i_k}$, reflecting constraints on the possible remaining times to upper bounds of all enabled transitions that did not timeout yet.
%If this domain is not empty, then there firing transition $t_f$ after letting a delay $\theta_f \in [b_r,b_{r+1}]$ elapse is allowed.

\paragraph{4) Variable Elimination:}
					As $t_f$ fired at time $\theta_f$, it introduced new relationships between remaining firing times of other transitions, i.e other $\theta_i'\neq \theta_f$, that must be preserved in the next state class. However, as $t_f$ is fired and changed the marking of the waiting net, in the next class, either $t_f$  is a newly enabled transition, or it is not enabled. In the first case, the new domain will contain a fresh variable attached to the new occurrence of $t_f$ enabled by the firing. In the second case, variable $\theta_f$ should not appear in the new domain.
We hence need to remove the "old" variable $\theta_f$ from inequalities, while preserving an equivalent set of constraints, before introduction of constraints associated with newly enabled transitions. This is achieved by elimination of variable $\theta_f$ from $D'^{u,b_r}$, for instance with the well known Fourier-Motzkin technique (see Appendix~\ref{appendix-Fourier} for details).
					We proceed similarly with variable $\theta'_i$ for every transition $t_i$ that is enabled in marking $M.N$ but not in $M.N - \preset{t_f}$. After elimination, we obtain a domain $D'^{E,b_r}$ over remaining variables.

					%\item {\color{red} \textbf{Adding constraints transitions becoming fully enabled:} For every transition $t_k$ that is fully enabled in $M'.N'$, but not in $M.N$, and that has already reached its upper bound, we add a constraint of the form $\theta_k=0$. This transition will then be urgent for the next steps.}
					
\paragraph{5) Addition of new constraints :}
          The last step to compute the next state classes is to introduce fresh constraints for firing times of newly enabled transitions. For every
          $t_i \in \uparrow \mathsf{enabled}(M.N, t_f)$ we add to $D'^{E,b_r}$ the constraint $\alpha(t_i) \leq \theta'_i \leq \beta(t_i)$.
					For every timed-out transition $t_k$ that becomes fully enabled, we add to $D'^{E,b_r}$ the constraint $\theta_k=0$. Timed-out transitions that become fully enabled are hence urgent in the next class.
					After adding all constraints associated to newly enabled transitions, we obtain a domain, in which we can rename every $\theta'_i$ to
$\theta_i$ to get a domain $D'^{F,b_r}$. Notice that this domain needs not be minimal, so we do a last normalization step (see Definition \ref{def-canonical-form}) to obtain a final canonical domain $\mathsf{next}_r(D, t_f) = D'^{F,b_r\:*}$.

\medskip
Let us compare the state class construction in TPNs and in waiting nets. In TPNs (see~\cite{BerthomieuD91,LimeR06}), for a domain $D$ and a transition $t_f$, there is a single successor domain in $\mathsf{Post}(D,t_f)$, and hence a state class can have up to $|T|$ successors.
In waiting nets, for a given firable transition $t_f$, $\mathsf{Post}(D,t_f)$ contains one domain per bound in $B^{t_f}_{M.N,D}$, i.e. $\mathsf{Post}(D,t_f)$ is not a singleton anymore.
%It is hence clear that a state class can have more than one successor, with different markings and domains.
Now, if a waiting net has no control place, there is no bound on waiting transitions to consider, as enabled transitions are fully enabled in every configuration. Step 1 of successor construction leaves the starting domain $D$ unchanged, and consequently the state class construction shown above is exactly the standard construction for TPNs. The number of successors for a state class in TPNs is exactly the number of enabled transitions.
	Let $\mathsf{Post}(C)$ be the set of successors of a class $C$. Then $|\mathsf{Post}(C)| \leq |\mathsf{enabled}(M.N)|^2$, as the number of bounds to consider for a transition cannot exceed the number of enabled transitions.
Computing successors can be repeated from each class in $Post(C)$.
For a given net $\mathcal W$, and a given marking $M_0.N_0$, we denote by $\creach{W}$ the set of classes that can be built inductively. This set need not be finite, but we show next that this is due to markings, and that the set of domains appearing in state classes is finite.

\begin{definition}{(State Class Graph)}
\label{def-stateclassgraph}
The {\em State Class Graph} of a waiting net $\mathcal{W}$ is  a graph $SCG(\mathcal{W})= (\creach{W}, C_0, \longrightarrow)$ where %$C = \mathbb{N}^{P \cup C}\times \mathbb{R}^T$,
$C_0 = (M_0.N_0, D_0)$, and $C \overset{t}\longrightarrow C'$ iff $C' \in Post(C,t)$.
\end{definition}

Let $\rho=(M_0.N_0,v_0) \overset{d_1}\longrightarrow (M_0.N_0,v_0\oplus d_1) \overset{t_1}\longrightarrow (M_1.N_1,v_1)\dots (M_k.N_k,v_k)$ be a run of $\mathcal W$ and
$\pi = (M'_0.N'_0,D_0).(M'_1.N'_1,D_1) \dots (M'_k.N'_k,D_k)$ be a path in $SCG(\mathcal{W})$. We will say that $\rho$ and $\pi$ {\em coincide} iff $\forall i\in 1..k, M_i.N_i=M'_i.N'_i$, and for every step  $(M_i.N_i,v_i) \overset{d_i}\longrightarrow (M_i.N_i,v_i \oplus d_i) \overset{t_i}\longrightarrow (M_{i+1}.N_{i+1},v_{i+1})$, there exists an interval $[b_r,b_{r+1}]$ such that $d_i\in [b_r,b_{r+1}]$ and
$D_{i+1} = \mathsf{next}_r(D_i, t_i)$.
%% in class $(M_{i+1}.N_{i+1},D_{i+1})$

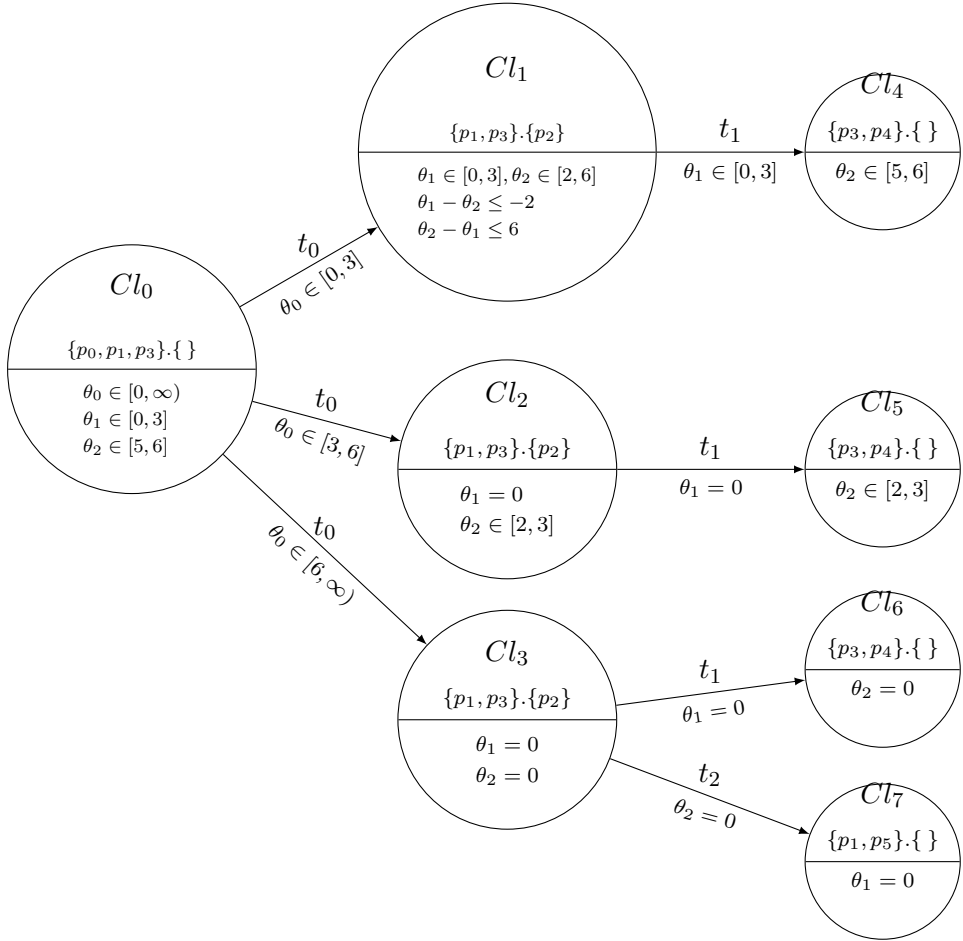
\begin{figure}[ht]
\vspace*{3mm}
\begin{center}
\begin{tikzpicture}
\begin{scope}[scale=1.1]
\node [draw,circle split,minimum size=0.7cm,scale=0.9] (C0) at (0,0)
{\scriptsize$\{ p_0,p_1,p_3\}.\{~\}$
\nodepart{lower} \scriptsize $\begin{array}{l}\theta_0\in[0,\infty)\\ \theta_1\in[0,3]\\ \theta_2\in[5,6]\end{array}$};

\node [] (lp0) at (0,1) {$Cl_0$};

\node [draw,circle split,minimum size=2cm,scale=0.9] (C1) at (4.5,2.6)
{\scriptsize$\{p_1,p_3\}.\{p_2\}$
\nodepart{lower}
\scriptsize$\begin{array}{l}
\theta_1\in[0,3], \theta_2\in[2,6]\\
\theta_1 - \theta_2 \leq -2\\
\theta_2 - \theta_1 \leq 6
\end{array}$
 };
\node [] (lc1) at (4.5,3.6) {$Cl_1$};

\node [draw,circle split,minimum size=0.7cm] (C2) at (4.5,-1.2)
{\scriptsize$\{p_1,p_3\}.\{p_2\}$
\nodepart{lower}\scriptsize$\begin{array}{l}\theta_1=0 \\ \theta_2\in [2,3]\end{array}$};

\node [] (lp2) at (4.5,-0.3) {$Cl_2$};

\node [draw,circle split,minimum size=0.7cm] (C3) at (4.5,-4.2) {\scriptsize$\{p_1,p_3\}.\{p_2\}$
\nodepart{lower}\scriptsize$\begin{array}{l}\theta_1=0 \\ \theta_2=0\end{array}$};
\node [] (lp3) at (4.5,-3.4) {$Cl_3$};

\node [draw,circle split,minimum size=0.7cm] (C4) at (9,2.6)
{\scriptsize$\{p_3,p_4\}.\{~\}$
\nodepart{lower}
\scriptsize$\theta_2\in[5,6]$};
\node [] (lc4) at (9,3.4) {$Cl_4$};

\node [draw,circle split,minimum size=0.7cm] (C5) at (9,-1.2)
{\scriptsize$\{p_3,p_4\}.\{~\}$
\nodepart{lower}
\scriptsize$\theta_2\in[2,3]$};
\node [] (lc5) at (9,-0.4) {$Cl_5$};

\node [draw,circle split,minimum size=0.7cm] (C6) at (9,-3.6)
{\scriptsize$\{p_3,p_4\}.\{~\}$
\nodepart{lower}
\scriptsize$\theta_2=0$};
\node [] (lc6) at (9,-2.8) {$Cl_6$};

\node [draw,circle split,minimum size=0.7cm] (C7) at (9,-5.9)
{\scriptsize$\{p_1,p_5\}.\{~\}$
\nodepart{lower}
\scriptsize$\theta_1=0$};
\node [] (lc5) at (9,-5.1) {$Cl_7$};

\draw [-latex] (C0) edge node [above] {$t_0$} node [below, rotate=35]{\scriptsize$\theta_0\in[0,3]$} (C1){};
\draw [-latex] (C0) edge node [above] {$t_0$} node [below, rotate=-20]{\scriptsize$\theta_0\in[3,6]$} (C2){};
\draw [-latex] (C0) edge node [above] {$t_0$} node [below, rotate=-42]{\scriptsize$\theta_0\in[6,\infty)$} (C3){};
\draw [-latex] (C1) edge node [above] {$t_1$} node [below]{\scriptsize$\theta_1\in[0,3]$}(C4){};
\draw [-latex] (C2) edge node [above] {$t_1$} node [below]{\scriptsize$\theta_1=0$}(C5){};
\draw [-latex] (C3) edge node [above] {$t_1$} node [below,rotate=8]{\scriptsize$\theta_1=0$}(C6){};
\draw [-latex] (C3) edge node [above] {$t_2$} node [below,rotate=-15]{\scriptsize$\theta_2=0$}(C7){};

\end{scope}
\end{tikzpicture}\vspace*{1mm}
\caption{A non-deterministic state class graph for the waiting net of Figure~\ref{FigureInfluenceBounds}.}
\label{FigureNonDetStateClass}
\end{center}\vspace*{-4mm}
\end{figure}

\vspace{\medskipamount}
Let us illustrate the construction of state classes with the example of Figure~\ref{FigureInfluenceBounds}.
We start from an inital state class $Cl_0= (M_0.N_0, D_0)$ with $M_0.N_0=\{p_0,p_1,p_3\}\{\}$ and $D_0$ is given by the set of constraints $0 \leq \theta_0, 0 \leq \theta_1\leq 3, 5 \leq \theta_2 \leq 6$. From the initial configuration $(M_0.N_0,v_0)$ one can only fire transition $t_0$, that is, wait for some duration $\theta_0$ and then perform a discrete move with transition $t_0$. So $t_0$ is also the only firable transition from state class $Cl_0$. Here, we have two important bounds, $3$ and $6$ which are the upper bounds attached to $t_1$ and $t_2$ respectively.
Let us assume that $t_0$ fires after a duration in $[0,3]$. We hence fire from a domain that satisfies the following constraints:

$$D'_0=\left \{
\begin{array}{l}
0 \leq \theta_0 \leq 3\\
0 \leq \theta_1 \leq 3 \\
5 \leq \theta_2 \leq 6 \\
\end{array}
\right.
$$

One can easily find values for $\theta_0,\theta_1,\theta_2$ satisfying $D'_0$ and such that $\theta_0 \leq \theta_1$ and $\theta_0 \leq \theta_2$. So $t_0$ is firable from $Cl_0$.
Let $\theta'_1,\theta'_2$ be variables describing the remaining times to fire after firing $t_0$. As $t_0$ is the fired transition, all remaining (positive) times are decreased by $\theta_0$, i.e. $\theta_1'=\theta_1 -\theta_0$ and $\theta_2'=\theta_2 -\theta_0$. This gives a variable change allowing to obtain the following domain:
$$
\left \{ \begin{array}{l}
0 \leq \theta_0 \leq 3\\
0 \leq \theta_1', 0 \leq \theta_2'\\
0 \leq \theta_1' + \theta_0 \leq 3 \\
5 \leq \theta_2' + \theta_0 \leq 6 \\
\end{array}
\right.
\mbox{~or equivalently~}
\left\{
\begin{array}{l}
0 \leq \theta_0 \leq 3\\
0 \leq \theta_1', 0 \leq \theta_2'\\
0 - \theta_1' \leq   \theta_0 \leq 3 - \theta_1' \\
5 - \theta_2' \leq   \theta_0 \leq 6 - \theta_2' \\
\end{array}
\right.
$$

The variable to eliminate is now $\theta_0$, we can hence reorganize lines containing $\theta_0$ in this domain as
$$ \max (0, 0-\theta_1',5-\theta_2') \leq \theta_0 \leq \min(3, 3 - \theta_1', 6 - \theta_2')$$

We hence obtain, after elimination of $\theta_0$
$$ \left\{
\begin{array}{lcl}
0 \leq 3 & (tautology)\\
0 \leq \theta_1', 0 \leq \theta_2'\\
0 \leq 3 - \theta_1' &\Longleftrightarrow &\theta_1' \leq 3 \\
0 \leq 6 - \theta_2' &\Longleftrightarrow &\theta_2' \leq 6 \\
0-\theta_1' \leq 3 &\Longleftrightarrow   & -3 \leq \theta_1'\\
0-\theta_1' \leq 3 - \theta_1'  &(tautology) \\
0-\theta_1' \leq 6 - \theta_2'\\
5-\theta_2'  \leq 3 &\Longleftrightarrow & 2 \leq \theta_2'\\
5-\theta_2' \leq  3 - \theta_1'\\
5-\theta_2' \leq  6 - \theta_2' &(tautology)\\
\end{array}
\right.
$$

After elimination of redundancies and tautologies, and renaming $\theta'_1,\theta'_2$ into $\theta_1,\theta_2$ we obtain a domain $D_1$ with the following constraints:
$$
D_1 = \left \{\begin{array}{l}
0 \leq \theta_1 \leq 3\\
2 \leq \theta_2 \leq 6\\
\theta_2 - \theta_1 \leq 6 \\
\theta_1-\theta_2 \leq  -2 \\
\end{array}\right.
$$

The state class reached by firing $t_0$ after elapsing $0$ to $3$ time units is $Cl_1= (M_1.N_1,D_1)$ with $M_1=\{p_1,p_3\}$ and $N_1= \{p_2\}$. One can notice that, as transitions $t_1,t_2$ remain enabled after firing of $t_0$, variable elimination adds diagonal constraints of the form
$\theta_2 - \theta_1 \leq 6 $ and $\theta_1 - \theta_2 \leq -2 $. This is the case in general for all pairs of transitions that remain enabled. Domains are hence not limited to interval contraints of the form $a_i \leq \theta_i \leq b_i$.
Similarly, starting from $Cl_0$ one can let between $3$ and $6$ time units elapse and reach ctate class $Cl_2=(M_2.N_2,D_2)$.
$D_2$ is easily obtained by progressing $D_0$ of $3$ time units (this is a variable change replacing $\theta_i$ by $\theta'_i+3$), and projecting away constraints containing $\theta'_1$. Then, firing $t_0$ gives a pair of constraints  $0 \leq \theta_0' \leq 3$ and $2 \leq \theta_2' + \theta_0' \leq 3$. We let the reader check that this reduces to  $2 \leq \theta_2 \leq 3$ after elimination of  $\theta_0$. Then by adding constraints on $\theta_1$ arising from the fact that after firing of $t_0$, $t_1$ becomes fully enabled, we obtain :
$$D_2 = \left\{ \begin{array}{l}
\theta_1=0\\
2 \leq \theta_2 \leq 3\\
\end{array}
\right. $$

One can immediately notice that from any configuration in class $Cl_2$, transition $t_1$ is urgent, and is the the only transition that can fire. We let the readers build the remaining classes, and give the complete state class graph for the waiting net of figure~\ref{FigureInfluenceBounds} in Figure~\ref{FigureNonDetStateClass}. One can observe that this state class graph {\em is not} deterministic.

\begin{proposition}[Completeness]
\label{prop-stateclass-wnet-Compl}
%The language of state class graph defined is same as the language of underlying Waiting net.
For every run $\rho=(M_0.N_0,v_0) \dots (M_k.N_k,v_k)$ of $\mathcal{W}$ there exists a path $\pi$ of $SCG(\mathcal{W})$
such that $\rho$ and $\pi$ coincide.
\end{proposition}

\begin{proof}
We only give a proof sketch, a complete proof of this proposition is provided in Appendix~\ref{Appendix_proof_reach}.
We proceed by induction on the length of runs. For the base case, we can easily prove
that any transition firing from the initial configuration after some delay $d$ gives a possible solution for $D_0$ and a successor class, as $D_0$ does not contain constraints of the form $\theta_i -  \theta_j \leq c_{ij}$ .
The induction step is similar, and slightly more involved, because domains contain constraints involving pairs of variables. However, we can show (Lemma~\ref{lem_durations_match_abstract_runs} in~Appendix~\ref{Appendix_proof_reach}) that along run $\rho$ for every pair of steps composed of a time elapsing of duration $d_i$ followed by the firing of a transition $t_f$, we have $d_i\in [a_{i,f},b_{i,f}]$, where $a_{i,f}$ is the lower and $b_{i,f}$ the upper bound on variable $\theta_f$ at step $i$ of the run. Hence, for every run of $\mathcal{W}$ there is a path that visits the same markings and maintains consistent constraints.
\end{proof}

\begin{proposition}[Soundness]
\label{prop_SCG_sound}
Let $\pi$ be a path of $SCG(\mathcal{W})$. Then there exists a run $\rho$ of $\wnet$ such that $\rho$ and $\pi$ coincide.
\end{proposition}
\begin{proof}
As for Proposition~\ref{prop-stateclass-wnet-Compl}, we can prove this property by induction on the length of paths in $SCG(\mathcal{W})$. A complete proof of this proposition is provided in Appendix~\ref{Appendix_proof_reach}.
\end{proof}

Proposition~\ref{prop-stateclass-wnet-Compl} shows by induction on runs length that every marking reached by a run of a waiting net appears in its state class graph. The proof of Proposition~\ref{prop_SCG_sound} uses a similar induction on paths length, and shows that we do not introduce new markings. These propositions show that the state class graph is a sound and complete abstraction, even for unbounded nets. We can show a stronger property, which is that the set of domains appearing in a state class graph is finite.

\begin{proposition}
\label{prop-fin-domain}
The set of firing domains in $SCG(\mathcal{W})$ is finite.
\end{proposition}

\begin{proof}
We only give a proof sketch, and provide a complete proof in Appendix~\ref{Appendix_proof_reach}.
Domains are of the form $\{ a_i \leq \theta_i \leq b_i\}_{t_i\subseteq T} \cup \{ \theta_i -\theta_j \leq c_{i,j}\}_{t_i,t_j\subseteq T}$. We can easily adapt proofs of~\cite{BerthomieuD91} (lemma 3 page 9) to show that every domain generated during the construction of the SCG has inequalities of the form $a_i \leq \theta_i \leq b_i$ and $\theta_i - \theta_j \leq c_{ij}$, where $0\leq a_i \leq \alpha(t_i)$, $0\leq b_i \leq \beta(t_i)$ and $-\alpha(t_i)\leq  c_{ij} \leq \beta(t_i)$ and $\alpha_i,\beta(t_i)$ are the constant bounds of interval attached to $t_i$.
This does not yet prove that the set of domains is finite. We then show that domains are {\em bounded}, i.e. that $\theta_i$'s and their differences have rational upper and lower bounds (see Lemma~\ref{lem-boundedness} in Appendix).
We then show that domains are also {\em linear}, i.e. that constants appearing in inequalities are linear combinations of a finite set of constant values. Domain $D_0$ is bounded and linear, and a series of technical lemmas (given in Appendix~\ref{Appendix_proof_reach}) show that variable elimination, reduction to a canonical form, etc. preserve bounds and linearity (a similar result was shown in~\cite{BerthomieuD91} for domains of TPNs). The set of bounded linear domains between fixed bounds is finite, so the set of domains of a waiting net is finite.
\end{proof}

Finiteness of the set of domains of waiting nets is an important property, that is not a straightforward consequence of the syntactic proximity with TPNs nor of the effectiveness of successors construction. Indeed, waiting nets allow to stop time measurement for waiting transitions reaching their upper bounds. Petri nets with stopwatches, where time measurement can also be stopped, do not have a finite state class representation, even for bounded nets. The reason is that clock differences in domains of stopwatch Petri nets can take any value. Waiting nets avoid this issue because clocks are always stopped at a predetermined instant (when they reach the upper bound of an interval).

\begin{corollary}
\label{Cor-stateclass-finite}
If $\wnet$ is a bounded waiting net then $SCG(\wnet)$ is finite.
\end{corollary}
\begin{proof}
States of $SCG(\mathcal{W})$ are of the form $(M.N,D)$ where $M.N$ is a marking and $D$ a domain for times to fire of enabled transitions. By definition of boundedness, the number of markings in $SCG(\mathcal{W})$ is finite. By Prop.~\ref{prop-fin-domain}, the set of domains appearing in $SCG(\wnet)$ is finite, so $SCG(\mathcal{W})$ is finite.
\end{proof}

\vspace{-\medskipamount}
More precisely, if a net is $k_M$-bounded, there are at most $k_M^{|P|}$ possible markings, and the number of possible domains is bounded by  $(2\cdot K_\wnet+1)^{|T+1|^2}$, where $K_\wnet = \max_{i,j} \lfloor \frac{\beta_i}{\alpha_j} \rfloor$ is an upper bound on the number of linear combinations of bounds appearing in domains. Hence the size of $SCG(\wnet)$ is in $O(k_M^{|P|}\cdot (2\cdot K_\wnet+1)^{|T+1|^2})$.
A direct consequence of Proposition~\ref{prop-stateclass-wnet-Compl}, Proposition~\ref{prop_SCG_sound}, and Corollary~\ref{Cor-stateclass-finite} is that many properties of bounded waiting nets are decidable.

\begin{corollary}[Reachability and Coverability]
\label{Cor_reachability_coverability}
The reachability and coverability problems for bounded waiting nets are decidable and PSPACE-complete.
\end{corollary}

\vspace{-\medskipamount}
\begin{proof}
For PSPACE membership, given a target marking $M_t.N_t$ it suffices to explore non-deterministically all runs starting from $(M_0.N_0,D_0)$ of length at most $|SCG(\wnet)|$ to find marking $M_t.N_t$ and prove reachability, or a marking that covers $M_t.N_t$ witnessing coverability. Random walks of bounded length are feasible in NLOGSPACE w.r.t. the size of the explored graph, whence the NPSPACE=PSPACE complexity.
For hardness, we already know that reachability for $1$-safe Petri nets is PSPACE-Complete~\cite{ChengEP95}, and a (bounded) Petri net is a (bounded) waiting net without control places and with $[0,\infty)$ constraints. Similarly, given $1$-safe Petri net and a place $p$, deciding if a marking with $M(p)=1$ (which is a coverability question) is reachable is PSPACE-complete~\cite{Esparza96}. This question can be recast as a coverability question for waiting nets, thus proving the hardness of coverability.\vspace*{-5mm}
\end{proof}
%%%%%%%%%%%%%%%%%%%%%%%%%%%%%%%%%%%%%%%%%%%%%%%%%%%

%\vspace{-\medskipamount}
Let us illustrate the advantages and expressive power of waiting nets with the example of Figure~\ref{fig_loop} containing a cyclic behavior. This net contains a control part (in Gray), composed of standard places $p_2,p_3$, of control places $c_3,c_4$, and of two transitions $tc_3$ and $tc_4$ with respective time intervals $[2,2]$ and $[1,1]$. The main idea in this control part is to move a token from place $c_3$ to place $c_4$ after 2 time units, and then move the token from $c_4$ to $c_3$ after 1 time unit. Control places $c_3$ and $c_4$ are used to allow/forbid firing of transitions $t_1$ and $t_2$ in the right part of the net. Transitions $t_1$ and $t_2$ are independent transitions with respective time intervals $[1,2]$ and $[2,3]$ and which firings are respectively allowed when place $c_3$ and $c_4$ are filled. Note that allowing or forbidding firings of transitions does not affect their enabledness: while transition $t_2$ is waiting for a token in $c_4$ to fire, one measures the time elapsed since it was last enabled (and similarly when transition $t_1$ is waiting for a token in $c_3$).

\begin{figure}[ht]
\begin{center}
\scalebox{0.9}{
\begin{tikzpicture}
%%%%%%%%%%%%%%%%%%%%%%%%%%%%%%%%%%%%%%%%%%%%%%%%%%%%%%%%%%%%%%%%%%%%%%%%
\begin{scope}

\node[rounded corners,draw,fill=gray!10,rectangle,dotted, minimum width=5.5cm, minimum height=4.5cm] (ctrl) at (0.6,-1.5) {};

\node [draw,circle,minimum size=0.7cm] (p2) at (-1,0) {$\bullet$};
\node [] (lp0) at (-1.5,0.5) {$p_2$};

\node [draw,rectangle,fill=black,minimum width=1cm] (tc3) at (-1,-1.5){};
\node [] (lt0) at (-1.8,-1.5) {\scriptsize$tc_3$};
\node [color=red] (it0) at (-0.1,-1.5) {\scriptsize$[2,2]$};

\node [draw,circle,minimum size=0.7cm] (p3) at (-1,-3) {};
\node [] (lp0) at (-0.5,-3.5) {$p_3$};

\node [draw,dashed, circle,minimum size=0.7cm] (c3) at (2,0) {$\bullet$};
\node [] (lp0) at (1.5,0.5) {$c_3$};

\node [draw,rectangle,fill=black,minimum width=1cm] (tc4) at (2,-1.5){};
\node [] (lt0) at (1.2,-1.5) {\scriptsize$tc_4$};
\node [color=red] (it0) at (2.9,-1.5) {\scriptsize$[1,1]$};

\node [draw,dashed, circle,minimum size=0.7cm] (c4) at (2,-3) {};
\node [] (lp0) at (1.5,-3.5) {$c_4$};

\draw [-latex] (c3) edge (tc3){};
\draw [-latex] (p2) edge (tc3){};
\draw [-latex] (tc3) edge (p3){};
\draw [-latex, bend right] (tc3) edge (c4){};

\draw [-latex, bend right] (tc4) edge (p2){};
\draw [-latex] (p3) edge (tc4){};
\draw [-latex] (tc4) edge (c3){};
\draw [-latex] (c4) edge (tc4){};

\end{scope}
\begin{scope}[xshift=0.7cm]
%%%%%%%%%%%%%%%%%%%%%%%%%%%%%%%%%%%%%%%%%%%%%%%%%%%
% Right part
%%%%%%%%%%%%%%%%%%%%%%%%%%%%%%%%%%%%%%%%%%%%%%%%%%%
\node [draw,circle,minimum size=0.7cm] (p0) at (3,1) {$\bullet$};
\node [] (lp0) at (2.3,1.1) {$p_0$};

\node [draw,rectangle,fill=black,minimum height=1cm] (t1) at (4,0){};
\node [] (lt1) at (4,-0.7) {\scriptsize$t_1$};
\node [color=red] (it1) at (4,0.7) {\scriptsize$[1,2]$};

\draw [-latex] (p0) edge (t1){};
\draw [-latex] (c3) edge (t1){};

\path [-latex, draw, bend right, rounded corners] (t1) ..controls (5,0.5) and (4.5,1.5) .. (p0){};
\path [-latex, draw, bend right, rounded corners] (t1) ..controls (5,-0.5) and (4.5,-1.6) .. (c3){};

\node [draw,circle,minimum size=0.7cm] (p1) at (3,-4) {$\bullet$};
\node [] (lp1) at (2.3,-4.1) {$p_1$};

\node [draw,rectangle,fill=black,minimum height=1cm] (t2) at (4,-3){};
\node [] (lt2) at (4,-3.7) {\scriptsize$t_2$};
\node [color=red] (it2) at (4,-2.3) {\scriptsize$[2,3]$};

\draw [-latex] (p1) edge (t2){};
\draw [-latex] (c4) edge (t2){};

\path [-latex, draw, bend right, rounded corners] (t2) ..controls (5,-3.5) and (4.5,-4.5) .. (p1){};
\path [-latex, draw, bend right, rounded corners] (t2) ..controls (5,-2) and (4,-1.5) .. (c4){};
\end{scope}
\end{tikzpicture}
}
\vspace{-\medskipamount}
\caption{A cyclic Waiting net}\vspace*{-1mm}
\label{fig_loop}
\end{center}
\end{figure}
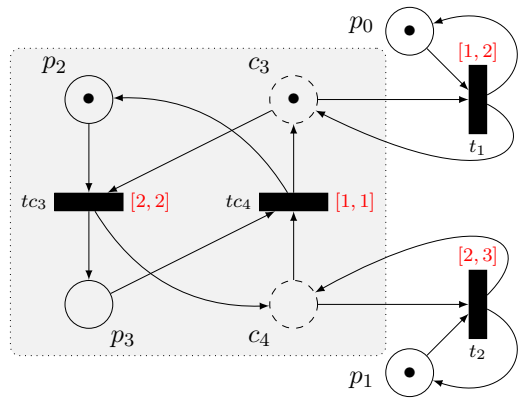

The main idea of this example is to give control alternatively to transition $t_1$ and $t_2$, but measure enabled time for $t_2$ when control allows firing of $t_1$ and enabled time for $t_1$ when control allows firing of $t_2$. The timed word $(t_1,1)(t_1,2)(tc_3,3)(t_2,3.2)$ is a possible behavior of this net. Notice that transition $t_2$ becomes fully enabled after firing of $t_3$, but was enabled from the beginning of the run.
In a symmetric way, the timed word $(t_1,1)(t_1,2)(tc_3,3)(tc_4,4)(t_1,4)(t_1,5)(t_1,6)$ is a legal behavior of the model. Indeed, after firing of $tc_3$ time still elapses for $t_1$, which can fire immediately after firing of $tc_4$. These two example behaviors show that one can design controllers that decouple control and time measurement with waiting nets.
The state class computed for this example contains 44 state classes, and show that one can fire up to 5 occurrences of $t_1$ before an occurrence of $t_2$ fires urgently. Notice that if one replaces control places by standard places in this example to obtain a TPN, then transition $t_2$ never fires, as it is disabled at most 1 time unit after its enabling by firing of $t_4$. This example shows that waiting nets can be a tool to design easily fair controllers.

\section{Expressiveness}
\label{sec:expressiveness}

The next natural question for waiting nets is their expressiveness w.r.t. other timed models. There are several ways to compare expressiveness of timed models: One can build on relations between models such as isomorphism of their underlying timed transition systems, timed similarity, or bisimilarity. In the rest of this section, we compare models w.r.t. the timed languages they generate. For two particular types of model $\mathcal M_1$ and $\mathcal M_2$, we will write $ \mathcal M_1 \lessexp \mathcal M_2$ when, for every model $X_1\in \mathcal M_1$, there exists a model $X_2$ in $\mathcal M_2$ such that $\lang(X_1) = \lang(X_2)$. Similarly, we will write $\mathcal M_1 \strictlessexp \mathcal M_2$ if $\mathcal M_1 \lessexp \mathcal M_2$ and there exists a model $X_2 \in \mathcal M_2$ such that for every model $X_1 \in \mathcal M_1$, $\lang(X_2) \neq \lang(X_1)$.
%(the timed language of $X_2$ cannot be encoded  by a model from $\mathcal M_1$).
Lastly, we will say that $\mathcal M_1 $ and $\mathcal M_2$ are equally expressive and write $\mathcal M_1 \equlang \mathcal M_2$ if $\mathcal M_1 \!\lessexp \!\mathcal M_2$ and $\mathcal M_1 \!\lessexp \!\mathcal M_2$.
In the rest of this section, we compare bounded and unbounded waiting nets with injective/non-injective labelling, with or without silent transitions labelled by $\epsilon$ to timed automata, TPNs, Stopwatch automata, and TPNs with stopwatches.

We first have obvious results. It is worth nothing that every model with non-injective labeling is more expressive than its injective counterpart. Similarly, every unbounded model is strictly more expressive than its bounded subclass.
Waiting nets can express any behavior specified with TPNs. Indeed, a WTPN without control place is a TPN. One can also remark that (unbounded) TPNs, and hence WTPNs are not regular. It is also well known that the timed language of a bounded TPN can be encoded by a time bisimilar timed automaton~\cite{CassezR06,LimeR06,BCHLR13}.
We show next that one can extend the results of~\cite{LimeR06} to waiting nets, i.e. reuse the state class construction of section~\ref{sec:reachability} to build a finite timed automaton $\ta_\wnet$ that  recognizes the same language as a waiting net $\wnet$.
%
%\subsection{A State class Timed automaton}
%
As shown by Proposition~\ref{prop-stateclass-wnet-Compl} and Proposition~\ref{prop_SCG_sound}, the state class graph $SCG(\wnet)$ is sound and complete. State class graphs abstract away the exact values of clocks and only remember constraints on remaining times to fire of fully enabled transitions, and times to upper bounds for enabled transitions. If we label moves of $SCG(\wnet)$ by the label of the transition used to move from a state class to the next one, the state class graph becomes an automaton that recognizes the untimed language of $\wnet$. Further, one can add guards on clock values to moves and clock invariants to states of  a state class graph to recover the timing information lost during abstraction.

%Now, we see a simple extension of the $SCG$ by adding clocks so that it can represent the actual timed language.

\begin{definition}[Extended State Class]
	An {\em extended state class} is a tuple $C_{ex} = (M.N,D, \chi, trans,XP)$, where $M. N$ is a marking, $D$ a domain,
	%as in the definition of state class graphs,
	$\chi$ is a set of real-valued clocks, $trans \in (2^T)^\chi$ maps clocks to sets of transitions and $XP\subseteq T$ is a set of transitions which upper bound have already been reached.
\end{definition}	

Extended state classes were already proposed in~\cite{LimeR06} as a building step for state class timed automata recognizing languages of bounded TPNs. Here, we add information on transitions that have been enabled for a duration that is at least their upper bound. This is needed to enforce urgency when such transitions become firable. In extended state classes, every clock $x\in \chi$ represents the time since enabling of several transitions in $trans(x)$, that were enabled at the same instant. So, for a given transition $t$, the clock representing the valuation $v(x_t)$ is $trans^{-1}(t)$\footnote{An alternative is to associate a clock to every enabled transition, but using a map $trans$ as proposed in~\cite{LimeR06} allows to build automata with a reduced number of clocks.}.
Let $\mathbb C^{ex}$denote the set of all state classes. We can now define the state class timed automaton $\scta(\mathcal{W})$ by adding guards and resets  to the transitions of the state class graph, and invariants to state classes.

\begin{definition}[State Class Timed Automaton]
\label{def-scta}
	The state class timed automaton of $\wnet$ is a tuple $\scta(\wnet)= (L,l_0, X, \Sigma,Inv, E,F)$ where:
	\begin{itemize}

		\item $L\subseteq \mathbb C^{ex} $ is a set of extended state classes. $l_0 = (M_0.N_0,D_0,\{ x_0\}, trans_0,XP_0)$, where $trans_0(x_0)=\mathsf{Enabled}(M_0.N_0)$ and $XP_0=\emptyset$.
		
		\item $F=L$, i.e. every run is accepting,

		\item $\Sigma=\lab(T)$, and $X= \underset{(M.N,D,\chi,trans,XP)\in  L} \bigcup  \chi \subseteq   \{ x_1, \dots x_{|T|}\}$ is a set of clocks
		%which size is at most the number of transitions.
		
		\item $E$ is a set of transitions of the form $(C_{ex},\lab(t),g,R,C'_{ex})$. In each transition, the locations \\
		$C_{ex}=(M.N,D,\chi,trans,XP)$ and
		$C'_{ex}=(M'.N',D',\chi',trans',XP')$ are two extended state classes such that $(M'.N',D') \longrightarrow (M.N,D)$ is a move of the STG with $D'=next_r(D,t)$ for some index $r$ chosen among the upper bounds ordering computed from domain $D$.

		We can compute the set of transitions disabled by the firing of $t$ from $M.N$, denoted\\  $\mathsf{Disabled}(M.N,t)$
		and from there, compute a new set of clocks $\chi'$. We have $\chi'= \chi \setminus \{ x\in \chi \mid trans(x) \subseteq \mathsf{Disabled}(M.N,t) \}$ if firing $t$ does not enable new transitions. If new transitions are enabled or become fully enabled after reaching their upper bound, we have $\chi'= \chi \setminus \{ x\in \chi \mid trans(x) \subseteq \mathsf{Disabled}(M.N,t) \} \cup \{ x_i\}$, where $i$ is the smallest index for a clock in $X$ that is not used.	Intuitively, $x_i$ will memorize the time elapsed since the firing of $t$, and hence the time elapsed since new enabling of all transitions enabled by this firing instead of memorizing the same value for every transition.
		
		Similarly, we can set
		
		$trans'(x_k)= \!\left \{\! \begin{array}{l}
		\!\newenab{M.N, t} \cup \{ t \in XP \cap \mathsf{FullyEnabled}(M'.N')\} \mbox{~if~} x_k=x_i\\
		\!trans(x_k)\! \setminus \mathsf{Disabled}(M.N,t)\mbox{~if~} x_k \neq x_i \mbox{~and~}  trans(x_k) \nsubseteq \mathsf{Disabled}(M.N,t) \\

		\!\mbox{Undefined otherwise}
		\end{array}\right.$
		
		The set $XP$ lists the waiting transitions that have reached their upper bounds and were not fully enabled.
		$$\begin{array}{ll}
		XP' = XP &\cap \:\:\mathsf{Enabled}(M - \preset{t}) \setminus \mathsf{FullyEnabled}(M'.N')\\
	%	&\cup \{ t_k \in Enabled(M'.N') \mid 0 \leq \theta_k \leq 0 \in D'\}\\
		&\cup \:\: \{ t_k \in \mathsf{Enabled}(M'.N') \mid \theta_k \not\in D'\}\\
	
		\end{array}$$
		In words, $XP'$ contains the transitions already in $XP$ that are still enabled, and did not receive a token in the control places of their preset, plus enabled waiting transitions that reached their upper bound, and for which no constraint is kept in domain $D'$.

		The set of clocks reset in a transition is $R=\{ x_i\}$, where $x_i$ is the first unused clock, if some transition is newly enabled, or was in $XP$ and becomes fully enabled. $R=\emptyset$ otherwise.
		
The guard $g$ is set to $\alpha(t) \leq trans^{-1}(t)$ if there is no constraint of the form $trans^{-1}(t)\leq 0$ in the invariant of $C_{ex}$, that is if $t$ is not a transition that became urgent after timing out, and $g$ is set to true otherwise.

		\item Let us now consider invariants attached to states. Let $\mathsf{Urgent}(C_{ex})$ denote the set of fully enabled transitions of $M.N$ that are urgent. One can compute this set by looking at fully enabled transitions with a constraint of the form $\theta_k=0$ in domain $D$. We distinguish two cases. The first one is when no transition is urgent, i.e., $\mathsf{Urgent}(C_{ex})=\emptyset$. In this situation, we can let some time elapse in location $C_{ex}$ until the upper bound for some clock attached to a transition is reached. The invariant $Inv(C_{ex})$ is hence:
				$$Inv(C_{ex}) = \underset{\scriptsize\begin{array}{l}
		x_j \in trans^{-1}(\mathsf{FullyEnabled}(M'.N'),\\
		t_k \in trans(x_j) \cap \mathsf{FullyEnabled}(M'.N')\end{array} }\bigwedge \hspace{-2cm}x_j \leq \beta(t_k)$$

If $\mathsf{Urgent}(C_{ex}') \neq \emptyset$ the invariant $Inv(C_{ex})$ is set to:
$$Inv(C_{ex}) = \underset{x_k\in trans^{-1}\big(\mathsf{Urgent}(C_{ex})\big )}\bigwedge \hspace{-1cm}x_k \leq 0$$

The invariant in location $C_{ex}$ then guarantees that if $t$ is urgent in $C_{ex}$, then it is fired immediately without elapsing time, or another urgent transition is fired immediately.
		\end{itemize}
\end{definition}

\begin{proposition}
\label{prop_scta_is_correct}
Let $\wnet$ be a waiting net. Then $\lang(\scta(\wnet)) = \lang(\wnet)$.
\end{proposition}

\begin{proof}
Obviously, every sequence of transitions in $\lang(\scta(\wnet))$ is a sequence of transitions of the STG, and hence there exists a timed word that corresponds to this sequence of transitions. Furthermore, in this sequence, every urgent transition is fired in priority before elapsing time, and the delay between enabling and firing of a transition $t$ lays between the upper and lower bound of the time interval $[\alpha_t,\beta_t]$ if some time elapses in a state before the firing of $t$, and at least $\beta_t$ time units if $t$ fires immediately after reaching some state in the sequence (it is an urgent transition, so the upper bound of its interval has been reached, possibly some time before full enabling). Hence, every timed word of $\scta(\wnet)$ is also a timed word of $\wnet$.
We can reuse the technique of Prop.~\ref{prop-stateclass-wnet-Compl} and prove by induction on the length of runs of $\wnet$ that for every run of $\wnet$, there exists a run of  $\scta(\wnet)$ with the same sequence of delays and transitions.
Let us first consider a run of size 1, that is of the form $(M_0.N_0,v_0)\overset{d_0,t_0}\longrightarrow (M_1.N_1,v_1)$.
The domain $D_0$ is a conjunction of constraints of the form $\alpha_i \leq \theta_i \leq b_i$ for every enabled transition, and  firability of $t_0$ after a delay $d_0$ is checked on the projection of $D_0$ on transitions that are not timed-out in the bound interval $[b_r, b_{r+1}]$ that contains $d_0$, i.e. on a subset of the original set of constraint $D_0$. We have associated to location $l_0$ the invariant $Inv(l_0)=x_0 \leq \min_{t_k \in \mathsf{FullyEnabled}(M_0.N_0)} (\beta_k)$ (there is a single clock $x_0$ as all enabled transitions are enabled at date $0$). All transition labeled by $t_0$ leaving $l_0$ are attached a guard of the form $g = \alpha(t_0) \leq x_0$.
The move of $\wnet$ via transition $t_0$ after $d_0$ time units is allowed, so we necessarily have $v(x_{t_0})=d_0 \geq b_r \geq \alpha(t_0)$. As we have $b_r\leq \beta(t_0)$, and as $t_0$ is firable from $D_0$, we can also claim that $v_0+d_0 \models Inv(l_0)$ and $v_0+d_0\models g$. Hence $(d_0,t_0)$ is a timed word of $\lang(\scta(\wnet))$.

Let us now consider a sequence of moves of the form $(M_0.N_0,v_0) \overset{d_0,t_0} \longrightarrow \dots \overset{d_{k-1},t_{k-1}}\longrightarrow (M_k.N_k, v_k)$.
As the state class graph is complete, there exists a sequence of abstract moves $(M_0.N_0,D_0) \overset{t_0} \longrightarrow \dots \overset{t_{k-1}}\longrightarrow (M_k.N_k, D_k)$.
By induction hypothesis, there exists a sequence of delays and transitions in $\scta(\wnet)$ that leads to a configuration $(l_k,\mu_k)$ where $l_k=(M_k.N_k,D_k, \xi_k, trans_k)$ is an extended state class with invariant  $inv_k$  and $\mu_k$ is a valuation of clocks. Assume that a fully enabled transition $t_k$ is firable from $(M_k.N_k, v_k)$, and let $h$ be the index of move at which $t_k$ was enabled. Still using the completeness argument, this means that $t_k$ is firable from $(M_k.N_k, D_k)$. Hence, there exists a transition from $l_k=(M_k.N_k,D_k, \xi_k, trans_k)$ to some $l_{k+1}=(M_{k+1}.N_{k+1},D_{k+1}, \xi_{k+1}, trans_{k+1})$ labeled by transition $t_k$. Assume that $t_k$ is not a transition that has timeout. Then the move is attached with guard $g ::= \alpha(t_k) \leq trans^{-1}(t_k)$. We know that in configuration $(M_k.N_k, v_k)$ of $\wnet$,  $v_k(x_{t_k}) = d_k+ \underset{i\in h..k-1}\sum d_i \geq \alpha(t_k)$. As $t_k$ did not timeout, the clock $x_h=trans_k^{-1}(t_k)$ was never reset since the $h^{th}$ transition. So, $\mu_k = v_k(x_{t_k})$, and $t_k$ can fire from $(l_k,\mu_k)$ after $d_k$ time units. Now assume that $t_k$ is an urgent transition that became urgent after a timeout. Then, we have $d_k=0$, and the invariant $Inv(l_k) = trans_k^{-1}(t_k) \leq 0$ attached to $l_k$ does not let time elapse, and $t_k$ can fire, and cannot let a duration of more than $d_k$ time units elapse.
%HERE
\end{proof}

%\subsection{Comparison of time nets, waiting nets, and automata}

We are now ready to compare expressiveness of waiting nets and their variants w.r.t. other types of time Petri nets, and with timed automata. For a given class $\mathcal N$ of net, we will denote by $B$-$\mathcal N$ the bounded subclass of $\mathcal N$, add the subscript $\epsilon$ if transitions with $\epsilon$ labels are allowed in the model, and a superscript $\noninj$ if the labeling of transitions is non-injective. For instance $B$-$WTPN_\epsilon^{\noninj}$ denotes the class of bounded waiting nets with non-injective labeling and $\epsilon$ transitions. It is well known that adding $\epsilon$ moves to timed automata increases the expressive power of the model~\cite{DiekertGP97}. Similarly, allowing non-injective labeling of transitions or adding $\epsilon$ transitions increases the expressive power of nets~\cite{BCHLR13}. Lastly, adding stopwatches to timed automata or bounded time Petri nets makes them Turing powerful~\cite{CassezL00}.

\begin{theorem}
\label{prop_BWPN_stric_less_TA}
	$BWTPN \strictlessexp TA(\leq,\geq)$.
\end{theorem}

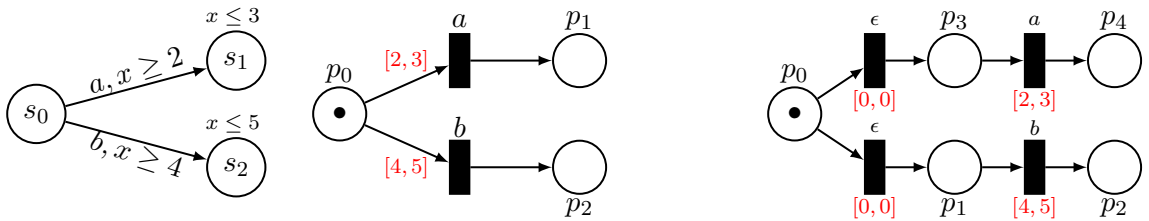
\begin{figure}[!b]
	\begin{center}
		\tikzset{every picture/.style={line width=0.75pt}} %set default line width to 0.75pt
		
		\begin{tikzpicture}[x=0.5pt,y=0.5pt,yscale=-1,xscale=1]
		\begin{scope}
		\node[circle, draw]  (s0) at (0,0) {$s_0$};

		\node[circle, draw]  (s1) at (150,-40) {$s_1$};
		\node[circle, draw]  (s2) at (150,40) {$s_2$};

		\draw[-latex] (s0) -- (s1) node [midway,yshift=0.2cm, rotate=18] {$a, x\geq 2$};
		\draw[-latex] (s0) -- (s2)node [midway,yshift=-0.2cm, rotate=-18] {$b,x\geq 4$};

			% Invariants
			\draw (125,-82) node [anchor=north west][inner sep=0.75pt]    {\scriptsize$x \leq 3$};
			\draw (125,-2) node [anchor=north west][inner sep=0.75pt]    {\scriptsize$x \leq 5$};
			\end{scope}

			\begin{scope}[xshift=4cm]
					\node[circle, draw, minimum width=0.7cm]  (p0) at (0,0) {$\bullet$};
					\node [] (lp0) at (0,-30) {$p_0$};
					
					\node[circle, draw, minimum width=0.7cm]  (p1) at (180,-40) {};
					\node [] (lp1) at (180,-70) {$p_1$};
					
					\node[circle, draw, minimum width=0.7cm]  (p2) at (180,40) {};
					\node [] (lp1) at (180,70) {$p_2$};

					\node[draw,rectangle, fill=black, minimum height=0.7cm] (t0) at (90,-40) {};
					\node [] (lt0) at (90,-70) {$a$};
	\node [color=red] (lit0) at (50,-40) {\scriptsize$[2,3]$};
	
					\node[draw,rectangle, fill=black, minimum height=0.7cm]  (t1) at (90,40) {};
					\node [] (lt1) at (90,10) {$b$};
	\node [color=red] (lit0) at (50,40) {\scriptsize$[4,5]$};
					
			\path[draw,-latex] (p0) -- (t0){};
			\path[draw,-latex] (p0) -- (t1){};
			\path[draw,-latex] (t0) -- (p1){};
			\path[draw,-latex] (t1) -- (p2){};
			\end{scope}

%%%%%%%%%%%%%%%%%%%%%%%%%%%%%%%%%%%%%%%%%%%%%
% Encoding the automaton wit epsilon transitions
%%%%%%%%%%%%%%%%%%%%%%%%%%%%%%%%%%%%%%%%%%%%%%%%			
			\begin{scope}[xshift=10cm]
\node [draw,circle,minimum size=0.7cm] (p0) at (0,0) {$\bullet$};
\node [] (lp0) at (0,-30) {$p_0$};

\node [draw,circle,minimum size=0.7cm] (p1) at (120,40) {};
\node [] (lp1) at (120,70) {$p_1$};

\node [draw,circle,minimum size=0.7cm] (p2) at (240,40) {};
\node [] (lp2) at (240,70) {$p_2$};

\node [draw,circle,minimum size=0.7cm] (p3) at (120,-40) {};
\node [] (lp3) at (120,-70) {$p_3$};

\node [draw,circle,minimum size=0.7cm] (p4) at (240,-40) {};
\node [] (lp4) at (240,-70) {$p_4$};

\node [draw,rectangle,fill=black,minimum height=0.7cm] (t0) at (60,40){};
\node [] (lt0) at (60,10) {\scriptsize$\epsilon$};
\node [color=red] (it0) at (60,70) {\scriptsize$[0,0]$};

\node [draw,rectangle,fill=black,minimum height=0.7cm] (t1) at (180,40){};
\node [] (lt1) at (180,10) {\scriptsize$b$};
\node [color=red] (it1) at (180,70) {\scriptsize$[4,5]$};

\node [draw,rectangle,fill=black,minimum height=0.7cm] (t2) at (60,-40){};
\node [] (lt2) at (60,-70) {\scriptsize$\epsilon$};
\node [color=red] (it1) at (60,-10) {\scriptsize$[0,0]$};

\node [draw,rectangle,fill=black,minimum height=0.7cm] (t3) at (180,-40){};
\node [] (lt2) at (180,-70) {\scriptsize$a$};
\node [color=red] (it1) at (180,-10) {\scriptsize$[2,3]$};

\draw [-latex] (p0) edge  (t0){};
\draw [-latex] (p0) edge  (t2){};
\draw [-latex] (t0) edge  (p1){};
\draw [-latex] (p1) edge  (t1){};
\draw [-latex] (t1) edge  (p2){};
\draw [-latex] (t2) edge  (p3){};
\draw [-latex] (p3) edge  (t3){};

\draw [-latex] (t3) edge  (p4){};
\end{scope}
		\end{tikzpicture}
		%\vspace{-2\medskipamount}
		\caption{a) A timed automaton $\ta_1$ b) An equivalent time\underline{d} Petri net c) An equivalent TPN in $TPN_\epsilon$ }
		\label{fig_example_ta1}
		%\vspace{-\medskipamount}
		\end{center}
		\vspace{-\medskipamount}
		\end{figure}
	
\begin{proof}
	From Proposition~\ref{prop_scta_is_correct}, we can translate every bounded waiting net $\wnet$ to a finite timed automaton $\scta(\wnet)$. Notice that $\scta(\wnet)$ uses only constraints of the form $x_i\geq a$ in guards and of the form $x_i \leq b$ in invariants. Thus, $BWTPN \subseteq TA(\leq,\geq)$. This inclusion is strict. Consider the timed automaton $\ta_1$ of Figure~\ref{fig_example_ta1}$-a)$. Action $a$ can occur between date~2 and~3 and  $b$ between date~$4$ and~$5$. The timed language of $\ta_1$ cannot be recognized by a BWTPN with only two transitions $t_a$ and $t_b$, because $t_a$ must be firable and then must fire between dates 2 (to satisfy the guard) and $3$ (to satisfy the invariant in $s_1$). However, in TPNs and WTPNs, transitions that become urgent do not let time elapse, and cannot be disabled without making a discrete move. As $t_b$ is the only other possible move, but is not yet allowed, no WTPN with injective labeling can encode the same behavior as $\ta_1$.
\end{proof}

\begin{remark}
\label{remark_epsilons}
It was proved in~\cite{BCHLR13} that timed automata (with $\epsilon-$transitions) have the same expressive power as bounded TPNs with $\epsilon-$transitions. These epsilon transitions can be used to "steal tokens" of an enabled transition, and prevent it from firing urgently after a delay. This cannot be done with TPNs nor waiting nets without $\epsilon$, because a transition that is fully enabled can be disabled only by a discrete move, and not by elapsing time. Hence, bounded TPN with $\epsilon-$transitions are strictly more expressive than waiting nets, and than waiting net with non-injective labeling.
\end{remark}

\vspace{-2\medskipamount}
\begin{remark}
Another easy result is that timed Petri nets and waiting nets are incomparable. Indeed, timed Petri nets cannot encode urgency of TPNs, and as a consequence some (W)TPNs have no timed Petri net counterpart, even in the bounded case. Similarly, one can design a timed Petri net in which a transition is firable only in a bounded time interval and is then disabled when time elapses. We have seen in Figure~\ref{fig_example_ta1}-a) that $\lang(\ta_1)$ cannot be recognized by a waiting net. However, it is easily recognized by the timed Petri net of figure~\ref{fig_example_ta1}-b).
\end{remark}

\begin{theorem}
\label{prop-BWPN-BTPN}
	$TPN \strictlessexp WTPN$ and 	$BTPN \strictlessexp BWTPN$.
\end{theorem}

\begin{proof}
	Clearly, from Definition~\ref{def:WTPN}, we have $TPN \lessexp WTPN$ and $BTPN \lessexp BWTPN$, as TPNs are WTPNs without control places. It remains to show that inclusions are strict. Consider the waiting net $\wnet$ in Figure~\ref{ex-counter-TPN}-a. We have $T=\{t_0,t_1\}$,  $P=\{p_0,p_1\}, C= \{ c_0\}$ and $M_0.N_0=\{p_0\}\{\}$.
	Hence, from the starting configuration, $t_0$ is fully enabled and firable (because $M_0 \geq \preset{t_0}$), but $t_1$ is enabled and not firable ($M_0 \geq \preset{t_1}$ and $N_0(c_0)=0$).
Every valid run of $\wnet$ is of the form $\{p_0\}\{\} \xrightarrow[]{(d_0,t_0)} \{\}\{c_0\}\xrightarrow[]{(d_1,t_1)} \{p_1\}\{\} \text{ where } 0\leq d_0 \leq 20 \text{ and } d_0+d_1=20$.

\medskip
Thus the timed language of $\wnet$ is $\lang(\wnet) = \{(t_0,d_0)(t_1,20) \,|\, 0 \leq d_0 \leq 20\}$. Let us show that there exists no $TPN$ that recognizes the language $\lang(\wnet)$. In TPN, one cannot memorize the time already elapsed using the clocks of newly enabled transitions. A TPN $\mathcal N$ recognizing $\lang(\wnet)$ should satisfy the following properties:
		\begin{enumerate}
			\item The $TPN$ should contain at least two different transitions namely $t_0$ and $t_1$
			\item $t_0$ and $t_1$ are the only transitions which fire in any run of $\mathcal N$. %starting from the initial marking.
			\item $t_0$ and $t_1$ are fired only once.
			\item $t_0$ must fire first and should be able to fire at any date in $[0,20]$ units.
			\item $t_1$ must fire second at time 20 units, regardless of the firing date of $t_0$
		\end{enumerate}

\tikzset{every picture/.style={line width=0.75pt}}
\begin{figure}[!t]
\begin {center}
\begin{tikzpicture}[x=0.5pt,y=0.5pt,yscale=-1,xscale=0.94]
	
	\begin{scope}
	
		\node [draw,circle, minimum width=0.7cm] (p0) at (10,0) {$\bullet$};
	\node [] (lp0) at (10,30) {$p_0$};
	
	\node [draw,rectangle, fill=black, minimum height=0.7cm] (t0) at (70,0) {};
	\node [] (lt0) at (70,-34) {$t_0$};
	\node [color=red] (lit0) at (70,30) {\scriptsize$[0,20]$};

	\node [draw,circle, dashed, minimum width=0.7cm] (pc0) at (130,0) {};
	\node [] (lpc0) at (130,30) {$c_0$};
	
	\node [draw,rectangle, fill=black, minimum height=0.7cm] (t1) at (190,0) {};
	\node [] (lt0) at (190,-34) {$t_1$};
	\node [color=red] (lit1) at (190,30) {\scriptsize$[20,20]$};
	
	\node [draw,circle, minimum width=0.7cm] (p1) at (260,0) {};
	\node [] (lp1) at (270,30) {$p_1$};

	\draw[-latex] (p0) -- (t0) {};
	\draw[-latex] (t0) -- (pc0) {};
	\draw[-latex] (pc0) -- (t1){};
	\draw[-latex] (t1) -- (p1){};
		
	%\node [] (figlaba) at (140,55) {$a)$};

	\end{scope}
	\begin{scope}[xshift=6cm]
	\node [draw,circle, minimum width=0.7cm] (p0) at (0,0) {$\bullet$};
	\node [] (lp0) at (0,30) {$p_0$};
	
	\node [draw,rectangle, fill=black, minimum height=0.7cm] (t0) at (60,0) {};
	\node [] (lt0) at (60,-34) {$t_0$};
	\node [color=red] (lit0) at (60,30) {\scriptsize$[0,20]$};
	
	\node [draw,circle, minimum width=0.7cm] (p) at (120,0) {};
	\node [] (lpp) at (120,30) {$p$};
	
	\node [draw,rectangle, fill=black, minimum height=0.7cm] (t1) at (180,0) {};
	\node [] (lt0) at (180,-34) {$t_1$};
	\node [color=red] (lit1) at (180,30) {\scriptsize$[\alpha,\beta]$};
	
	\draw[-latex] (p0) -- (t0) {};
	\draw[-latex] (t0) -- (p) {};
	\draw[-latex] (p) -- (t1){};	
	%\node [] (figlab) at (140,55) {$b)$};
	\end{scope}

	\begin{scope}[xshift=10.8cm]
	
		\node [draw,circle, minimum width=0.7cm] (p0) at (0,0) {$\bullet$};
	\node [] (lp0) at (0,30) {$p_0$};
	
	\node [draw,rectangle, fill=black, minimum height=0.7cm] (t0) at (60,0) {};
	\node [] (lt0) at (60,-34) {$t_0$};
	\node [color=red] (lit0) at (60,30) {\scriptsize$[0,20]$};

	\node [draw,circle, dashed, minimum width=0.7cm] (pc0) at (120,0) {};
	\node [] (lpc0) at (120,30) {$c_0$};
	
	\node [draw,rectangle, fill=black, minimum height=0.7cm] (t1) at (180,0) {};
	\node [] (lt0) at (180,-34) {$t_1$};
	\node [color=red] (lit1) at (180,30) {\scriptsize$[20,20]$};
	
	\node [draw,circle, minimum width=0.7cm] (p1) at (240,0) {};
	\node [] (lp1) at (240,-30) {$p_1$};
	
    \node [draw,rectangle, fill=black, minimum width=0.7cm] (t2) at (240,90) {};
	\node [] (lt0) at (204,90) {$t_2$};
	\node [color=red] (lit1) at (280,90) {\scriptsize$[0,1]$};
	
	\node [draw,circle, minimum width=0.7cm] (p2) at (240,150) {};
	\node [] (lp2) at (204,150) {$p_2$};
		
	\draw[-latex] (p0) -- (t0) {};
	\draw[-latex] (t0) -- (pc0) {};
	\draw[-latex] (pc0) -- (t1){};
	\draw[-latex] (t1) -- (p1){};
	\path[-latex] (p1) edge [bend left] (t2){};
	
	\path [-latex] (t2) edge [bend left] (p1){};
	\draw[-latex] (t2) -- (p2){};
		
	%\node [] (figlaba) at (140,55) {$a)$};
	\end{scope}
\end{tikzpicture}\vspace*{-3mm}
		\caption{$a)$ A waiting net $\wnet$ \hspace{0.6cm} $b)$ a part of TPN needed to encode $\lang(\wnet)$ \hspace{0.3cm}
$c)$ An unbounded waiting net}
\label{ex-counter-TPN}
\end{center}
\vspace{-2\medskipamount}
\end{figure}
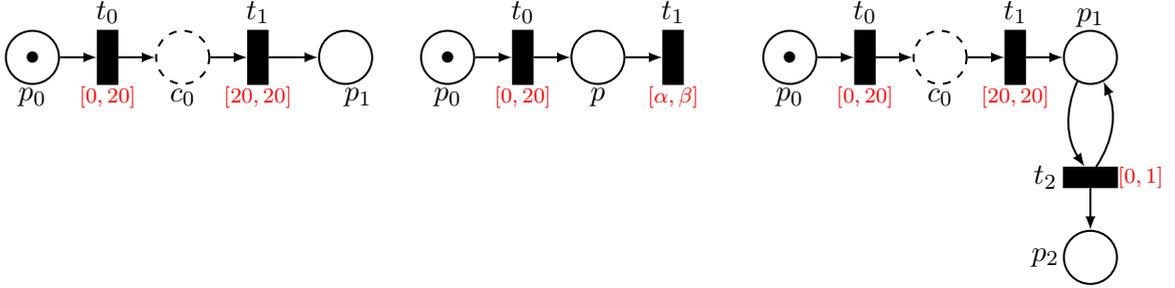

	The above conditions are needed to ensure that $\lang(\mathcal N)= \lang(\wnet)$. We can now show that it is impossible to build a net $\mathcal N$ satisfying all these constraints. Since $t_1$ must fire {\em after} $t_0$ , $\mathcal N$ must contain a subnet of the form shown in figure~\ref{ex-counter-TPN}-b, where $p$ is an empty place preventing firing $t_1$ before $t_0$. Notice however that $p$ forbids enabling $t_1$ (and hence measuring time) from the beginning of a run.
We can force $t_0$ to fire between 0 and 20 units with the appropriate time interval $[0,20]$, but the $TPN$ of Figure~\ref{ex-counter-TPN}-$b)$ can not remember the firing date of $t_0$, nor associate to $t_1$ a clock which value increases  before $t_0$ fires.
	Let $[\alpha,\beta]$ be the time interval associated to  $t_1$,
	and consider the two extreme but legal firing dates of $t_0$, namely 0 and 20 time units. Allowing these two dates amounts to requiring that $0+\alpha(t_1) = 0 + \beta(t_1) = 20 + \alpha(t_1) = 20 + \beta(t_1) = 20$ which is impossible. Hence, there exists no $TPN$ recognizing $\lang(\wnet)$. This shows that $BTPN \strictlessexp BWTPN$. The proof easily extends to show that $TPN \strictlessexp WTPN$, simply by adding an unbounded part of net that becomes active immediately after firing $t_1$, as depicted in the waiting net of Figure~\ref{ex-counter-TPN}-c.	
\end{proof}

%STOPWATCHES
As mentioned earlier in the paper, one can consider that in a run of waiting net, a transition with an empty control
place in its preset can let its clock reach the upper bound of its interval, and then freeze it, waiting for the control place to be filled. A {\em stopwatch time Petri net} (SwTPN for short) is a TPN with an additional relation $sw:P\to T$, defining stopwatch arcs. An enabled transition $t$ such that $M\geq sw(t)$ is said {\em active}, and its clock evolves during timed moves. An enabled transition such that $M< sw(t)$ is said {\em suspended} and it clock does not change during timed move until it becomes active again.
A natural question is hence to compare waiting nets and SwTPNs. The first thing to notice is that bounded SwTPNs (BSwTPNs) are already Turing powerful~\cite{BerthomieuLRV07} while bounded waiting nets are not. So we immediately obtain the following proposition.

\begin{proposition}
\label{Prop_some_swnets_cannotWnets}
There exists bounded SwTPNs that cannot be simulated by bounded waiting nets.
\end{proposition}

One interesting point to remark is that expressiveness of waiting nets and SwTPNs have different origins. Waiting nets (and TPNs) can encode instructions of a Counter Machine only if they are unbounded. For stopwatch Petri nets, counter values in a Minsky Machine can be encoded by differences between upper and lower bounds on clock constraints in domains~\cite{BerthomieuLRV07}, and this encoding works even whith bounded SwTPNs.
We can go further and prove that bounded waiting nets and BSwTPNS are incomparable. Consider again the example of Figure~\ref{ex-counter-TPN}-a. We already know that the language of this bounded waiting net cannot be encoded by a bounded TPN with only two transitions and an injective labeling. Now, allowing stopwatch arcs still does not alllow for the encoding of a SwTPN recognizing the same timed language as the waiting net of Figure~\ref{ex-counter-TPN}-a.

\begin{proposition}
\label{Prop_some_wnets_cannotBswNets}
Bounded waiting nets and bounded stopwatch Petri nets are uncomparable.
\end{proposition}

\begin{proof}
We already know from Proposition~\ref{Prop_some_swnets_cannotWnets} that some BSwTPNs cannot be simulated by waiting nets. It remains to show the converse direction.
Assume that there exists a BSwTPN that can encode the same language as the waiting net of Figure~\ref{ex-counter-TPN}-a.
We know from the proof of Theorem~\ref{prop-BWPN-BTPN} that, to encode this example with a TPN that has only two transitions $t_0$ and $t_1$, if $\postset{t_0}\cap \preset{t_1}\neq \emptyset$, then, one has to enforce occurrence of $t_1$ at date $20-d_0$ for every possible date $d_0$ for the firing of $t_0$. We also proved that this is impossible for any time interval $[\alpha,\beta]$ attached to $t_1$.
Let us assume that $\postset{t_0}\cap \preset{t_1}\neq \emptyset$, and that there exists $p_{sw} \in sw(t_1)$, i.e. that a stopwatch arc from $p_{sw}$ to $t_1$ controls clock and firing of $t_1$. If $p_{sw} \in \postset{t_0}$ then $t_1$ is suspended until $t_0$ fires, and we are still in a situation where we cannot enforce firing of $t_1$  exactly $20-d_0$ time units after firing of $t_0$. If $p_{sw} \in \preset{t_0}$, then after firing of $t_0$, $t_1$ cannot fire.
Let us now assume that $\postset{t_0}\cap \preset{t_1} = \emptyset$. We already know from Theorem~\ref{prop-BWPN-BTPN} that without a stopwatch arc restricting firing of $t_1$, one cannot prevent $t_1$ from firing at date $20$ before $t_0$. As in the preceding case, if we attach interval $[\alpha,\beta]$ to transition $t_1$ and require $p_{sw} \in \postset{t_0}$, then $t_1$ is initially suspended, and we mush have $[\alpha,\beta]=[20-d_0,20-d_0]$ for all values of $d_0$, which is impossible. Now assume that $p_{sw} \in \preset{t_0}$. Then $t_1$ is suspended after firing of $t_0$ and hence we cannot fire it. We hence have $p_{sw} \not\in \preset{t_0} \cup \postset{t_0}$, and as the contents of $p_{sw}$ is unchanged by firing of $t_0$, then $t_1$ is either never firable if $M_0(p_{sw})=0$, or can fire before $t_0$ if $M_0(p_{sw})>0$ and we set $[\alpha,\beta]=[20,20]$ , or at an unwanted dates if we attach another arbitrary interval $[\alpha,\beta]$ to $t_1$.
\end{proof}

\begin{theorem}
\label{thm_generic_inj-nonInj_expressiveness}
Injective TPNs and WTPNS are strictly less expressive than their non-injective counterparts, i.e.
$$\begin{array}{ll} BTPN \strictlessexp BTPN^{\noninj}, &   TPN \strictlessexp TPN^{\noninj},\\
BWTPN \strictlessexp BWTPN^{\noninj}, &  WTPN \strictlessexp WTPN^{\noninj}\\
\end{array}$$
\end{theorem}

\begin{proof}
With non-injective labeling, (W)TPNs can recognize unions of timed language, which is not the case for models with injective labeling.
In every timed word of the language of a model with injective labeling, a letter represents an occurrence of a transition. That is, every occurrence of some letter  $\sigma$ labeling a transition $t_\sigma$ is constrained by time in a similar way in a word: if $t_\sigma$ is enabled at some date $d$, then $t_\sigma$ must occur later at some date $d'$ with $d+\alpha(t_\sigma) \leq d' \leq \beta(t_\sigma)$ {\em at every position} in the considered word. This remark also holds for distinct words with the same prefix: let $w.(\sigma,d)$ be a timed word, with $w=(\sigma_1,d_1).(\sigma_2,d_2)\dots (\sigma_k,d_k)$. The possible values for $d$ lay in an interval that only depends on the unique sequence of transitions followed to recognize $w$.
%, which uniquely determines when $t_\sigma$ was last newly enabled.
With non-injective labeling, one can recognize a word $w$ via several sequences of transitions, and associate different constraints to the firing date of $\sigma$. The union of this set of constraints need not be a single interval. Consider for instance the TPN of Figure~\ref{fig_example_injective}, that defines the language $\lang(\tpn)= \{(a,d_1).(b,d_2)\mid d_1\in [0,1] \wedge d_2 \in [d_1+4,d_1+5] \} \cup \{(a,d_1).(b,d_2) \mid d_1\in [0,1] \wedge d_2 \in [d_1+7,d_1+8] \}$. Hence, TPNs and WTPNs variants are strictly less expressive than their non-injective counterparts.
\end{proof}

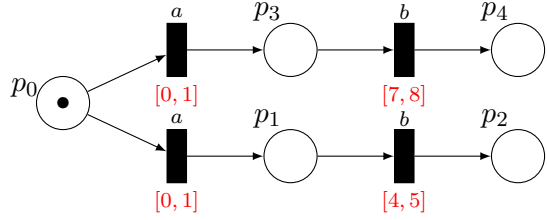
\begin{figure}[ht]
\vspace{-2\medskipamount}
\begin{center}
\begin{tikzpicture}
\begin{scope}
\node [draw,circle,minimum size=0.7cm] (p0) at (1,-0.7) {$\bullet$};
\node [] (lp0) at (0.5,-0.5) {$p_0$};

\node [draw,circle,minimum size=0.7cm] (p1) at (4,-1.4) {};
\node [] (lp1) at (3.7,-0.9) {$p_1$};

\node [draw,circle,minimum size=0.7cm] (p2) at (7,-1.4) {};
\node [] (lp2) at (6.7,-0.9) {$p_2$};

\node [draw,circle,minimum size=0.7cm] (p3) at (4,0) {};
\node [] (lp3) at (3.7,0.5) {$p_3$};

\node [draw,circle,minimum size=0.7cm] (p4) at (7,0) {};
\node [] (lp4) at (6.7,0.5) {$p_4$};

\node [draw,rectangle,fill=black,minimum height=0.7cm] (t0) at (2.5,-1.4){};
\node [] (lt0) at (2.5,-0.9) {\scriptsize$a$};
\node [color=red] (it0) at (2.5,-2) {\scriptsize$[0,1]$};

\node [draw,rectangle,fill=black,minimum height=0.7cm] (t1) at (5.5,-1.4){};
\node [] (lt1) at (5.5,-0.9) {\scriptsize$b$};
\node [color=red] (it1) at (5.5,-2) {\scriptsize$[4,5]$};

\node [draw,rectangle,fill=black,minimum height=0.7cm] (t2) at (2.5,0){};
\node [] (lt2) at (2.5,0.5) {\scriptsize$a$};
\node [color=red] (it1) at (2.5,-0.6) {\scriptsize$[0,1]$};

\node [draw,rectangle,fill=black,minimum height=0.7cm] (t3) at (5.5,0){};
\node [] (lt2) at (5.5,0.5) {\scriptsize$b$};
\node [color=red] (it1) at (5.5,-0.6) {\scriptsize$[7,8]$};

\draw [-latex] (p0) edge  (t0){};
\draw [-latex] (p0) edge  (t2){};
\draw [-latex] (t0) edge  (p1){};
\draw [-latex] (p1) edge  (t1){};
\draw [-latex] (t1) edge  (p2){};
\draw [-latex] (t2) edge  (p3){};
\draw [-latex] (p3) edge  (t3){};

\draw [-latex] (t3) edge  (p4){};
\end{scope}
\end{tikzpicture}
\vspace{-1\medskipamount}
\caption{A TPN $\tpn_2$ with non-injective labeling.}
\label{fig_example_injective}
\end{center}
\vspace{-1\medskipamount}
\end{figure}

\begin{corollary}
	$BTPN^{\noninj} \strictlessexp BWTPN^{\noninj}$
\end{corollary}
\begin{proof}
Inclusion $BTPN^{\noninj} \lessexp BWTPN^{\noninj}$ is straightforward from Definition~\ref{def:WTPN}. Now this inclusion is strict. Take the example of Figure~\ref{ex-counter-TPN}-a), and attach distinct labels to $t_0$ and $t_1$. The language recognized cannot be encoded with a non-injective TPN, for the reasons detailed in the proof of Theorem.~\ref{prop-BWPN-BTPN}.
\end{proof}

To conclude on the effects of non-injective labeling, we can easily notice that 	$BWTPN^{\noninj} \lessexp TA(\leq,\geq)$ because the automaton construction of Definition~\ref{def-scta} still works. Indeed, the automaton construction depends on transitions fired in extended state classes, and not on their labels. Hence, starting from a non-injective BWTPN, the construction of Definition~\ref{def-scta} gives a timed automaton that recognizes the same language. As in the injective case, this relation is strict, because the example $\ta_1$ of Figure~\ref{fig_example_ta1}-a cannot be implemented by a waiting net in $BWTPN^{\noninj}$.

The last point to consider is whether allowing silent transitions increases the expressive power of the model.
Consider again the automaton $\ta_1$ of Figure~\ref{fig_example_ta1}-a. If transitions labeled by $\epsilon$ are allowed, then the timed language $\lang(\ta_1)$ can be encoded by the TPN of Figure~\ref{fig_example_ta1}-c.
It was shown in~\cite{DiekertGP97} that timed automata with epsilon transitions are strictly more expressive than timed automata without epsilon. We hence have $TA(\leq,\geq) \strictlessexp TA_\epsilon(\leq,\geq)$. We can also show that differences between WTPNs, TPN, and automata disappear when silent transitions are allowed.

\begin{theorem}
\label{thm_equality_with_epsilon}
$TA_\epsilon(\leq,\geq) \equlang 	BTPN_\epsilon \equlang  	BWTPN_\epsilon$
\end{theorem}

\begin{proof}
The equality $TA_\epsilon(\leq,\geq) =  	BTPN_\epsilon$ was already proved in~\cite{BCHLR13}. From Definition~\ref{def:WTPN}, we have
$BTPN_\epsilon \lessexp WBTPN_\epsilon$.  Further, for any waiting net  $\wnet \in BWTPN_\epsilon$, one can apply the construction of Definition~\ref{def-scta}  to obtain a state class timed automaton $SCTA(\wnet)$ (with $\epsilon$ transitions) recognizing the same language as $\wnet$. We hence have  $BWTPN_\epsilon \lessexp TA_\epsilon(\leq,\geq) \equlang BTPN_\epsilon$, which allows to conclude that $TA_\epsilon(\leq,\geq),	BTPN_\epsilon$, and $BWTPN_\epsilon$ have the same expressive power.
\end{proof}

Figure~\ref{fig_timed_model_relations} summarizes the relations among different classes of nets and automata, including TPNs and automata with stopwatches. An arrow $\mathcal M_1 \longrightarrow \mathcal M_2$ means that $\mathcal M_1$ is strictly less expressive than $\mathcal M_2$, and this relation is transitively closed.
Extensions of timed automata with stopwatches allow the obtained model to simulate two-counters machines (and then Turing machines)~\cite{CassezL00}. It is well known that unbounded TPNs can simulate Turing Machines too. More surprisingly, adding stopwatches to bounded TPNs already gives them the power of Turing Machines~\cite{BerthomieuLRV07}.
Obviously, all stopwatch models can simulate one another.
%However, TPN simulate counters with place contents while stopwatch automata and stopwatch BTPN simulate them with clock differences.
Hence, these models are equally expressive in terms of timed languages as soon as
%one can hide the intermediate steps needed to simulate a counter operation, i.e. as soon as
they allow $\epsilon$ transitions. The thick dashed line in Figure~\ref{fig_timed_model_relations} is the frontier for Turing powerful models, and hence also for decidability of reachability or coverability.

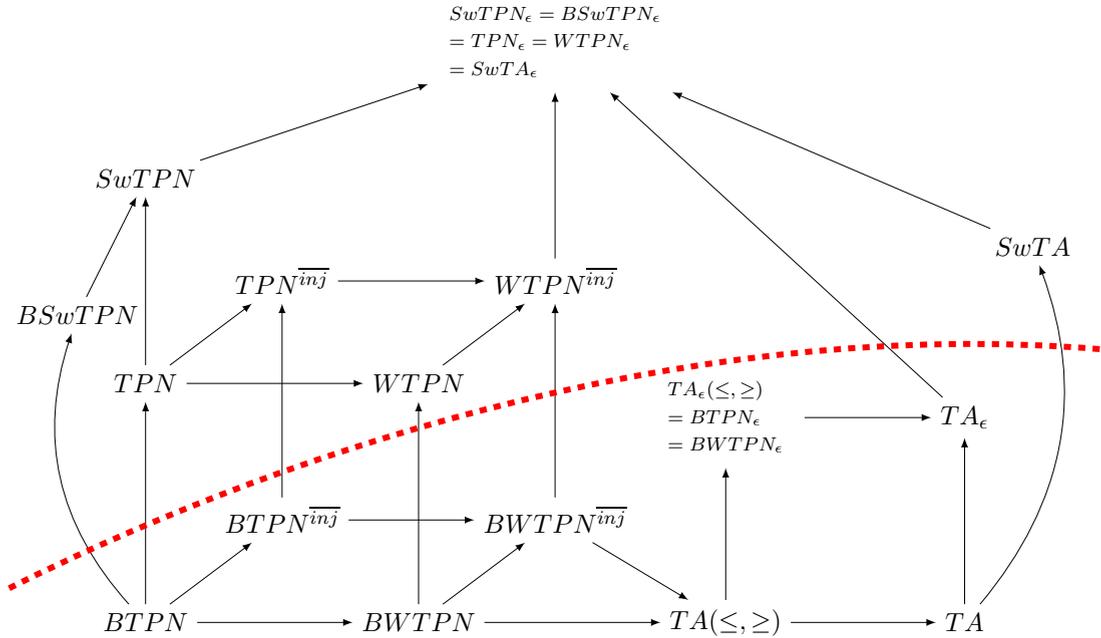
\begin{figure}[htbp]
\vspace{2mm}
\begin{center}
\scalebox{0.9}{
\begin{tikzpicture}

\node (btpn) at (0,-0.5){$BTPN$};
\node (tpn) at (0,3){$TPN$};
\node (swtpn) at (0,6){$SwTPN$};
\node (bswtpn) at (-1,4){$BSwTPN$};
\path (btpn) edge [-latex, bend left] (bswtpn){};
\draw [-latex] (btpn) -- (tpn){};
\draw [-latex] (tpn) -- (swtpn){};
\draw [-latex] (bswtpn) -- (swtpn){};

\node (bwtpn) at (4,-0.5){$BWTPN$};
\node (wtpn) at (4,3){$WTPN$};
\draw [-latex] (tpn) -- (wtpn){};
\draw [-latex] (bwtpn) -- (wtpn){};
\draw [-latex] (btpn) -- (bwtpn){};

\node (talg) at (8.5,-0.5){$TA(\leq,\geq)$};
\node (ta) at (12,-0.5){$TA$};
\node (swta) at (13,5){$SwTA$};
\path (ta) edge[-latex, bend right] (swta) {};

\draw [-latex] (bwtpn) -- (talg){};

\node (taelg) at (8.5,2.5){\scriptsize$\begin{array}{l}TA_\epsilon(\leq,\geq)\\= BTPN_\epsilon\\ = BWTPN_\epsilon \end{array}$};

\node (tae) at (12,2.5){$TA_\epsilon$};
\draw [-latex] (taelg) -- (tae){};
\draw [-latex] (talg) -- (ta){};
\draw [-latex] (talg) -- (taelg){};
\draw [-latex] (ta) -- (tae){};

\node (btpninj) at (2,1){$BTPN^{\noninj}$};
\node (bwtpninj) at (6,1){$BWTPN^{\noninj}$};
\node (tpninj) at (2,4.5){$TPN^{\noninj}$};
\node (wtpninj) at (6,4.5){$WTPN^{\noninj}$};

\draw [-latex] (btpninj) -- (bwtpninj){};
\draw [-latex] (wtpn) -- (wtpninj){};
\draw [-latex] (tpn) -- (tpninj){};
\draw [-latex] (bwtpninj) -- (wtpninj){};
\draw [-latex] (btpninj) -- (tpninj){};
\draw [-latex] (tpninj) -- (wtpninj){};
\draw [-latex] (btpn) -- (btpninj){};
\draw [-latex] (bwtpn) -- (bwtpninj){};
\draw [-latex] (bwtpninj) -- (talg){};

\node (allTuring) at (6, 8) {\scriptsize$\begin{array}{l}
SwTPN_\epsilon = BSwTPN_\epsilon\\
= TPN_\epsilon = WTPN_\epsilon\\
= SwTA_\epsilon
\end{array}$};

\draw [-latex] (tae) -- (allTuring){};
\draw [-latex] (swtpn) -- (allTuring){};
\draw [-latex] (swta) -- (allTuring){};

\draw [-latex] (wtpninj) -- (allTuring){};

%%% Decidability limit
\draw[thick, dashed,color=red,line width=2.5pt] (-2,0) .. controls (2,2) and (8,4).. (14,3.5) {};

\end{tikzpicture}
}
%\vspace{-\medskipamount}
\caption{Relation among net and automata classes, and frontier of decidability.}
\label{fig_timed_model_relations}
\vspace{-3\medskipamount}
\end{center}
\end{figure}

\section{Conclusion}
\label{sec:conclusion}

We have proposed waiting nets, a new variant of time Petri nets, that measure time elapsed since enabling
of a transition while waiting for additional control allowing its firing. This class obviously subsumes time Petri nets. More interestingly, expressiveness of bounded waiting nets lays between that of bounded TPNs and timed automata.
Waiting nets allow for a finite abstraction of the firing domains of transitions. A consequence is that one can compute a finite state class graph for bounded WTPNs, and decide reachability and coverability.

\medskip
As future work, we will investigate properties of classes of WTPN outside the bounded cases. In particular, we should investigate if being free-choice allows for the decidability of more properties in unbounded WTPNs~\cite{AkshayHP20}. A second interesting topic is control. Waiting nets are tailored to be guided by a timed controller, filling control places in due time to allow transitions firing. A challenge is to study in which conditions one can synthesize a controller to guide a waiting net in order to meet a given objective.

  %%%\nocite{*}\bibliographystyle{fundam}% the mandatory bibstyle\bibliography{biblio}

\appendix
\label{sec-appendix}

\section{Fourier-Motzkin elimination}
\label{appendix-Fourier}

Fourier-Motzkin Elimination~\cite{DantzigE73} is a method to  eliminate a set of variables $V\subseteq X$ from a system of linear inequalities over $X$. Elimination produces another system of linear inequalities over $X\setminus V$, such that both systems have the same solutions over the remaining variables. Elimination can be done by removing one variable from $V$ after another.

\medskip
Let $X = \{ x_1,\dots x_r \}$ be a set of variables, and w.l.o.g., let us assume that $x_r$ is the variable to eliminate in $m$ inequalities. All inequalities are of the form
$$c_1.x_1 + c_2.x_2 +\dots +c_r.x_r \leq d_i$$ where $c_j$'s and $d_i$ are rational values, or equivalently
$c_r.x_r \leq d_i - (c_1.x_1 + c_2.x_2 +\dots + c_{r-1}.x_{r-1})$

If  $c_r$ is a negative coefficient, the inequality can be rewritten as
$x_r \geq b_i - (a_{i,1}.x_1 + a_{i,2}.x_2 +\dots a_{i,r-1}.x_{r-1})$, and if $c_r$ is positive, the inequality rewrites as
$x_r \leq b_i - (a_{i,1}.x_1 + a_{i,2}.x_2 +\dots a_{i,r-1}.x_{r-1})$, where $b_i=\frac{d_i}{c_r}$ and $a_i=\frac{c_i}{c_r}$.

\medskip
We can partition our set of inequalities as follows.
\begin{itemize}
    \item inequalities of the form $ x_{r} \geq b_i-\sum_{k=1}^{r-1} a_{ik} x_k;$
    denote these by\\
    $x_{r} \geq A_j(x_1, \dots, x_{r-1})$ (or simply $x_{r} \geq A_j$ for short), for $j$ ranging from $1$ to $n_{A}$ where $n_{A}$ is the number of such inequalities;
    \item  inequalities of the form $ x_{r} \leq b_i-\sum_{k=1}^{r-1} a_{ik} x_k;$
    denote these by\\
    $x_{r} \leq B_j(x_1, \dots, x_{r-1})$ (or simply $x_{r} \leq B_j$ for short), for $j$ ranging from $1$ to $n_{B}$ where $n_{B}$ is the number of such inequalities;
    \item  inequalities $\phi_1, \dots \phi_{m-(n_A+n_B)}$ in which $x_r$ plays no role.
\end{itemize}

The original system is thus equivalent to:\\
$\max\left (A_1, \ldots, A_{n_A}\right) \leq  x_r \leq \min\left(B_1, \ldots, B_{n_B}\right) \land \underset{i\in 1..m-(n_A+n_B)}\bigwedge  \phi_i$.

%Elimination consists in producing a system equivalent to $\exists x_r, S$.
One can find a value for $x_r$ in a system of the form $a \leq x \leq b$ iff $a\leq b$. Hence, the above formula is equivalent to :\\
$\max(A_1, \ldots, A_{n_A}) \leq \min(B_1, \ldots, B_{n_B}) \land \underset{i\in 1..m-(n_A+n_B)}\bigwedge  \phi_i$

Now, this inequality can be rewritten as system of  $n_A \times n_B + m-(n_A+n_B)$ inequalities $\{ A_i \leq B_j \mid i\in 1..n_A,j\in 1..n_B\} \cup \{ \phi_i \mid i\in 1..m-(n_A+n_B)\}$, that does not contain $x_r$ and is satisfiable iff the original system is satisfiable.

\begin{remark}
The Fourier-Motzkin elimination preserves finiteness and satisfiability of a system of constraints. In general, the number of inequalities can grow in a quadratic way at each variable elimination. However, when systems describe firing domains of transitions and are in canonical form, they always contain less than $2\cdot|T|^2+2\cdot|T|$ inequalities, and then elimination produces a system of at most $2\cdot|T|^2+2\cdot|T|$ inequalities once useless inequalities have been removed.
\end{remark}

\section{Canonical forms : the Floyd-Warshall algorithm}
\label{appendix_Floyd_Warshall}

A way to compute the canonical form for a firing domain $D$ (i.e. the minimal set of constraints defining $\llbracket D \rrbracket$, the possible delays before firing of transitions) is to consider this domain as constraints on distances from variable's value to value $0$, or to the value of another variable. The minimal set of constraints can then be computed using a slightly modified version of the Floyd-Warshall algorithm. Given a marking $M.N$, a firing domain $D$ is defined with inequalities of the form:
\begin{itemize}
\item for every transition $t_i\in \mathsf{enabled}(M) $ we have an inequality of the form $a_i\leq \theta_i \leq b_i \,\forall t_i \in \mathsf{enabled}(M)$, with  $a_i,b_i \in \mathbb{Q}^+$,
\item for every pair of transitions $t_j \not= t_k \in \mathsf{enabled}(M)^2$ we have an inequality of the form $\theta_j -\theta_k \leq c_{jk} \, \forall t_j \not= t_k \in \mathsf{enabled}(M)$ with  $c_{jk} \in \mathbb{Q}^+$
\end{itemize}

We then apply the following algorithm:
\begin{algorithm}
\label{alg-flo-war}
\SetAlgoLined
 Input: $D$, $E= \mathsf{enabled}(M)$\;
 \For{$t_k\in E$}{
  \For{$t_j \in E$}{
   \For{$t_i \in E$}{
    $r := \min(r,b_k -  a_k)$\\
    $a_j := \max(a_j, a_k - c_{kj})$\\
    $b_i := \min(b_i, b_k + c_{ik})$\\
    $c_{ij} := \min(c_{ij}, c_{ik} + c_{kj})$\\
   }
  }
 }
 Output : $D^*$
 \caption{Floyd-Warshall}
\end{algorithm}

The system of inequalities input to the algorithm is satisfiable iff at every step of the algorithm $r\geq 0$. One can easily see that computing a canonical form, or checking satisfiability of a firing domain can be done in cubic time w.r.t. the number of enabled transitions.

\section{Difference Bound Matrices}
\label{DBM}

An interesting property of firing domains is that they can can be represented with Difference Bound Matrices (DBMs).
To have a unified form for constraints
%(here we will interpret $\theta_i$'s as clocks)
we introduce a reference value \textbf{0} with the constant value 0.
Let $\mathcal{X}_0 = \mathcal{X}\cup \{0\}$. Then any domain on variables in $\mathcal{X}$ can be rewritten as a
conjunction of constraints of the form $x - y \preceq n$ for $x, y \in \mathcal{X}_0$, $\preceq \in \{<, \leq\}$ and $n \in \mathbb{Q}$.\\
For instance, $1 \leq \theta_1 \leq 2$ can be rewritten as $\textbf{0}-\theta_1 \leq -1 \land \theta_1 - \textbf{0}\leq 2$.
Naturally, if the encoded domain has two constraints on the same pair of variables, we use their intersection: for a pair of constraints  $\textbf{0}-\theta_1 \leq -1 \land \textbf{0}-\theta_1 \leq 0$ we only keep $\textbf{0}-\theta_1 \leq -1$

\medskip
A firing domain $D$ can then be represented by a matrix $M_D$ with entries indexed by $\mathcal{X}_0$:
\begin{itemize}
\itemsep=0.85pt
    \item For each constraint $\theta_i - \theta_j \preceq n$ of $D$, $M_D[i,j]=(n, \preceq)$
		\item For each constraint $a_i \leq \theta_i \leq b_i$, that is equivalent to $\textbf{0} - \theta_i \leq - a_i $  and $ \theta_i - \textbf{0} \leq b_i $, we set $M_D[0,i]=(-a_i, \preceq)$ and $M_D[i,0]=(b_i, \preceq)$
    \item For each clock difference $\theta_i - \theta_j$ that is unbounded in $D$, let $D_{ij} = \infty$
    \item Implicit constraints $\textbf{0}- \theta_i \leq 0$ and $\theta_i - \theta_i \leq 0$ are added for all clocks.
\end{itemize}

\paragraph*{Canonical DBM}
There can be infinite number of DBMs to represent a single set of solutions for a domain $D$, but each domain has a unique canonical representation, that can be computed as a closure of $M_D$ or equivalently as a closure on a constraint graph. From matrix $M_D$, we create a directed graph $\mathcal{G}(D)$: Nodes of the graph are variables $\textbf{0}$ and variables $\theta_1, \theta_n$. For each entry of the form $M_D[ij] = (n,\preceq)$ we have an edge from node $\theta_j$ to $\theta_i$ labeled with $n$. The tightest constraint between two variables $\theta_i$ and $\theta_j$, is the value of the shortest path between the respective nodes in above graph. This can be done using the Floyd-Warshall algorithm shown above.

\section{Proofs for Section~\ref{sec:reachability}}
\label{Appendix_proof_reach}

The following lemma shows that if a transition is firable, then one can always find an appropriate timing allowing for the computation of a successor.

\begin{lemma}
\label{lemma_postset_not_empty}
Let $(M.N,D)$ be a state class, let $t_f$ be a transition firable from that class.
Then there exists a bound $r \in B^{t_f}_{M.N,D}$ such that $\mathsf{next}_r(D, t_f)$ is satisfiable.
\end{lemma}

\begin{proof}
Let $t_f$ be such a firable transition from the state class $(M\cdot N, D)$. This means, $\theta_f$ satisfies $D \land \bigwedge_{t_j\in \mathsf{FullyEnabled}(M\cdot N)} \theta_f \leq \theta_j$.
Now, since $t_f$ is firable, this implies $t_f \in E$ and $\forall t_j\in \mathsf{FullyEnabled}(M\cdot N)$ $a^*_f \leq b_j$ (i.e. $t_f$ can be fired before any other transition becomes urgent) $\ldots (1)$.

\medskip
Let $[\ bnd_k, bnd_{k+1}]\ $ be the smallest interval between two consecutive bounds of $B_{M.N,D}$ containing $a^*_f$. Such an interval exists as $ bnd_0 \leq a^*_f \leq b^*_f < \infty $.
Now, we know that elimination preserves satisfiability. So the successor domain $\mathsf{next}_k(D, t_f)$ is satisfiable, because it is obtained by elimination from $D \land \bigwedge_{t_j\in \mathsf{FullyEnabled}(M\cdot N)} \theta_f \leq \theta_j \wedge bnd_k \leq \theta_f \leq b_{k+1}$ which contains at least one solution with $\theta_f = a^*_f$, and by adding new individually satisfiable constraints associated with newly enabled transitions. These constraints do not change satisfiability of domains obtained after elimination, because each one constrains only a single variable that was not used in the domain $D'^{E,b_k}$.  Hence, choosing $r=k$ witnesses truth of the lemma.
\end{proof}

\noindent{\bf Proposition~\ref{prop-stateclass-wnet-Compl}}
%The language of state class graph defined is same as the language of underlying Waiting net.
For every run $\rho=(M_0.N_0,v_0) \dots (M_k.N_k,v_k)$ of $\mathcal{W}$ there exists a path $\pi$ of $SCG(\mathcal{W})$
such that $\rho$ and $\pi$ coincide.

\begin{proof}
Markings in a run $\rho$ and a path $\pi$ that coincide with $\rho$ must be the same. To guarantee existence of $\pi$, it hence remains to show
that there exists a consistent sequence of domains $D_1 \dots D_k$ such that path\\ $\pi=(M_0.N_0,D_0).(M_1.N_1,D_1)\dots(M_k.N_k,D_k)$ is a path of $SCG(\mathcal{W})$ and coincides with $\rho$. We show it by induction on the size of runs.

\begin{itemize}
    \item \textbf{Base Case : }\\
    Let $\rho=(M_0.N_0,v_0) \xrightarrow[]{e_1= (d_1,t_1)} (M_1.N_1,v_1)$. The corresponding path in $SCG$ is of the form
		$\pi=(M_0.N_0,D_0) \xrightarrow[]{t_1} (M_1.N_1,D_1)$, where $D_0$ is already known.
    By definition we have $v_0(t) = 0$ for every enabled transition in $M_0.N_0$ and $v_0(t) = \bot$ for other transitions.
		After letting $d_1$ time units elapse, the net reaches configuration $(M_0.N_0,v'_0)$ where $v'_0(t)= v_0(t) + d_1$ if $t \in \mathsf{FullyEnabled(M_0.N_0)}$, $v'_0(t)=\min(v(t)+d_1, \beta(t))$ if  $t \in \mathsf{enabled(M_0)}\setminus\mathsf{FullyEnabled(M_0)}$ and $v'_0(t)=\bot$ otherwise.
		
  The valuation $v_1$ is computed according to timed moves rules, which means that
	$$v_1(t) =\left \{ \begin{array}{l}
	0 \mbox{~if~} t \in \uparrow \mathsf{enabled}(M_1,t_1)\\,
	v'_0(t) \mbox{~if~}  t \in Enabled(M_0) \cap Enabled(M_1)\\
	\bot \mbox{~otherwise}
	\end{array}\right.$$
\hspace*{6mm}	Now, since $\rho=(M_0.N_0,v_0) \xrightarrow[]{(d_1,t_1)} (M_1.N_1,v_1)$ is a valid move in $\mathcal{W}$, $t_1$ as a valid discrete move after the timed move $d_1$, which gives us the following condition :
  $$
  \begin{cases}
    M_0.N_0 \geq {^\bullet}\!t_1 \\
    M_1.N_1 = M_0.N_0 - {^\bullet}\!t_1 + t_1^\bullet\\
    \alpha(t_1) \leq v_0'(t_1) \leq \beta(t_1)\\
    \forall t, v_1(t) = 0 \textbf{ if }  \uparrow \mathsf{enabled}(t, M_1, t_1) \textbf{ else } v_0'(t)
  \end{cases}
	$$

\hspace*{6mm}	Further, as $t_1$ can be the first transition to fire, we have $v_0'(t_i) \leq \beta_i$ for every $t_i$ in $\mathsf{FullyEnabled(M_0)}$. As all clocks attached to fully enabled transitions have the same value $d_1$, it means that for every $t_i$ we have $d_1 \leq \beta_i$.
Let us now show that $t_1$ is firable from class $C_0$. Transition $t_1$ is firable from $M_0.N_0,D_0$ iff there exists a value $\theta_1$ such that $\alpha(t_1) \leq v_0'(t_1) \leq \beta(t_1)$, and iff adding to $D_0$ the constraint that for every transition fully enabled $\theta_1 \leq \theta_j$ yields a satisfiable constraints. We can show that setting $\theta_1=d_1$ allows to find a witness for satisfiability. Clearly, from the semantics, we have $\alpha(t_1) \leq \theta_1=d_1 \leq \beta(t_1)$, and as no $t_i$ fully enabled is more urgent than $t_1$, we can find $\theta_i$, $d_1 \leq \theta_i \leq \beta_i$, and hence $\theta_i$ satisfies $\theta_1\leq \theta_i$. As firability from $D_0$ holds, and $M_0 . N_0 \geq \preset{t_1}$, there exists $D_1$ that is a successor of $D_0$ such that $(M_0.N_0,D_0).(M_1.N_1,D_1)$ is a path of $SCG(\wnet)$.

    \item \textbf{Induction step: }\\
    Suppose that for every run $\rho$ of size at most $n$ of the waiting net, there exists a path $\pi$ of size $n$ in $SCG(\wnet)$ that coincides with $\rho$. We have to prove that it implies that a similar property holds for every run of size $n+1$.
Consider a run $\rho$ from $(M_0.N_0,v_0)$ to $(M_n.N_n,v_n)$ of size $n$ and the coinciding run $\pi$, of size $n$ too, from $(M_0.N_0,D_0)$ to $(M_n.N_n,D_n)$. Assume that some transition $t_f$ is firable after a delay $d$, i.e.,
 $(M_n.N_n,v_n) \overset{(d,t_f)}\longrightarrow(M_{n+1}.N_{n+1},v_{n+1})$. We just need to show that $t_f$ is firable from $C_n=(M_n.N_n,D_n)$.

As $t_f$ is firable from $(M_n.N_n,v_n)$ after $d$ time units, we necessarily have $\alpha(t_f) \leq v_n(t_f) + d \leq \beta(t_f)$, and for every $t_j$ fully enabled  $v_n(t_f) + d \leq \beta(t_j)$ (otherwise $t_j$ would become urgent).
For every transition $t_k$ enabled in $(M_n.N_n,v_n)$, let us denote by $r_k$ the index in the run where $t_k$ is last enabled in $\rho$. We hence have that $v_n(t_k) = min(\beta_i,\sum_{i=r_k+1..n} d_i)$.
In the abstract run $\pi$ we have $b_i = \beta_i -\underset{k\in r_i..n}\Sigma a_{k,f})$ where $a_{k,f}$ is the lower bound of domain $D_k$ for the time to fire of the transition used at step $k$ and symmetrically $a_i = \max (0, \alpha_i -\underset{k\in r_i..n}\Sigma b_{k,f})$.

At every step of path $\pi$, at step $k$ of the path, domain $D_k$ sets constraints on the possible time to fore of enabled transitions.
In every $D_k$, for every enabled transition $t_i$, we have $a_{k,i} \leq \theta_i \leq b_{k,i}$ for some values $a_{k,i} \leq b_{k,i}\leq \beta_i$, and additional constraints on the difference between waiting times. In particular, the set of transitions enabled in $M_n.N_n$ is not empty, and domain $D_n$ imposes that $a_{n,f} \leq \theta_f \leq b_{n,f}$.
Transition $t_f$ is firable from $D_n$ iff adding one constraint of the form $\theta_f \leq \theta_i$ per fully enabled transitions still yields a satisfiable domain. Now, we can show that the time spent in every configuration of $\wnet$ satisfies the constraints on  value for the time to fire allowed for the fired transition at every step.
\end{itemize}

\vspace*{-7mm}
\end{proof}

\begin{lemma}
\label{lem_durations_match_abstract_runs}
Let $\rho$ be a run of size $n$ an $\pi$ be an abstract run (of size $n$ too) that coincide. Then, for every transition $(M_i.N_i,v_i) \overset{d_i,t_f^i}\longrightarrow (M_{i+1}.N_{i+1},v_{i+1})$ we have $d_i \in [a_{i,f},b_{i,f}]$.
\end{lemma}

\begin{proof}
A transition $t_f^i$ can fire at date $\underset{j\in 1..i}\sum d_j$ iff $v_i(t_f) \in [\alpha_i,\beta_i]$. At every step $i$ of a run, we have $a_{i,j} = \max(0,\alpha_j - \underset{q\in Rj+1..i-1} \sum b_f^q)$ and $b_{i,j} = \max(0,\beta_j - \underset{q\in Rj+1..i-1} \sum a_f^q)$ where $b_f^q$ (resp. $a_f^q$) is the upper bound (resp. the lower bound) of the interval constraining the value of firing time for the transition fired at step $q$. At step $i$ if transition $t_j$ is newly enabled, then $a_{i+1,j} = \alpha_j$ and $b_{i+1,j} = \beta_j$. Otherwise, $a_{i+1,j} = max (0, a_{i,j}-b_f^i)$ and $b_{i+1,j} = max (0, b_{i,j}-a_f^i)$. At step $i$, if $t_f^i$ was newly enabled at step $i-1$ then we necessarily have $a_{i+1,j} = \alpha_j$, $b_{i+1,j} = \beta_j$ $v_i(t_f)=0$, and hence $d_{i+1} \in [a_{i+1,j},b_{i+1,j}]$. Now, let us assume that the property is met up to step $i$. if transition $t_f$ fires at step $i+1$, we necessarily have
$\alpha_f \leq v_i(t_f) +d_{i+1} \leq \beta_f$. This can be rewritten as $\alpha_f \leq \underset{q\in R_f+1..i}\sum d_q +d_{i+1} \leq \beta_f$. Considering step $i$, we have $a_{i,f^i} \leq d_i \leq b_{i,f^i}$,and hence $\alpha_f - b_{i,f^i} \leq \underset{q\in R_f+1..i-1}\sum d_q +d_{i+1} \leq \beta_f - a_{i,f^i}$. we can continue until we get $\alpha_f - \sum{q\in R_f+1..i}  b_{q,f^q} \leq d_{i+1} \leq \beta_f - \sum{q\in R_f+1..i}  a_{q,f^q}$, that is $d_{i+1} \in [a_{i+1,f},b_{i+1,f}]$.
\end{proof}

As one can wait $d$ time units in configuration $(M_n.N_n,v_n)$, it means that for every fully enabled transition $t_j$, $v_n(t_j)+d \leq \beta_j$. It now remains to show that setting $\theta_f=d$ still allows for values for remaining variables in $D_n$.
Setting $\theta_f =d$ and $\theta_f \leq \theta_i$ for every fully enabled transition amount to adding constraint $ d \leq \theta_i$ to $D_n$.
Further, we have $a_{n,f} \leq d \leq b_{n,f}$.
We can design a constraint graph for $D_n$, where nodes are of the form $\{\theta_i \mid t_i \in FullyEnabled(M_n.N_n)\} \cup\{ x_0\}$ where $x_0$ represents value $0$, and an edge from $\theta_i$ to $\theta_j$ has weight $w>0$ iff $\theta_i-\theta_j \leq w$. Conversely, a weight $w \leq 0$ represents the fact that  $w \leq \theta_i-\theta_j$. Similarly, and edge of positive weight $w$ from $\theta_i$ to $x_0$ represents constraint $\theta_i \leq w$ and an edge of negative weight $-w$ from $x_0$ to $\theta_i$ represents the fact that $w \leq \theta_i$. It is well known that a system of inequalities such as the constraints defining our firing domains are satisfiable iff there exists no negative cycle in its constraint graph. Let us assume that $D_n$ is satisfiable, but $D'_n = D_n \uplus \theta_f=0 \wedge \underset{t_i FullyEnabled}\bigwedge d \leq \theta_i$ is not. It means that $CG(D_n)$ has no cycle of negative weight, but $D'_n$ has one. Now, the major difference between $D'_n$ is that there exists an edge $\theta_f \overset{d}\longrightarrow x_0$, another one $x_0 \overset{-d} \longrightarrow \theta_f$, and an edge $x_0 \overset{-d} \longrightarrow \theta_i$ for every $t_i$ that is fully enabled.
Hence, new edges are only edges from/to $x_0$. If a negative cycle exists in $CG(D'_n)$, as $D_n$ is in normal form, this cycle is of size two or three. If it is of size two, it involves a pair of edges $\theta_j \overset{b_{n,j}}\longrightarrow x_0$ and $x_0 \overset{-d}\longrightarrow \theta_j$. However, following lemma~\ref{lem_durations_match_abstract_runs}, $d \leq b_{n,i}$ for every fully enabled transition $t_i$, so the weight of the cycle cannot be negative.
%% We may need a lemma to prove this
Let us now assume that we have a negative cycle of size three, i.e. a cycle involving $\theta_i, \theta_f$ and $x_0$, with edges $\theta_i \overset{c}\longrightarrow  \theta_f \overset{d}\longrightarrow x_0 \overset{-d}\longrightarrow \theta_i$. This cycle has a negative weight iff $c<0$. However, we know that $\theta_i \geq \theta_f$, this is hence a contradiction. Considering a cycle with a value $\theta_k$ instead of $\theta_f$ leads to a similar contradiction, and we need not consider cycles of size more than 3 because $D_n$ is in normal form, and hence the constraint graph labels each edge with the weight of the minimal path from a variable to the next one.

\medskip
Last, using lemma~\ref{lemma_postset_not_empty}, as $t_f$ is firable from $(M_n.N_n,D_n)$ there exists $D_{n+1}\in Post(C_n,t_f)$, and hence $\pi.(M_{n+1}.N_{n+1},D_{n+1})$ is a path of the state class graph of $\wnet$ that coincides with $\rho.(M_n.N_n,
\linebreak v_n) \overset{(d,t_f)}\longrightarrow(M_{n+1}.N_{n+1},v_{n+1})$\hfill$\Box$

\medskip\smallskip
\noindent{\bf Proposition~\ref{prop_SCG_sound}}
Let $\pi$ be a path of $SCG(\mathcal{W})$. Then there exists a run $\rho$ of $\wnet$ such that $\rho$ and $\pi$ coincide.

%\eject
\begin{proof}
Since $\rho$ and $\pi$ must coincide, if $\pi=(M_0.N_0,D_0) \overset{t_1}\longrightarrow (M_1.N_1,D_1) \dots (M_k.N_k,D_k)$, then $\rho=(M_0.N_0,v_0) \overset{(d_1,t_1)}\longrightarrow (M_1. N_1,v_1) \ldots (M_k. N_k,v_k)$. Since successive markings in both $\pi$ and $\rho$ are computed in the same way from presets and postsets of fired transitions (i.e. $M_i.N_i = M_{i-1}.N_{i-1} - ^\bullet(t_i) + (t_i)^\bullet $), we just have to show that for every abstract run $\pi$ of $SCG(\wnet)$, one can find a sequence of valuations $v_0,v_1,\dots v_k$ such that
$\rho=(M_0.N_0,v_0) \overset{(d_1,t_1)}\longrightarrow (M_1.N_1,v_1) \ldots (M_k . N_k,v_k)$ is a run of $\wnet$ and
such that firing $t_i$ after waiting $d_i$ time units is compatible with constraint $D_i$. We proceed by induction on the length of runs.

\medskip
\noindent\textbf{Base Case : }
    Let $\pi=(M_0.N_0,D_0) \xrightarrow[]{t_1} (M_1.N_1,D_1)$, where $D_0$ represents the firing domain of transitions from $M_0 . N_0$. We have $D_0 = \{ \alpha_i \leq \theta_i \leq \beta_i\mid t_i \in \mathsf{enabled}(M_0) \}$.
    Now, for $t_1$ to be a valid discrete move, there must be a timed move $d_1$ s.t. $(M_0 . N_0, v_0) \xrightarrow[]{(d_1,t_1)} (M_1 . N_1, v_1)$, which follows the condition :
  $$
  \begin{cases}
    M_0 . N_0 \geq {^\bullet}\!t_1 \\
    M_1 . N_1 = M_0 . N_0 - {^\bullet}\!t_1 + t_1^\bullet\\
    \alpha(t_1) \leq v_0(t_1) \oplus d_1 \leq \beta(t_1)\\
    \forall t, v_1(t) = 0 \textbf{ if }  \uparrow \mathsf{enabled}(t, M_1, t_1) \textbf{ else } v_0(t) \oplus d_1
  \end{cases}
  $$
  The conditions on markings are met, since $t_1$ is a transition from $M_0 . N_0$ to $M_1 . N_1$ in the $SCG$.\\
  Existence of a duration  $d_1$ is also guaranteed, since, adding $\theta_1 \leq \theta_i$ for every fully enabled transition $t_i$ to domain $D_0$ still allows firing $t_1$, i.e. finding a firing delay $\theta_1$. Hence choosing as value $d_1$ any witness for the existence of a value $\theta_1$ guarantees that no urgency is violated (valuation $v_0(t_i) +d_1$ is still smaller than $\beta_i$ for every fully enabled transition $t_i$). Let $\rho=(M_0.N_0,D_0) \xrightarrow[]{(d_1,t_1)} (M_1.N_1,v_1)$, where $v_1$ is obtained from $v_0$ by elapsing $d_1$ time units and then resetting clocks of transitions newly enabled by $t_1$. Then $\rho$ is compatible with $\pi$.
	
 \medskip\smallskip
\noindent\textbf{General Case : } Assume that for every path $\pi_n$ of the state class graph $SCG(\wnet)$ of size up to $n\in \mathbb N$,  there exists a run $\rho_n$ of $\wnet$ such that $\rho_n$ and $\pi_n$ coincide. We can now show that, given a path
$\pi_n=(M_0 . N_0,D_0) \overset{t_1}\longrightarrow (M_1 . N_1,D_1) \dots (M_n . N_n,D_n)$ and a run
$\rho_n = (M_0 . N_0,v_0) \overset{(d_1,t_1)}\longrightarrow (M_1 . N_1,v_1) \ldots (M_n . N_n,v_n)$
 such that $\rho_n$ and $\pi_n$ coincide,
and an additional move  $(M_n.N_n,D_n)\overset{t_f}\longrightarrow (M_{n+1}.N_{n+1},D_{n+1})$ via transition $t_f$,
we can build a run $\rho_{n+1}$ such that $\pi_{n+1}=\pi_n.(M_n.N_n,D_n)\overset{t_f}\longrightarrow (M_{n+1}.N_{n+1},D_{n+1})$ and $\rho_{n+1}$ coincide. We have $D_{n+1} = \mathsf{next}_r(D_n, t_f) \in \mathsf{Post}(D_n)$ for some bound $r$. As we want $\rho_{n+1}$ and $\pi_{n+1}$ to coincide, we necessarily have $\rho_{n+1}=\rho_n.(M_n.N_n,v_n)\overset{(d,t_f)}\longrightarrow (M_{n+1}.N_{n+1},v_{n+1})$, i.e. $M_{n+1}.N_{n+1}, t_f$ are fixed for both runs, $v_{n+1}$ is unique once $d$ and $t_f$ are set, so we just need to show that there exists a value for $d$ such that $\rho_{n+1}$ and $\pi_{n+1}$ coincide.

\medskip
From $(M_n.N_n,v_n)$, $d$ time units can elapse iff $v_n(t_i)+d \leq \beta_i$, for every transition $t_i$ fully enabled in $M_n.N_n$, and
$t_f$ can fire from 	$(M_n.N_n,v_n)$ iff $v_n(t_f)+d \in [\alpha_f,\beta_f]$. Let us denote by $r_i$ the index in $\pi_n,\rho_n$ where
transition $t_i$ is last newly enabled, and by $a_{i,t_j}$ (resp $b_{i,t_j}$) the lower bound (resp. upper bound) on value $\theta_j$ in domain $D_i$. Last let $t_f^i$ be the transition fired at step i

\medskip
We have that $v_n(t_i) =\min (\beta_i,\underset{j\in  r_i+1 ..n }\sum d_j)$,
$a_{n,j} = \alpha_j - \underset{q\in  r_j+1 ..n }\sum b_{q,t_f^q}$, and
$b_{n,j} = \beta_j - \underset{q\in  r_j+1 ..n }\sum a_{q,t_f^q}$

\medskip
If $(M_n.N_n,D_n)\overset{t_f}\longrightarrow (M_{n+1}.N_{n+1},D_{n+1})$ is enabled, then we necessarily have that
$D_n \cup \theta_f \leq \bigvee \theta_i$ is satisfiable. Let us choose $d=\theta_f=a_{n,j}$ and show that this fulfills the  constraints to fire $t_f$ in $\rho_{n+1}$

\medskip
$v_n(t_i)+d \leq \beta_i$ iff
$v_n(t_i)+ \alpha_j - \sum\underset{q\in  r_j+1 ..n } b_{q,t_f^q} \leq \beta_i$, for every fully enabled transition $t_i$.
iff $\underset{q\in  r_j+1 ..n }\sum d_q + \alpha_j - \underset{q\in  r_j+1 ..n }\sum b_{q,t_f^q} \leq \beta_i$
As we are looking for a run that coincides with $\pi_{n+1}$, we can assume wlog that we always choose the smallest value for $d_q$, i.e. we choose $d_q = a_{q,t_f^q}$. Hence, the inequality rewrites as $\underset{q\in  r_j+1 ..n }\sum a_{q,t_f^q}  + \alpha_j - \underset{q\in  r_j+1 ..n }\sum b_{q,t_f^q} \leq \beta_i$, or equivalently $\alpha_j - \underset{q\in  r_j+1 ..n }\sum b_{q,t_f^q} \leq \beta_i - \underset{q\in  r_j+1 ..n }\sum a_{q,t_f^q}$. This amount to proving $a_{n,j} \leq b_{n,i}$. We can do a similar transformation to transform $v_n(t_f)+d \in [\alpha_f,\beta_f]$ into an inequality $v_n(t_f)+ \alpha_j - \underset{q\in  r_j+1 ..n }\sum b_{q,t_f^q} \leq \beta_f$, then transformed into $a_{n,f} \leq b_{n,f}$, and   $\alpha_f \leq v_n(t_f)+d$ into  $\alpha_f \leq v_n(t_f)+ \alpha_j - \underset{q\in  r_j+1 ..n }\sum b_{q,t_f^q}$, and then $\alpha_f \leq \underset{q\in  r_j+1 ..n }\sum a_{q,t_f^q} + \alpha_f - \underset{q\in  r_j+1 ..n }\sum b_{q,t_f^q}$, which can be rewritten as $a_{q,t_f^q} \leq b_{q,t_f^q}$. As $D_n \cup \theta_f \leq \bigvee \theta_i$ is satisfiable, the conjunction of these inequalities holds too.
 \end{proof}

\begin{lemma}(Boundedness)
\label{lem-boundedness}
For all $i,j,k$ the constants $a_i, b_i$ and $c_{jk}$, of a domain of any state class graph have the following bounds :
\begin{align*}
    0 \leq \,& a_i \leq \alpha(t_i)\\
    0 \leq \,& b_i \leq \beta(t_i)\\
    -\alpha(t_k) \leq \,& c_{jk} \leq \beta(t_j)
\end{align*}
\end{lemma}

\begin{proof}
First of all, every variable $\theta_i$ represents minimal and maximal times to upper bounds of interval, so by definition it can only be a positive value. We hence have $0 \leq a_i$, $0 \leq b_i$.
Now to prove $a_i \leq \alpha(t_i)$ and $b_i \leq \beta(t_i)$ always hold, we will study the effect of every step to compute $\mathsf{next}_r(D,t_f)$.

\medskip
Let us recall how $\mathsf{next}_r(D,t_f)$ is built. We first add $\theta_f \leq \theta_i$ to $D$ for every fully enabled transition $t_i$, and the inequality $b_r \leq \theta_f \leq b_{r+1}$.
We then do a variable substitution as follows. We write:
          \[
          \theta_j :=
          \begin{cases}
          b_j + \theta'_j \textbf{ if } b_j\leq b_r \textbf{ and } \mathsf{enabled}(M)\setminus \mathsf{FullyEnabled}(M.N)                           \\
          \theta_f + \theta'_j \textbf{ if } b_j >  b_r \textbf{ and } \mathsf{FullyEnabled}(M)\\
					0 \textbf{ otherwise}
          \end{cases}
          \]
					
After variable substitution we have inequalities of the form
$a_i \leq \theta_i' + b_i \leq b_i$, $a_i \leq \theta_i' + \theta_f \leq b_i$,
$b_r \leq \theta_f \leq b_{r+1}$, $\theta_f \leq \theta_i'$,
$\theta_i'  - \theta_j'\leq c_{ij}$ if $t_i, t_j$ are both fully enabled,
$a_i \leq \theta_i' + \theta_f \leq b_i$, and $\theta_i=0$ for every enabled transition $t_i$ reaching its upper bound $b_i$

\medskip
We use Fourier-Motzkin elimination to remove variable $\theta_f$. This elimination makes new positive values of the form $a'_j= \max(0,a_j -b_i)$  or $b'_j=\max(0, b_j-a_i)$ appear (See also Lemma~\ref{lem-li-FME}). Yet, we still have $a'_j\leq \alpha_j$ and $b'_j \leq \beta_j$.

%\medskip
Then addition new constraints for newly enabled transitions do not change existing constraints, and for every newly enabled transition $t_i$, we have $\alpha_i \leq \theta_i \leq \beta_i$. The last step consist in computing a canonical form. Remember that canonical forms consist in computing a shortest path in a graph. Hence $D^*$ also preserves boundedness. (See also Lemma~\ref{lem-li-FME}) Now, in the canonical form, we can consider bounds for $\theta_j - \theta_k$, knowing that both values are positive.  $-a'_k \leq \theta_j - \theta_k \leq b'_j$, and hence $-\alpha(t_k) \leq c_{jk} \leq \beta(t_j)$.
\end{proof}

\begin{definition}(linearity)
\label{def-bounded-linear}
Let $\mathcal{W}$n be a waiting net, and let $K_\wnet = \max_{i,j} \lfloor \frac{\beta_i}{\alpha_j} \rfloor$
A domain $D$ is \textit{linear} (w.r.t. waiting net $\mathcal{W}$) if, for every constraint in $D$,
lower and upper bounds $a_i,b_i$ of constraints of the form  $a_i \leq \theta_i \leq b_i$ and upper bounds $c_{i,j}$ of difference constraints of the form $\theta_i -\theta_j \leq c_{i,j}$ are linear combination of $\alpha_i$'s and $\beta_i$'s with integral coefficients in $[-K_\wnet, K_\wnet]$.
\end{definition}

Obviously, the starting domain $D_0$ of a waiting net $\wnet$ is bounded and linear. We can now show that the successor domains reached
when firing a particular transition from any bounded and linear domain are also bounded and linear.

\begin{lemma}
\label{lem-li-FME}
Elimination of a variable $\theta_i$ from a firing domain of a waiting net preserves boundedness and linearity.
\end{lemma}

\begin{proof}
Fourier-Motzkin elimination proceeds by reorganization of a domain $D$, followed by an elimination, and then pairwise combination of expressions (see complete definition of Fourier Motzkin elimination in appendix~\ref{appendix-Fourier}). We can prove that each of these steps produces inequalities that are both linear and bounded.

\medskip
Let $D$ be a firing domain, and $D'$ the domain obtained after choosing the fired transition $t_f$ and the corresponding variable substitution. An expression in $D'$ of the form $\theta_f - \theta_i \leq c_{f,i}$ can be rewritten as $\theta_f  \leq c_{f,i} + \theta_i$. An expression of the form $\theta_i - \theta_f \leq c_{i,f}$ can be rewritten as $\theta_i - c_{i,f} \leq \theta_f$.
We can rewrite all inequalities containing $\theta_f$ in such a way that they are always of the form $exp \leq \theta_f$ or $\theta_f\leq exp$. Then, we can separate inequalities in three sets :
\begin{itemize}
\itemsep=0.9pt
\item $D^+$, that contains inequalities of the form $exp^- \leq \theta_f$, where $exp^-$ is either constant $a_f$ or an expression of the form $\theta_i - c_{i,f}$. Let $E^-$ denote expression appearing in inequalities of this form.
\item $D^-$, that contains inequalities of the form $\theta_f \leq exp^+$, where $exp^+$ is either constant $b_f$ or an expression of the form $\theta_i + c_{f,i}$. Let $E^+$ denote expression appearing in inequalities of this form.
\item $D^{\bar{\theta_f}}$ that contains all other inequalities.
\end{itemize}

The next step is to rewrite $D$ into an equivalent system of the form $D^{\bar{\theta_f}} \cup \max(E^-) \leq \theta_f \leq \min(E^+)$,and then eliminate $\theta_f$ to obtain a system of the form $D^{\bar{\theta_f}} \cup \max(E^-) \leq \min(E^+)$. This system can then be rewritten as
$D^{\bar{\theta_f}} \cup \{ exp^- \leq exp^+ \mid exp^- \in E^- \wedge exp^+ \in E^+\}$. One can easily see that in this new system, new constants appearing are obtained by addition or substraction of constants in $D$, and hence the obtained domain is still linear.

%\medskip
At this point, nothing guarantees that the obtained domain is bounded by larger $\alpha$'s and $\beta's$. Let us assume that in $D$, we have $0 \leq a_i \leq \alpha_i$, $0 \leq b_i \leq \beta_i$ and $-\alpha_k \leq c_{j,k} \leq \beta_j$. Then the last step of FME can double the maximal constants appearing in $D$ (for instance when obtaining $\theta_j - c_{j,f} \leq \theta_i + c_{f,i}$ or its equivalent $\theta_j - \theta_i \leq  + c_{f,i} +c_{j,f}$. However, values of $a_i$'s and $b_i's$ can only decrease, which, after normalization, guarantees boundedness of $\theta_j - \theta_i$.
\end{proof}

\begin{lemma}
\label{lem-red-can-li}
Reduction to canonical form preserves linearity.
\end{lemma}

\begin{proof}
It is well known that computing a canonical form from a domain $D$ represented by a DBM $Z_D$ amounts to computing the shortest path in a graph representing the constraints. Indeed, a DBM is in canonical form iff, for every pair of indexes $0\leq i,j \leq |T|$, and for every index $0 \leq k \leq |T|$ we have $Z(i,j) \leq Z(i,k) + Z(k,j)$. The Floyd Warshall algorithm computes iteratively updates of shortest distances by executing instructions of the form $Z(i,j) := \min(Z(i,j), Z(i,k) + Z(k,j))$. Hence, after each update, if $Z(i,j)$ is a linear combination of $\alpha's$ and $\beta's$, it remains a linear combination.
\end{proof}

\begin{lemma}
\label{lem-red-can-bound-li}
Fourier Motzkin elimination followed by reduction to canonical form preserves boundedness and linearity.
\end{lemma}
\begin{proof}
From Lemma~\ref{lem-li-FME} and Lemma~\ref{lem-red-can-li}, we know that domains generated by FME + canonical reduction are linear.
However, after FME, the domain can contain inequalities of the form $a'_i \leq \theta_i \leq b'_i$ with $a'_i \leq a_i \leq \alpha_i$ and $b'_i \leq b_i \leq \beta_i$. However, it may also contain inequalities of the form $ x \leq \theta_i -\theta_j \leq y$ where $-2\cdot \max(\alpha_i) \leq x$ and $y \leq \cdot \max(\beta_i)$. Now, using the bounds on values of $\theta_i's$, the canonical form calculus will infer $ a'_i - b_j \leq \theta_i -\theta_j \leq b'_i - a'_j$, and we will have $ -\alpha_i \leq \theta_i -\theta_j \leq \beta_j$.
\end{proof}

\begin{lemma}(Bounded Linearity)
\label{lem-linearity}
For all $i,j,k$ the constants $a_i, b_i$ and $c_{jk}$, of a domain of any state class graph are linear in $\alpha$'s and $\beta$'s
\end{lemma}

\begin{proof}
Clearly, the constraints in $D_0$ are linear in $\alpha$'s and $\beta$'s see Definition (\ref{def-stateclassgraph}), it remains to prove that, if $D$ is bounded and linear, then for every fired transition $t$ and chosen time bound $r$,  $\mathsf{next}_k(D_r,t)$ is still bounded and linear.
We already know that Fourier Motzkin Elimination, followed by canonical form reduction preserves boundedness and linearity (Lemma \ref{lem-red-can-bound-li}). Addition of new constraints do not change constants of existing constraints, and the constants of new constraints (of the form $\alpha_i \leq \theta_i \leq  \beta_i$ are already linear. Further, these constraints are completely disjoint from the rest of the domain (there is no constraint of the form $\theta_k - \theta_i \leq c_{k,i}$ for a newly enabled transition $t_i$). Hence computing a canonical form before or after inserting these variables does not change the canonical domain. So, computing $D^*$ after new constraints insertion preserves linearity, and
thus, the constants appearing in the constraints in domain $\mathsf{next}_k(D_r,t)$ are bounded and linear w.r.t. $\alpha$'s and $\beta$'s.
\end{proof}

\noindent{\bf Proposition~\ref{prop-fin-domain}}
The set of firing domains in $SCG(\mathcal{W})$ is finite.

\begin{proof}
We know that a domain in a $SCG$ is of form :
$$
\begin{cases}
a_i^* \leq \theta_i \leq b_{i}^*\\
\theta_j - \theta_k \leq c_{jk}^*.
\end{cases}
$$
By boundedness (Lemma \ref{lem-boundedness}) we have proved that the constants $a^{*'}s , b^{*'}s$ and $c^{*'}s $ are bounded above and below by $\alpha's$ and $\beta's$ up to sign, we have also proved that they are linear combinations of $\alpha's$ and $\beta's$ (Lemma~\ref{lem-linearity}). Now, it remains to show that there can  only be finitely many such linear combinations, which was shown in~\cite{BerthomieuD91}.
Hence, the set of domains appearing in $SCG(\wnet)$ is finite.
\end{proof}

\end{document}